\documentclass[a4paper,11pt]{article} 
\usepackage[a4paper,hmargin=2.8cm,vmargin=3cm]{geometry}
\usepackage[utf8]{inputenc} 
\usepackage[english]{babel}

\usepackage{mathtools,amsthm, amsfonts, amssymb, mathrsfs, bbm, bm}
\usepackage{authblk}
\usepackage{lineno,csquotes}
\usepackage[dvipsnames]{xcolor}
\usepackage{hyperref}
\usepackage{comment}
\usepackage{upgreek}
\usepackage{amsmath}
\usepackage{dsfont}
\usepackage{mathrsfs}
\usepackage{slashed}
\usepackage{cancel}

\usepackage[
    backend=biber,
    style= numeric,
    sorting=none
  ]{biblatex}
\definecolor{BeauBlue}{rgb}{0, 0.2, .9}
\definecolor{BeauOrange}{rgb}{.8, .1, 0}
\usepackage{hyperref}
\hypersetup{
	colorlinks = true,
	linkcolor = BeauBlue,
	urlcolor = cyan,
	citecolor = BeauOrange,
}
 \usepackage[mathcal]{euscript}

\numberwithin{equation}{section}

\theoremstyle{plain}
\newtheorem{theorem}{Theorem}[section] 
\newtheorem*{theorem*}{Theorem} 
\newtheorem{proposition}[theorem]{Proposition} 
\newtheorem{corollary}[theorem]{Corollary}
\newtheorem{lemma}[theorem]{Lemma}

\theoremstyle{definition}
\newtheorem{definition}[theorem]{Definition} 
\newtheorem{assumption}{Assumption}

\theoremstyle{remark}

\newtheorem{remark}[theorem]{Remark}

\def\a{\ensuremath\alpha}
\def\b{\ensuremath\beta}
\def\d{\ensuremath\delta}
\def\v{\ensuremath\vec}
\def\t{\ensuremath\tau}
\def\th{\ensuremath\theta}

\def\o{\ensuremath\omega}
\def\r{\ensuremath\rho}
\def\s{\ensuremath\sigma}
\def\e{\ensuremath\epsilon}

\def\de{\ensuremath\partial}

\def\mc{\ensuremath\mathcal}
\def\l{\ensuremath\lambda}
\def\L{\ensuremath\Lambda}

\def\bs{\ensuremath\boldsymbol}

\def\dg{\ensuremath\dagger}
\def\it{\textit}
\def\mf{\ensuremath\mathfrak}

\def\mbb{\ensuremath\mathbb}

\newcommand{\dist}[1]{\ensuremath |\!\!|#1|\!\!|}

\newcommand{\id}{\mathds{1}} 

\newcommand{\bra}[1]{\big\langle #1 \big|}
\newcommand{\ket}[1]{\big| #1 \big\rangle}

\newcommand{\R}{\mathcal{R}}

\newcommand{\virg}[1]{``#1''}

\DeclareMathOperator{\Span}{Span} 
\newcommand{\ie}{{\sl i.\,e.\ }}   
\newcommand{\eg}{{\sl e.\,g.\ }} 

\DeclareMathOperator{\Tr}{Tr}

\usepackage{tikz}
\usetikzlibrary{shapes}
\usetikzlibrary{positioning}
\tikzset{vertex/.style={circle,fill=black,inner sep=2pt},
	ctVertex/.style={diamond,fill=black,inner sep=2pt},
    CtVertex/.style={diamond,fill=black,inner sep=3pt},
	sqVertex/.style={rectangle,fill=black,inner sep=3pt},
	bigVertex/.style={circle,fill=black,inner sep=4pt},
	E/.append style={fill=white,draw},
    F/.append style={fill=gray,draw},
	probeEP/.style={circle,fill=black,draw,inner sep=2pt,
		prefix after command= {\pgfextra{\tikzset{every pin/.style = {pin edge={decorate,decoration={snake,amplitude=2pt,segment length =4pt}}}}}}
	},
	bareProbeEP/.style={rectangle,fill=black,draw,inner sep=3pt,
		prefix after command= {\pgfextra{\tikzset{every pin/.style = {pin edge={decorate,decoration={snake,amplitude=2pt,segment length =4pt}}}}}}
	},
	nuEP/.style={circle,fill=white,draw, inner sep=2pt},
	linelabel/.style={sloped,above,very near start, inner sep=1pt,execute at begin node=$\scriptstyle,execute at end node=$},
	baseline=(current  bounding  box.center),doubled/.style={double distance= 1pt,line width=1.5pt}
}

\title{The conductivity matrix at the topological phase transition}

\author[1]{Simone Fabbri}
\author[2]{Giovanna Marcelli}
\author[2]{Alessandro Giuliani}

\affil[1]{SISSA, Via Bonomea 265, 34136 Trieste, Italy}
\affil[2]{Roma Tre University, Largo S. Leonardo Murialdo 1, 00146 Rome, Italy}

\date{\today}

\begin{document}

\maketitle

\begin{abstract}
For a general class of two-dimensional non-interacting semi-metallic electron systems on a lattice, we derive an explicit formula for the whole conductivity matrix, by using Euclidean many-body formalism. The semi-metallic phase takes place whenever two Bloch bands touch at the Fermi energy with conical intersections and generally occurs at the transition between distinct integer quantum Hall phases.

Unlike the longitudinal conductivity (which depends solely on the shape of the cones at the Fermi level), the transverse conductivity is remarkably determined both by the conical structures of the bands and by the nature of the two nearby topological phases.

Generically, in the semi-metallic state, 
neither the longitudinal nor the transverse conductivities are quantized in integer multiples of a universal conductivity quantum; however, 
universality of the conductivity matrix is restored
under the assumption of emergent 
rotational symmetry of the linearized Hamiltonian at the Fermi points.

Using our formula, we compute the conductivity 
matrix for several physically relevant models 
of quantum Hall fluids at the topological phase transition, exhibiting cases where the transverse conductivity is quantized in half-integer multiples of $e^2/h$ and cases where it depends continuously on an external strain parameter, thus shedding light on its universality or non-universality features.  
\end{abstract}

\tableofcontents

\section{Introduction}

The Integer Quantum Hall Effect (IQHE) is an emblematic phenomenon of universal transport property in condensed matter physics, see \cite{Graf2007} and references therein. Universality stems from the robustness of quantization of the Hall conductivity $\s_{12}$ under continuous microscopic modifications of the system, which preserve the insulating condition. Any insulating phase is characterized by the corresponding integer value  of $\s_{12}$ (integer in units of $\frac{e^2}{h}$, where $e$ is the electron charge and $h$ is the Planck constant).
A great deal of works in mathematical physics has been devoted to its rigorous understanding
since its discovery \cite{Klitzing1980}, involving several different approaches. After the pioneering papers \cite{Laughlin1981,TKNN, Halperin1982, Avronetal83},  several aspects of IQHE have been rigorously investigated. The quantization of the Hall conductivity has been established, first relying on non-commutative geometry \cite{BES94} and index method \cite{Avronetal94}, and then extended also to non-covariant systems \cite{Elgart2005, Marcelli23} and in the regime of strong disorder \cite{AG98}. The robustness of the quantization of $\s_{12}$ under weak interactions has been shown in \cite{HastingsMichalakis2015, GMP17} and more recently, including also the effect of a continuous magnetic flux for Hofstadter-like models, in \cite{Marcellietal2026}. Regarding the validity of the Kubo formula from quantum dynamics, it has been proved for gapped models in the continuum \cite{ASY87,ElgartSchlein04,Marcelli22}, and for interacting systems on a lattice  by using adiabatic theorems \cite{BachmannDeRoeckFraas2018, MonacoTeufel2019} and non-equilibrium almost-stationary states \cite{Teufel2020, HenheikTeufel2022}. By using rigorous renormalization group methods, the validity of the Kubo formula has been obtained also at small positive temperature \cite{GLMP24} for gapped fermionic systems and for one-dimensional gapless fermionic systems \cite{PS25, PSS25}.  
Actually, it has been shown that the response of the current to the applied transverse electric field is nearly linear \cite{BachmannDeRoeckFraasLange2021, MarcelliMonaco2022, Wesleetal2025}. 

The previous works provide a satisfactory understanding of insulating Integer Quantum Hall
(IQH) phases for periodic, disordered or weakly interacting electrons. A big open problem is to characterize the transition between two distinct 
IQH phases, and possibly to establish the universality of the corresponding Topological Phase Transition (TPT). The problem can be studied from different perspectives; \eg, in \cite{MCC19} the TPT is characterized by the divergent behavior of the topological curvature functions via renormalization-group methodology, while in \cite{Marrazzo25} the critical exponents 
of a disorder-driven TPT are computed in terms of 
space-resolved topological markers. In the 
current work we adopt the same viewpoint
as in \cite{GJMP16,GMP20,FGR25}, and we study the 
TPT of a class of lattice quantum Hall states in  
terms of the critical Kubo conductivity matrix 
$\s$, computed exactly at the transition. 
In previous works \cite{GJMP16,GMP20,FGR25}, 
we specifically considered the example of the 
weakly-interacting Haldane model: we determined 
the location of the critical lines separating distinct IQH phases and proved universality of the longitudinal and transverse conductivities both away and on the critical lines. However, 
the specific setting we considered did not allow us to fully address the universal nature of the critical conductivity, thus leaving us with a number of unanswered questions: how general and robust is the prediction that the longitudinal and transverse critical conductivities are quantized in integer multiples of a fundamental unit, as we change the model under investigation and continuously modify its parameters, such as the hopping strength and magnetic field? In particular, should we always expect the transverse conductivity at the TPT to be 
proportional to a half-integer multiple of $e^2/h$, as in the case of the Haldane model? If not, can one exhibit physically relevant examples 
of tight binding models with arbitrary half-integer values of the transverse conductivity?

In order to attack these problems, we consider a general class of $2D$ non-interacting electron systems in the tight-binding approximation, modeled by an exponentially decaying, translation-invariant, Hamiltonian kernel, with entries labeled by $M\ge 2$ internal degrees of freedom. 
For simplicity, in this work we neither include the effect of disorder nor of interactions, but we still adopt a second-quantized formalism and use the Euclidean framework, which is the natural setting for proving, in perspective, stability of the quantization of the Kubo conductivity matrix under weak short range interactions, as in \cite{GMP12,GMP17,GMP20}, to be addressed in future work.

As widely discussed in the context of 
Weyl and Dirac semi-metals, see \eg 
\cite{AMV18,Bu18}, and as rigorously established in \cite{D21}, a TPT generically corresponds to a semi-metallic behavior of the system, where the two closest bands to the Fermi energy
touch with conical intersections at a finite number of Fermi points. Typically, the TPT separates two distinct gapped IQH phases, with gap proportional to a ``mass parameter" $m$, which controls the topological nature of the phase. With this picture in mind, we model the system in terms of a hopping Hamiltonian $H^0$ smoothly depending on a parameter $m$ controlling the TPT: while the system 
at $m>0$ and $m<0$ is insulating and generically belonging to two distinct topological phases, 
its behavior at $m=0$ is semi-metallic. In agreement with \cite[Eq.(1.2)]{D21}, we assume that at each Fermi point $\v{k}_{F,\omega}$, where $\omega\in\{1,\ldots,N\}$ is an index labelling the Fermi points, the two energy bands closest to the Fermi energy $\mu$, denoted by $\lambda_\pm(\vec k;m)$, satisfy: 
\begin{equation}
\label{eqn:intro1}
\l_{\pm}(\v{k}_{F,\omega}+ \v{k}';m)- \mu= \v{c}_\omega\cdot\v{k}'+ u_\omega m \pm \sqrt{|A_\omega\v{k}'|^2+ w_\omega^2m^2+\v{k}'\!\cdot \v{b}_\omega m} + \mc{O}(|\v{k}'|^2+m^2), 
\end{equation}
asymptotically as $|\v{k}'|,m\to 0$, 
for suitable $\v{c}_\omega,\v{b}_\omega\in\mbb{R}^2, u_\omega, w_\omega\in\mbb{R}$ and an invertible $2\times 2$ real matrix $A_\omega$.
In addition, we require that the \emph{conical} term dominates the \emph{affine} one in the r.h.s.\ of Eq.\eqref{eqn:intro1}, namely that there exists a positive constant $C_{\star}$ such that
\begin{equation}
\label{eqn:intro2}
\sqrt{|A_\omega\v{k}'|^2+ w^2_\omega m^2 +\v{k}'  \!\cdot \v{b}_\omega m} -|\v{c}_\omega\cdot\v{k}'+ u_\omega m| \geq C_{\star}(|\v{k}'|+|m|)\qquad\text{for all $\v k'\in\mbb{R}^2, m\in\mbb{R}$},
\end{equation}
for every $\omega\in\{1,\ldots,N\}$. Note that for $m=0$ these two hypotheses (Eqs.\eqref{eqn:intro1}-\eqref{eqn:intro2}) reduce to those assumed in \cite{MPP26}, where the mass parameter $m$ is fixed to be zero.

In Section \ref{sect:examples} below we discuss different examples of physically relevant models 
displaying a TPT and satisfying the assumptions above in the vicinity thereof. The examples we consider are the Hofstadter model at different rational values of the magnetic flux, and a deformation of the standard Haldane model, including a straining term. 

In our setting, given the $m$-dependent hopping Hamiltonian, we define the mass-dependent conductivity matrix $\s(m)$ via the Kubo formula, see Eq.\eqref{def:sigma} below, rewritten in the Euclidean formalism after Wick's rotation. Here, we do not discuss the validity of the Kubo formula, which, however, can be deduced by combining the strategy of \cite[Theorem 2.7]{MPP26} for the massless case and \cite[Theorem 3.7]{GLMP24} for the massive one. Informally, our main result can be stated as follows, see Section \ref{sect:model} below for a precise statement of the assumptions and results: for every $i,j\in\{1,2\}$ the \emph{critical} conductivity, namely $\s_{ij}(m)$ at $m=0$ is given by:
\begin{equation}
\label{eqn:intro3}
\s_{ij}(0)=\frac{1}{2}\lim_{\e\rightarrow0^+} \Big(  \s_{ij}(\e)+ \s_{ij}(-\e) \Big)+ \frac{1}{16}\sum_{\o=1}^N \frac{(|A_{\o}|^2)_{ij}}{\det |A_{\o}|},
\end{equation}
where $A_\omega$ is the same matrix as in \eqref{eqn:intro1}-\eqref{eqn:intro2}. In particular, since the massive longitudinal conductivity $\s_{ii}(\pm\e)=0$, the critical longitudinal conductivity reads as:
\begin{equation}
\label{eqn:intro4}
\s_{ii}(0)=\frac{1}{16}\sum_{\o=1}^N \frac{(|A_{\o}|^2)_{ii}}{\det |A_{\o}|},
\end{equation}
which agrees with the one in \cite[Theorem 2.4]{MPP26}. Hence, the value of $\s_{ii}(0)$ does depend on the shape of the cones in quasi-momentum $\v k'$ defining the conical intersections of the Bloch bands at the Fermi energy, and so in general neither it is  quantized nor it corresponds to a topological invariant; see the case of the strained Haldane model in Subsection \ref{ssect:strain} below. On the other hand, for the critical Hall conductivity we have that 
\begin{equation}
\label{eqn:intro5}
\s_{12}(0)=\frac{1}{2}\lim_{\e\rightarrow0^+} \Big(  \s_{12}(\e)+ \s_{12}(-\e) \Big)+ \frac{1}{16}\sum_{\o=1}^N \frac{(|A_{\o}|^2)_{12}}{\det |A_{\o}|},
\end{equation}
which generalizes in the non-interacting setting the result obtained in \cite[Theorem 2.1]{FGR25} for the Haldane model: in that case $N$ is either 
$1$ or $2$, and the matrix $A_{\o}$ is proportional to the identity, so that  $(|A_{\o}|^2)_{12}$ vanishes and $\sigma_{12}(0)$ reduces to the average of the transverse conductivities in the two contiguous topological insulating phases. However, in general, whenever the second term in the right hand side of \eqref{eqn:intro5} does not vanish, the transverse conductivity is neither quantized nor a topological invariant, see again the strained Haldane model in Subsection \ref{ssect:strain}. Still, the anti-symmetric part of the critical conductivity matrix \it{is}
quantized at a half-integer value of the fundamental conductivity quantum, see Remark \ref{rmk:main}.\ref{it:thm3} below. 

Even though in general $\sigma(0)$ is not quantized, it becomes so under the additional hypothesis of \emph{emergent} discrete rotational invariance, see Assumption \ref{assum:R} below. In this case we get:
\begin{equation*}
\s_{ij}(0)= \frac{1}{2}\lim_{\e\to 0^+} \Big( \s_{ij}(\e)+ \s_{ij}(-\e) \Big)+ \frac{N}{16}\d_{i,j},
\end{equation*}
where $N$ is the number of Fermi points. Specifically, the longitudinal conductivity $\s_{ii}(0)= \frac{N}{16}$ and the transversal one $\s_{12}(0)=-\s_{21}(0)$ is given by arithmetic average of $\s_{12}(\pm\e)$.

In Section \ref{sect:examples}, we apply these formulas to several examples, and exhibit cases in which the transverse conductivity is quantized at different half-integer multiples of the fundamental conductivity quantum (in our examples, we either find that $2\pi \sigma_{12}(0)=\pm 1/2$, the same value found in the Haldane model \cite{FGR25}, or $2\pi \sigma_{12}(0)=\pm 3/2$; clearly one can generalize our examples below to obtain a generic half-integer value), and cases in which it is not quantized and, rather, it varies continuously with an external control parameter (the strain, in our example). 

\medskip

Our result, specifically formula \eqref{eqn:intro4}, is in agreement with \cite{MPP26}, where, however, the analysis was limited to the longitudinal conductivity only, 
which can in fact be studied by restricting the analysis to the critical case $m=0$ (no need to introduce the 
auxiliary mass parameter $m$ in that case). Moreover, differently from the latter paper, here we choose to work in a many-body, Euclidean, formalism, which is suited for including electron interactions, as we plan to address in the future. Note also that, by assuming $A_{\o} \propto \mathds{1}_2$ and $\v{c}_{\o}=\v{0}$ in Eq.\eqref{eqn:intro1}  (isotropic Dirac cones), our result, Eq.\eqref{eqn:intro4}, reproduces the one in \cite{Cances}, where, however, time-reversal symmetry was also assumed (so that $\s_{12}(0)=0=\s_{21}(0)$ automatically). 

The strategy of proof of our main theorem incorporates the idea from \cite{FGR25} of considering the \emph{symmetric difference} of the conductivity matrix $\mf{D}\s_{ij}(\e):= \s_{ij}(0)- \frac{1}{2}\s_{ij}(\e)- \frac{1}{2}\s_{ij}(-\e)$ and to show that it is dominated by contributions from quasi-momenta and Matsubara frequencies close to the Fermi points and zero respectively. Moreover, it combines the following ingredients: first, we use the Kato-Nagy unitary operator \cite[Chapter I, Section 4.6]{Kato} (see also \eg \cite{MP14}) to reduce the analysis to the semi-metallic two-dimensional sub-space corresponding to the two critical bands; next, we show that the value of the conductivity 
can be computed in terms of an effective \emph{relativistic model}, obtained by linearizing the energy bands around each Fermi point; finally, we perform the explicit computation of the symmetric difference of the 
conductivity in the relativistic model via a linear mapping, which transforms the general effective relativistic Hamiltonian into the Dirac Hamiltonian (see the beginning of Section \ref{sect:proof} for further details and comparison with the existing literature).

\paragraph{Outlook.}
As already mentioned, a natural generalization 
of this work is to include the effect of weak 
short-range electron-electron interactions, which we expect to be possible by combining the analysis of this work with constructive fermionic Renormalization Group methods, in the spirit of \cite{GJMP16,GMP20,FGR25}. The inclusion of disorder in a semi-metallic setting is more challenging: the only constructive result we are aware of is \cite{AFP20}, which concerns a 
hierarchical approximation of $3D$ semi-metals. 
Actually, it would already be interesting to 
investigate universality (or lack of universality) in the presence of weak disorder, 
at all orders in renormalized perturbation theory, in the spirit of \cite{Ma13}.

\paragraph{Organization of the paper.} The paper is structured as follows. Section \ref{sect:model} is dedicated to the mathematical set-up and presentation of the two main results, namely Theorem \ref{thm:main} and Corollary \ref{cor:rotation}. Section \ref{sect:examples} illustrates a series of physically relevant models, where the computation of $\s(0)$ is explicitly performed. The last section consists of the proof of Theorem \ref{thm:main}. Appendix \ref{app:rot} provides the proof of Corollary \ref{cor:rotation} and the fact that an \emph{exact} discrete rotation symmetry is also \emph{emergent}. Finally, Appendix \ref{app:tech} collects several technical results, which are propaedeutic to the proof of the main result.

\section{The model and the main results}
\label{sect:model}

Let $\Lambda$ be a two-dimensional Bravais lattice, \ie $\Lambda:=\Span_{\mbb{Z}}\{\vec{a}_1,\vec{a}_2\}$ where $\v{a}_1,\v{a}_2$ are linearly independent vectors in $\mbb{R}^2$. For $L\in\mbb{N}$, with $L$ positive, one introduces the lattice dilated by $L$ as $L\Lambda:=\Span_{L\mbb{Z}}\{\vec{a}_1,\vec{a}_2\}$. The \emph{discrete torus} of side $L$ is defined as $\Lambda_L:=\Lambda/(L\Lambda)$.
We allow the presence of a finite number of sublattice sites within the unit cell, indexed by vectors $\vec{\d}_1,...,\v{\d}_M$, so that each point $\v{y}$ of the total configuration space is uniquely specified by $\v{y}=\v{x}+\v{\d}_{\r}$, with $\v{x}\in\L_L$ and $\r\in \{1,...,M\}$. We introduce the standard distance on the discrete torus $\L_L$:
\[
\dist{\v{x}-\v{y}}_{\L_L}:= \min_{n_1,n_2\in\mbb{Z}} |\v{x}-\v{y}+ Ln_1\v{a}_1+ Ln_2 \v{a}_2|,
\]
where $|\,\cdot\,|$ denotes the Euclidean norm in $\mbb{R}^2$.

\medskip

We denote by $\mc{F}_L$ the fermionic Fock space over $\L_L\times I$, endowed with the annihilation and creation operators $a_{\v{x},\r}, a^{\dg}_{\v{x},\r}$,\footnote{In this work we use the symbol $\dg$ to denote the Hermitian conjugate, while we use the symbol $*$ to denote the complex conjugate.} meant with periodic boundary conditions, satisfying the canonical anti-commutation relations, namely for every $\vec{x},\vec{y}\in \Lambda_L$ and $\rho,\rho'\in \{1,\dots,M\}$:
\begin{equation} 
\label{eq_1}
\big\{ a_{\v{x},\r}, a_{\v{y},\r'} \big\}= \big\{ a^{\dg}_{\v{x},\r}, a^{\dg}_{\v{y},\r'} \big\}=0, \;\; \big\{ a_{\v{x},\r}, a^{\dg}_{\v{y},\r'} \big\}= \d_{\v{x},\v{y}} \d_{\r,\r'}. \end{equation}

In the following, we denote by $\mathrm{Mat}(M,\mbb{C})$ the set of $M\times M$ matrices with complex entries.
\begin{assumption}
\label{assum:H0}
Let the non-interacting Hamiltonian kernel $ 
H^0\colon \L\to \mathrm{Mat}(M,\mbb{C})$ be \emph{exponentially decaying}, namely there exists $C,\xi>0$ such that
$\|H^{0}(\v{x})\|\le C e^{-\xi|\v{x}|}$, where the norm of the l.h.s.\ is the operator norm. We assume that $H^0$ is Hermitian, in the sense that 
\begin{equation} \big[H^0(\vec x)\big]^\dagger=H^0(-\vec x).\end{equation}
\end{assumption}

We define its finite-volume counterpart as its periodization:
\begin{equation*}
H^{0,L}(\v{x}):= \sum_{n_1,n_2\in\mbb{Z}} H^0(\v{x}+ n_1 L\v{a}_1+ n_2 L\v{a}_2).
\end{equation*}
The associated finite-volume Hamiltonian reads:
\begin{equation}
\label{eq_2}
\mc{H}_L= \sum_{\substack{\v{x},\v{y}\in \L_L\\ \rho,\rho'\in \{1,\dots,M\}}} a^{\dg}_{\v{x},\r} H^{0,L}_{\r\r'}(\v{x}-\v{y}) a_{\v{y},\r'}- \mu\mc{N}_L \equiv \sum_{\v{x},\v{y}\in\L_L} a^{\dg}_{\v{x}} \Big(H^{0,L}(\v{x}-\v{y}) - \mu\d_{\v{x},\v{y}}\Big) a_{\v{y}}, 
\end{equation}
where $\mu\in\mathbb R$ is the \it{chemical potential}, to be fixed in Assumption \ref{assum:H}, $\mc{N}_L:= \sum_{\v{x}\in\L_L}a^{\dg}_{\v{x}}a_{\v{x}}$ is the \emph{number operator}and with the understanding that $a_{\v{x}}$ (resp. $a^{\dg}_{\v{x}}$) denotes the column (resp. row) vector with components $(a_{\v{x},1}, \dots, a_{\v{x},M})$ (resp. $(a^{\dg}_{\v{x},1}, \dots, a^{\dg}_{\v{x},M})$).
The translation invariance of the model allows us to rewrite the Hamiltonian in momentum space. 
Let:
\begin{equation}
\mbb{B}_L:= \Big\{ \kappa_1 \v{g}_1+ \kappa_2\v{g}_2\; : \; \kappa_1,\kappa_2\in \big(\tfrac{1}{L}\mbb{Z}\big)\cap (-\tfrac{1}{2},\tfrac{1}{2}]  \Big\}
\end{equation}
be the \emph{discretized Brillouin zone}, where $\v{g}_1$ and $\v{g}_2$ are such that $\v{g}_i\cdot\v{a}_j= 2\pi\d_{i,j}$. We denote the \emph{infinite-volume Brillouin zone} by:
\begin{equation}
\label{def:Brillouin}
\mbb{B}:= \Big\{ \kappa_1 \v{g}_1+ \kappa_2\v{g}_2\; : \;  \kappa_1,\kappa_2\in (-\tfrac{1}{2}, \tfrac{1}{2}]  \Big\}
\end{equation}
 which is equipped with the norm $\dist{\v{k}}_{\mbb{B}}:= \min_{n_1,n_2\in\mbb{Z}}|\v{k}+ n_1\v{g}_1+ n_2\v{g}_2|$. 
 
 One defines, for any $\v{k}\in\mbb{B}_L$ and $\r\in\{1,\dots,M\}$,
\[\hat{a}_{\v{k},\r}:= \sum_{\v{x}\in\L_L} e^{i\v{k}\cdot\v{x}} a_{\v{x},\r}\,, \;\;\; \hat{a}^{\dg}_{\v{k},\r}= \sum_{\v{x}\in\L_L} e^{-i\v{k}\cdot\v{x}} a^{\dg}_{\v{x},\r}, \]
so that
\[a_{\v{x},\r}= \frac{1}{L^2} \sum_{\v{k}\in\mbb{B}_L} e^{-i\v{k}\cdot\v{x}} \hat{a}_{\v{k},\r}\,, \;\; a^{\dg}_{\v{x},\r}= \frac{1}{L^2} \sum_{\v{k}\in\mbb{B}_L} e^{i\v{k}\cdot\v{x}} \hat{a}^{\dg}_{\v{k},\r}. \]

Note that from Eq. \eqref{eq_1} it follows that for every $\v{p},\v{k}\in\mbb{B}_L$ and $\rho,\rho'\in \{1,\dots,M\}$:
\[\big\{ \hat{a}_{\v{k},\r}, \hat{a}_{\v{p},\r'} \big\}= \big\{ \hat{a}^{\dg}_{\v{k},\r}, \hat{a}^{\dg}_{\v{p},\r'} \big\}=0, \;\; \big\{ \hat{a}_{\v{k},\r}, \hat{a}^{\dg}_{\v{p},\r'} \big\}= L^2\d_{\v{k},\v{p}} \d_{\r,\r'}.\] 

Eq.\eqref{eq_2} can be hence rewritten as:
\[\mc{H}_L=\frac{1}{L^2}\sum_{\v{k}\in\mbb{B}_L} \hat{a}^{\dg}_{\v{k}} \,\big(\hat{H}^0(\v{k})-\mu\big) \,\hat{a}_{\v{k}}\]
with 
\begin{equation}
\label{def:Hamiltonian}
\hat{H}^0(\v{k}):= \sum_{\v{x}\in\L} e^{i\v{k}\cdot\v{x}} H^0(\v{x}), \;\; \v{k}\in \mbb{R}^2
\end{equation}
(observe that $\hat{H}^0$ does not depend on $L$). By Assumption \ref{assum:H0}, one has that $\hat{H}^0(\v{k})= \hat{H}^0(\v{k})^{\dg}$, $\hat{H}^0$ is periodic over $\mbb{B}$, i.e., 
$\hat H^0(\vec k+n_1\vec{g}_1+n_2\vec{g}_2)=\hat H^0(\vec k)$ for all $n_1,n_2\in\mathbb Z$, 
and analytic in $\v{k}$ in a complex neighborhood of $\mathbb B$.

Additionally, we assume that $H^0$ analytically depends on an extra \it{mass parameter} $m\in[-m_0,m_0]$, for some $m_0>0$. In order to emphasize such a dependence, from now on we shall denote by $H^0(\vec x;m)$ the kernel of $H^0$ and by $\hat H^0(\vec k;m)$ its Fourier transform. Correspondingly, we denote by $\lambda_1(\v{k};m)\le \ldots \le \l_M(\v{k};m)$  the eigenvalues of $\hat{H}^0(\v{k};m)$ (i.e., its \it{energy bands}), repeated according to multiplicity. Observe that the energy bands are continuous functions in $\v{k}$ and $m$, which follows from the fact that the roots of a polynomial are continuous functions of its coefficients \cite{Harrisetal}.

\begin{assumption}
\label{assum:H}
Let $\hat{H}^0(\vec k;m)$ depend analytically on $m$ in a complex neighborhood of $[-m_0,m_0]$. Furthermore, we assume the following conditions.
\begin{enumerate}
\item\label{it:1}  There exists an index $n$ with  $1\le n< M$ such that, denoting by $\l_-(\v{k};m)\equiv \l_n(\v{k};m)$ and $\l_+(\v{k};m)\equiv \l_{n+1}(\v{k};m)$, and letting 
$\mu_+(m):=\min_{\v{k}\in\mbb{B}} \l_+(\v{k};m)$
and $\mu_-(m):=\max_{\v{k}\in\mbb{B}} \l_-(\v{k};m)$, we have that 
\begin{equation}
\label{eqn:54}\begin{split}
& \mu_+(m)>\mu_+(0)=\mu_-(0)>\mu_-(m) \quad \forall \, m\in[-m_0,m_0]\setminus\{0\} 
\end{split}
\end{equation}
and
$\l_\pm(\v{k};0)=\mu_\pm(0)$ if and only if 
$\v{k}\in\{\v{k}_{F,1},\dots\v{k}_{F,N}\}$.
We call $\v{k}_{F,1},\dots\v{k}_{F,N}$ the \it{Fermi points}.

\item\label{it:2} Fixing the chemical potential 
$\mu=\mu_+(0)=\mu_-(0)$, there exists $\Delta'>0$ such that
\begin{equation}
\label{eqn:1}    
\min_{\v{k}\in\mbb{B}} \min_{\substack{j=1,...,M\\ j\ne n,n+1}} |\l_j(\v{k};0)- \mu|\ge \Delta'. 
\end{equation}

\item\label{it:3} For every $\o\in\{1,...,N\}$, which is  called the \emph{valley index}, there exist $\v{c}_{\o},\v{b}_{\o}\in\mbb{R}^2, u_{\o}\in\mbb{R},w_{\o}>0, A_{\o}\in \mathrm{Mat}(2,\mbb{R})$ invertible, such that
\begin{equation}
\label{eqn:bands1}
\l_{\pm}(\v{k}_{F,\o}+ \v{k}';m)- \mu= \v{c}_{\o}\cdot\v{k}'+ u_{\o}m \pm \sqrt{|A_{\o}\v{k}'|^2+ w_{\o}^2m^2+\v{k}'\!\cdot \v{b}_\o m} + \mc{O}(|\v{k}'|^2+m^2), 
\end{equation}
asymptotically as $|\v{k}'|,m\to 0$. Moreover, there exists a positive constant $C_{\star}$ such that
\begin{equation}
\label{eqn:bands2b}
\sqrt{|A_{\o}\v{k}'|^2+ w_{\o}^2m^2 +\v{k}'  \!\cdot \v{b}_\o m} -|\v{c}_{\o}\cdot\v{k}'+ u_{\o}m| \geq C_{\star}(|\v{k}'|+|m|)
\end{equation}
 for every $\o\in\{1,...,N\}$, $\v{k}'\in\mbb{R}^2, m\in \mbb{R}$.
\end{enumerate}
\end{assumption}

We will call $\l_{\pm}$ the \emph{semi-metallic} or \emph{quasi-massless} bands.

\begin{remark}
\label{rmk:bands}
Here are some comments about the properties of the energy bands implied by Assumption \ref{assum:H}.
\begin{enumerate}
\item\label{it:bands1} 
First, since inequality \eqref{eqn:bands2b} is homogeneous in $(\v{k}',m)$, requiring its validity for every $\v{k}'\in\mbb{R}^2,m\in\mbb{R}$ is equivalent to assume it asymptotically as $|\v{k}'|,m\to 0$. Second, notice that inequality \eqref{eqn:bands2b} is equivalent to
\[
\sqrt{|A_{\o}\v{k}'|^2+ w_{\o}^2m^2+\v{k}'\!\cdot \v{b}m}\pm(\v{c}_{\o}\cdot\v{k}'+ u_{\o}m)   \geq C_{\star}(|\v{k}'|+|m|),
\]
for every $\o\in\{1,...,N\}$, $\v{k}'\in\mbb{R}^2, m\in \mbb{R}$.

\item\label{it:bands2} Eq.\eqref{eqn:bands1} is equivalent to the general \emph{conical crossing} defined in \cite[Eq. (1.2)]{D21}, which characterizes a generic phase transition \cite[Theorem 1]{D21}. Indeed, under Assumptions \ref{assum:H0} and \ref{assum:H}.\ref{it:1}-\ref{assum:H}.\ref{it:2}, the remainder $\mathcal{O}(|\v k'|^2+ m^2)$ in the r.h.s.\ of Eq.\eqref{eqn:bands1}   is equivalent to be of the type $o(|\v k'|+|m|)$, as in \cite[Eq. (1.2)]{D21}.

\item\label{itt:1} From Assumption \ref{assum:H}.\ref{it:3} it follows that there exist $a>0$ and $C'_{\star}>0$ such that
\begin{align}
&\label{eq_6a}\tfrac{1}{2}C_{\star} \big(|\v{k}'|+|m|\big)\le \pm \Big( \l_{\pm}(\v{k}'+\v{k}_{F,\o};m)- \mu \Big) \le C'_{\star}\big(|\v{k}'|+|m|\big)
\end{align}
for every $ |\v{k}'|\le a$,  $|m|\le a$ and $\o\in\{1,...,N\}$.

\item\label{itt:4} Observe that from the above hypotheses, the following strengthened version of inequality \eqref{eqn:1} holds true:
\begin{equation}
\label{eqn:1bis}    
\min_{\substack{\v{k}\in\mbb{B}\\ |m|\le m_0}} \min_{\substack{j=1,...,M\\ j\ne n,n+1}} |\l_j(\v{k};m)- \mu|\ge \Delta
\end{equation}
for a suitable $\Delta>0$. 
Indeed, in view of \eqref{eqn:1} and the continuity of the energy bands, it is possible to find $0<m_1<m_0$ such that $\min_{ \substack{\v{k}\in\mbb{B}\\|m|\le m_1}} \min_{\substack{j=1,...,M\\ j\ne n,n+1}} |\l_j(\v{k};m)- \mu|\ge\frac{\Delta'}{2}$. On the other hand, for $m_1\le |m|\le m_0$ and $j>n+1$, we have that
\begin{equation*}
\l_j(\v{k};m)-\mu\ge \l_+(\v{k};m)-\mu \ge \mu_+(m)-\mu_+(0) \ge \min_{m_1\le |m|\le m_0}\big(\mu_+(m)-\mu_+(0)\big)>0.
\end{equation*}
Similarly, we have that for $m_1\le |m|\le m_0$ and $j<n$, $\l_j(\v{k};m)-\mu \le \max_{m_1\le |m|\le m_0}\big(\mu_-(m)-\mu_-(0)\big)<0$, so that, all in all, Eq. \eqref{eqn:1bis} holds for a suitable $\Delta>0$.

\item\label{itt:3} The continuity of the energy bands, in combination with the assumption that $\lambda_\pm(\vec k;0)=\mu$ only at the Fermi points, implies the \emph{local-gap property} for the bands $\l_{\pm}$, namely  for any $\d>0$,

\begin{equation}  \min_{\substack{m\in[-m_0,m_0], \v{k}\in\mbb{B},\\ \dist{\v{k}-\v{k}_{F,\o}}_{\mbb{B}}\ge\d, \,\forall \o}}|\l_{\pm}(\v{k};m)-\mu|>0.
\end{equation}

Indeed, if $0<|m|<m_0$, by Assumption \ref{assum:H} we have that $\l_+(\v{k};m)-\mu\ge \mu_+(m)-\mu_+(0)>0$. On the other hand, if $m=0$, $\min_{|\v{k}-\v{k}_{F,\o}|\ge\d, \,\forall \o} \big(\l_+(\v{k};m)-\mu\big)>0$ as $\l_+(\v{k};0)=\mu$ only if $\v{k}$ is one of the $N$ Fermi points. An analogous argument holds true also for the band $\l_-$.

\item\label{itt:2} Our setting, specified by Assumption \ref{assum:H}, extends the one of \cite{MPP26} to a whole neighborhood of the phase transition, considering also non-zero values of the mass parameter. In Section \ref{sect:examples} we exhibit several physically relevant examples of hopping Hamiltonians displaying a TPT and satisfying Assumptions \ref{assum:H0} and \ref{assum:H}.

\end{enumerate}
\end{remark}

The current operator is defined via Peierls substitution, namely:
\begin{equation}\label{eqn:2.12}
\hat{\mc{J}}_{j,\v{p}}:= L^2 \frac{\de \mc{H}_L}{\de\hat{A}_{j,\v{p}}} (\v{0}), \;\; \v{p}\in \mbb{B}_L \cap B_1(\v{0}),
\end{equation}
where $B_r(\v{0})$ denotes the $2D$ ball of radius $r$ centered at the origin (the radius $r$ is chosen to be 1 in \eqref{eqn:2.12} for definiteness, any finite neighborhood of the origin would be equally fine),

\medskip

\[\mc{H}_L(\v{A}):= \sum_{\v{x},\v{y}\in\L_L}\sum_{\r,\r'\in I} a^{\dg}_{\v{x},\r} \Big( H^{0,L}_{\r\r'}(\v{x}-\v{y}) - \mu\d_{\r,\r'}\d_{\v{x},\v{y}} \Big) a_{\v{y},\r'} e^{i \int_{\v{x}+\v{\d}_{\r}}^{\v{y}+\v{\d}_{\r'}} \v{d\ell}\cdot\v{A}}, \]
where $\int_{\v{x}+\v{\d}_{\r}}^{\v{y}+\v{\d}_{\r'}} \v{d\ell}\cdot\v{A}$ denotes the line integral of $\vec A\cdot d\vec \ell$ along the segment from $\v{x}+\v{\d}_{\r}$ to $\v{y}+\v{\d}_{\r'}$ on the discrete torus $\L_L$: 
\begin{equation*}
\int_{\v{x}+\v{\d}_{\r}}^{\v{y}+\v{\d}_{\r'}} \v{d\ell}\cdot\v{A}= \int_0^1ds \big(\vec{\pi}_L(\v{y}-\v{x})+ \v{\d}_{\r'}-\v{\d}_{\r} \big)\cdot \v{A}\big( \v{x}+\v{\d}_{\r}+ s(\vec{\pi}_L(\v{y}-\v{x})+ \v{\d}_{\r'}-\v{\d}_{\r}) \big)
\end{equation*}
with $\v{\pi}_L:\L\mapsto C_L:= \big\{n_1\v{a}_1+ n_2\v{a}_2:\; -\lfloor{\frac{L-1}{2}\rfloor}\le n_i \le \lfloor{\frac{L}{2}\rfloor} \big\}$ the projection map onto $C_L$ such that $\v{\pi}_L(\v{x})-\v{x}\in L\L$ for all $\v{x}\in\L$; $A_j$ with $j=1,2$ is a real-valued function 
of the plane, whose Fourier transform $\hat A_j$, defined via 
\[A_j(\v{x})\equiv \frac{1}{L^2}\sum_{\substack{\v{p}\in \mbb{B}_L\\ |\v{p}|\le1}}e^{-i\v{p}\cdot\v{x}} \hat{A}_{j,\v{p}}, \qquad \v{x}\in\mbb{R}^2,\]
is supported on $B_1(\vec 0)$, both for $j=1$ and $j=2$. Note that the current operator can be rewritten in the form:
\begin{equation}
\hat{\mc{J}}_{j,\v{p}}= \frac{1}{L^2} \sum_{\v{k}\in\mbb{B}_L} \hat{a}^{\dg}_{\v{k}+\v{p}} \,\Gamma_j^L(\v{k},\v{p};m)\, \hat{a}_{\v{k}}
\end{equation}
where $\Gamma_j^L$, $j=1,2$, are $M\times M$ matrix-valued functions, satisfying, in the infinite-volume limit:
\begin{equation}
\label{def_vertex}
\v{\Gamma}(\v{k},\v{0};m)_{\r\r'}:= \lim_{L\to\infty}\v{\Gamma}^L(\v{k},\v{0};m)_{\r\r'} = -\v{\nabla}_{\v{k}}\hat{H}^0(\v{k};m)_{\r\r'} - i(\v{\d}_{\r}-\v{\d}_{\r'}) \hat{H}^0(\v{k};m)_{\r\r'}, \end{equation}
as it follows from a straightforward computation.

The system at equilibrium at inverse temperature $\b>0$ is described by the \emph{Fermionic Gibbs state}:

\begin{equation}
\label{eqn:gibbs}
\langle \,\cdot\, \rangle_{\b,L}:= \frac{\mathrm{Tr}_{\mc{F}_L}\left\{e^{-\b \mc{H}_L} \,\cdot\, \right\}}{\mathrm{Tr}_{\mc{F}_L}\left\{e^{-\b \mc{H}_L} \right\}}.\end{equation}

\begin{definition}
We define the \emph{Euclidean current-current correlator}, at \emph{Matsubara frequency} $p_0\in\mbb{R}$ and at mass parameter $m\in[-m_0,m_0]$, in the thermodynamic and zero-temperature limits, as:

\begin{equation}
\label{def_K}
\begin{split}
\hat {K}_{ij}(p_0;m):=&\lim_{\beta\to\infty}\lim_{L\to\infty}\frac{1}{\beta L^2|\v{a}_1\times\v{a}_2|}\int_0^\beta dt_1\int_0^\beta dt_2\, e^{-ip_{0,\beta}(t_1-t_2)}\cdot\\
&\cdot \Big[\langle T\big(\hat{\mc{J}}_{\v{0},i}(t_1)\hat{\mc{J}}_{\v{0},j}(t_2)\big)\rangle_{\b,L}-\langle\hat{\mc{J}}_{\v{0},i}\rangle_{\b,L}\langle\hat{\mc{J}}_{\v{0},j}\rangle_{\b,L}\Big]
\end{split}
\end{equation}

for $i,j\in\{1,2\}$ , where:

\begin{itemize}
\item the dependence on $m$ in the r.h.s.\ is implicit both in the Gibbs state $\langle\cdot\rangle_{\b,L}$ and in the current operator;
\item $p_{0,\beta}=\tfrac{2\pi}{\beta}\lfloor \tfrac{\beta p_0}{2\pi}\rfloor$ is an integer multiple of $2\pi/\beta$ tending to $p_0$ as $\beta\to\infty$;
\item for any $0\le t<\beta$, $\hat{\mc{J}}_{\v{p},i}(t):=e^{t \mc{H}_L} \hat{\mc{J}}_{\v{p},i}e^{-t\mc{H}_L}$ is the \emph{imaginary-time evolution} of the current;
\item the operator $T$ inside the Gibbs state is the \emph{fermionic time-ordering}, acting on its argument as follows: 
\begin{equation*}
T\big(\hat{\mc{J}}_{\v{0},i}(t_1)\hat{\mc{J}}_{\v{0},j}(t_2)\big):= 
\left\{\begin{array}{cc}
\hat{\mc{J}}_{\v{0},i}(t_1)\hat{\mc{J}}_{\v{0},j}(t_2),     & \text{if}\ t_1>t_2 \\
\hat{\mc{J}}_{\v{0},j}(t_2)\hat{\mc{J}}_{\v{0},i}(t_1),     & \text{if}\ t_2>t_1,
\end{array}\right.
\end{equation*}
while at coinciding times $T$ acts as the fermionic normal order.
\end{itemize}
\end{definition}

\begin{definition}
\label{def:sigma} In units such that $e^2/\hbar=1$, the \emph{Kubo Euclidean conductivity matrix} $\s_{ij}(m)$ in the ground state is:
\begin{equation}
\s_{ij}(m):= \lim_{p_0\rightarrow 0^-} \frac{\hat{K}_{ij}(p_0;m)- \hat{K}_{ij}(0;m)}{p_0},
\end{equation}
whenever the limit in the r.h.s.\ exists.
\end{definition}

We are ready to state our main result.

\begin{theorem}
\label{thm:main}

Under Assumptions \ref{assum:H0} and \ref{assum:H}, we have that the conductivity matrix $\s(m)$ is well defined for any $m\in[-m_0,m_0]$ and satisfies:
\begin{equation}
\label{eqn:main}
\s_{ij}(0)= \frac{1}{2}\lim_{\e\rightarrow0^+} \Big(  \s_{ij}(\e)+ \s_{ij}(-\e) \Big)+ \frac{1}{16}\sum_{\o=1}^N \frac{(|A_{\o}|^2)_{ij}}{\det |A_{\o}|},
\end{equation}
where $A_{\o}$ is the $2\times2$ matrix defined in Eq.\eqref{eqn:bands1} and $|A_{\o}|\equiv (A^T_{\o}A_{\o})^{\frac{1}{2}}$, with $A_{\o}^T$ the transpose of $A_{\o}$.
\end{theorem}

\begin{remark}
\label{rmk:main} 
Some observations are in order.

\begin{enumerate}
    
\item \label{it:thm1} The matrix $A_{\o}$ is not defined uniquely, specifically it is defined up to the left multiplication by an orthogonal matrix, while  $A^T_{\o}A_{\o}\equiv |A_{\o}|^2$ is uniquely determined as it stems from the quadratic form $|A_{\o}\v{k}'|^2$. 

\item \label{it:thm2} For every $m\neq 0$, namely in the gapped case, the conductivity matrix can be rewritten as \cite{TKNN,BES94}
\begin{equation} 
\label{eqn:Chern}
\sigma_{ij}(m)=-i \int_{\mbb{B}} \frac{d\v{k}}{(2\pi)^2}\, \Tr_{\mbb{C}^M}\Big\{P_{\mu}(\v{k};m) \big[\de_{k_i}P_{\mu}(\v{k};m), \de_{k_j}P_{\mu}(\v{k};m) \big] \Big\},
\end{equation}
where 
$$P_{\mu}(\v{k};m) := \mathds{1}_{\{\hat{H}^0(\v{k};m)\le \mu \}}$$ is the \emph{Fermi projector}. In particular, $\s_{ii}(m)=0$ and $\s_{12}(m)= -\s_{21}(m)= \frac{1}{2\pi} \mathrm{Ch}(P_{\mu}) \in \frac{1}{2\pi}\mbb{Z}$, with $\mathrm{Ch}(P_{\mu})$ is the Chern number associated with the Fermi projector. Note also that $\s_{ij}$ is constant in $m$ as $m$ varies in the intervals $[-m_0,0)$ and $(0,m_0]$, since the Fermi projector is smooth there. Indeed, by the Riesz formula:
\begin{equation*}
P_{\mu}(\v{k};m)= \frac{i}{2\pi}\oint_{\gamma}dz\,\big(\hat{H}^0(\v{k};m)-z\big)^{-1},
\end{equation*}
with $\gamma$ a complex contour (independent of $m$ in view of \eqref{eqn:54}) encircling the real set $\l_1(\mbb{B};m)\cup\ldots\cup\l_-(\mbb{B};m)$ and with no intersection with the remaining spectrum $\l_+(\mbb{B};m)\cup\ldots\cup\l_M(\mbb{B};m)$. Therefore, the r.h.s.\ of Eq.\eqref{eqn:Chern} is constant in the same intervals (in general the values 
assumed in the two intervals are different, signaling the fact that the system belongs to two different topological phases for $m>0$ and $m<0$). 

\item \label{it:thm3} Remarkably, from Eq.\eqref{eqn:main} we have that the anti-symmetric part of the critical conductivity is quantized in units of $\frac{1}{2\pi}$, namely:
\begin{equation*}
\s_{12}(0)- \s_{21}(0)= \lim_{\e\to 0^+} \Big( \s_{12}(\e) + \s_{12}(-\e) \Big)= \s_{12}(m_0) + \s_{12}(-m_0) \in \tfrac{1}{2\pi}\mbb{Z},
\end{equation*}
in view of Item \ref{it:thm2}. This is in agreement with \cite{FGR25} in the context of the Haldane model, where an even stronger result was established, namely that $\s_{12}(0)=-\s_{21}(0)\in \frac{1}{4\pi}\mbb{Z}$, also in the presence of weak electron interactions. 

Thus, the \emph{non-universal} part of the critical conductivity matrix only stems from its symmetric part.   

\end{enumerate}
\end{remark}

\paragraph{Systems with an emergent discrete rotational symmetry.}

We conclude this section by discussing systems which exhibit an \emph{emergent discrete rotational symmetry}, namely an invariance under discrete rotations emerging at leading order in quasi-momentum and mass, for the Hamiltonian reduced to the two semi-metallic bands. This property is precisely specified in Eq. \eqref{eqn:35}.

Denoting by $P_\pm$ the eigen-projector associated with the eigenvalue $\lambda_\pm$ of $\hat{H}^0$, we define the semi-metallic projector:
\begin{equation}
\label{def:Pi}
\Pi(\v{k};m):=P_+(\v{k};m) +P_-(\v{k};m).
\end{equation}
For every $\o= 1,\dots,N$, we define the semi-metallic two-band Hamiltonian at valley index $\o$ as:
\begin{equation}
\label{eqn:R2}
\hat{H}^0_{\Pi_\o}(\v{k}';m):=\Pi_\o(\v{k}';m)\Big(\hat{H}^0_\o(\v{k}';m) -\mu \mathds{1}_M\Big)\Pi_\o(\v{k}';m),
\end{equation}
with $\Pi_{\o}(\v{k}';m):= \Pi(\v{k}_{F,\o}+\v{k}';m)$, $\hat{H}^0_{\o}(\v{k}';m):= \hat{H}^0(\v{k}_{F,\o}+\v{k}';m)$ and  the variable $\v{k}'$ is the quasi-momentum relative to the Fermi point $\v{k}_{F,\o}$.

The next assumption specifies what ``emergent discrete rotational symmetry" precisely means. 

\begin{assumption}[Emergent discrete rotational symmetry]
\label{assum:R}
For every $\o\in\{1,...,N\}$ there exist a unitary operator $\mathcal{R}_{\o}$  (independent of $(\v{k}';m)$) acting on $\mbb{C}^M$ and a rotation matrix $R_{\o}=\left(\begin{array}{cc}
    \cos{\alpha_{\o}} & -\sin{\alpha_{\o}} \\
     \sin{\alpha_{\o}}& \cos{\alpha_{\o}} 
\end{array}\right)$ with $\a_{\o}\in(0,2\pi)\setminus \{\pi\}$,\footnote{We exclude the value $\a_\o=\pi$ since in that case the rotational symmetry reduces to \emph{parity}, which is insufficient for Corollary \ref{cor:rotation} below to hold.} such that
\begin{equation}
\label{eqn:35}
\R^{\dg}_{\o} \,\hat{H}^0_{\Pi_\o}(\v{k}';m)\,\R_{\o}=\hat{H}^0_{\Pi_\o}(R_{\o}\v{k}';m)+\mc{O}({|\v{k}'|}^2+m^2),
\end{equation}
asymptotically as $|\vec k'|,m\to 0$. 
\end{assumption}

A particular case of emergent discrete rotational symmetry is that of \emph{exact} discrete rotational symmetry, namely the case where the whole Bloch Hamiltonian for every quasi-momentum and mass is covariant under discrete rotations, a condition that additionally requires that each Fermi point is preserved under the rotation. This is specified by the following assumption.

\begin{assumption}[Exact discrete rotational symmetry]
\label{assum:ER}
There exists a $\mathscr{C}^1$ unitary-valued map $\mbb{B}\times[-m_0,m_0]\ni (\vec k,m) \mapsto \R(\v{k},m)$, such that
\begin{equation}
\label{eqn:R1}
\R^{\dg}(\v{k},m) \,\hat{H}^0(\v{k};m)\,\R(\v{k},m)=\hat{H}^0(R\v{k} +\v{v};m),
\end{equation}
with $R= \left(\begin{array}{cc}
    \cos\alpha & -\sin\alpha \\
     \sin\alpha& \cos\alpha 
\end{array}\right)$, for some $\a\in(0,2\pi)\setminus\{\pi\}$ and some $\v{v}\in\mbb{B}$. Moreover, $R\v{k}_{F,\o} +\v{v}= \v{k}_{F,\o}$ mod. $\Span_{\mbb{Z}}\{\v{g}_1,\v{g}_2\}$, for every $\o\in\{1,...,N\}$. 
\end{assumption}

As shown in Appendix \ref{app:rot}, Assumption \ref{assum:ER} implies the validity of Assumption \ref{assum:R}, see Lemma \ref{lemma:R_exact}.

\begin{corollary}
\label{cor:rotation}
Under Assumptions \ref{assum:H0}, \ref{assum:H} and \ref{assum:R}, we have that
\begin{equation}
\s_{ij}(0)= \frac{1}{2} \lim_{\e\to 0^+} \Big( \s_{ij}(\e)+ \s_{ij}(-\e) \Big)+ \frac{N}{16}\d_{i,j}.
\end{equation}
\end{corollary}

\begin{remark}
Notice that Assumption \ref{assum:R} is strictly more general than Assumption \ref{assum:ER}. In fact, consider as an example the following 
\emph{anisotropic deformation} of the standard nearest neighbor lattice graphene model, defined in terms of the Bloch-Hamiltonian:
\begin{equation}
\hat{H}^0_{\eta}(\v{k};m):= \left(\begin{array}{cc}
     m & -t_1\Omega^*(\v{k})- t_1 \eta |\Omega(\v{k})|^2 \\
    -t_1\Omega(\v{k})- t_1\eta |\Omega(\v{k})|^2 & - m
\end{array} \right),
\end{equation}
where $t_1>0$ is the hopping parameter, $m\in \mbb{R}$ is the mass, and $\Omega(\v{k}):= 1+ e^{-i\v{k}\cdot\v{a}_1}+ e^{-i\v{k}\cdot\v{a}_2}$, where $\v{a}_1= \frac{1}{2}(3,-\sqrt{3})$ and $\v{a}_2= \frac{1}{2}(3,\sqrt{3})$ are the generators of the triangular lattice of lattice spacing $\sqrt{3}$. We assume the deformation parameter $\eta$ to satisfy $0\leq \eta<\big(\max_{\v{k}\in\mbb{B}}|\Omega(\v{k})|\big)^{-1}$, which ensures that the two energy bands touch only if $m=0$ and $\v{k}= \v{k}_{F,\pm}$, where $\v{k}_{F,+}\equiv\v{k}_{F,1}= \frac{4\pi}{3}(0,1)$ and $\v{k}_{F,-}\equiv \v{k}_{F,2}=-\v{k}_{F,1}$. For $\eta=0$ the model fulfills Assumption \ref{assum:ER} with $\a=\frac{2\pi}{3}$, $\mc{R}(\v{k};m)\equiv \mc{R}(\v{k}) = e^{-\frac{i}{2}\v{k}\cdot\v{a}_2(\mathds{1}_2-\s_3)}$ and $\v{v}=\v{0}$, see \cite[Lemma B.1]{GMP12}. For $\eta$ non-zero, such exact symmetry does not hold anymore. Indeed, one can easily check that the energy bands are not invariant under the spatial rotation $\v{k}\mapsto R\v{k}$. Conversely, for every admissible $\eta$, Assumption \ref{assum:R} is still satisfied, since $\Omega(\v{k}_{F,\pm}+\v{k}')=\frac{3}{2}(ik'_1\pm k'_2)+ \mc{O}(|\v{k}'|^2)$ and the perturbation is quadratic in $\Omega$.
\end{remark}

\section{Examples}
\label{sect:examples}

This section is devoted to presenting three lattice models of physical interest that satisfy Assumptions \ref{assum:H0}-\ref{assum:H}, so that Theorem \ref{thm:main} can be applied. The critical conductivity matrix is explicitly computed, providing a characterization of the TPT. The first example we consider is the Hofstadter model with magnetic flux $\frac{\pi}{2}$ and staggered potential, which displays the same critical conductivity matrix as the one of the Haldane model, in particular $2\pi \sigma_{12}(0)=-1/2$. The second examples is again  the Hofstadter model now with magnetic flux $\frac{\pi}{3}$ and anisotropic on-site potential; in this case the transverse conductivity attains a different half-integer value, namely $2\pi\sigma_{12}(0)=3/2$; similar examples with higher half-integer values of the conductivity can be easily constructed likewise. 
Finally, we compute $\s(0)$ for the strained Haldane model, and show that it is manifestly non-universal, i.e., 
it varies continuously with the strain strength. 

{For the analysis of the first and second examples, we have made use of a computing software \cite{Wol} in order to check 
the validity of Assumption 2. Specifically, the numerical computations for the first and second example are collected in the files \texttt{MathematicaNotebook1.nb} and \texttt{MathematicaNotebook2.nb} respectively, which are available among the ancillary material on the arXiv page of this article.}

\subsection{The Hofstadter model with flux $\Phi=\frac{\pi}{2}$ and staggered potential}
\label{ssect:Hofs_1}
As a first example, we consider one of the simplest models used for the study of the Quantum Hall Effect, namely the Hofstadter model \cite{Hofstadter}. Specifically, we start by considering the model with magnetic flux $\Phi=\frac{\pi}{2}$, in the \emph{square gauge}. Denoting by $\v{e}_1$ and $\v{e}_2$ the canonical basis of $\mbb{R}^2$, 
the lattice is generated by $\v{a}_1= 2\v{e}_1$ and $ \v{a}_2=2\v{e}_2$. Within each unit cell we consider four internal degrees of freedom, labeled by $\r\in\{1,\dots,4\}$, associated with the four sites in Fig. \ref{fig_Hofstadter_square}.

\begin{figure}[ht]
\centering
\includegraphics[scale=1]{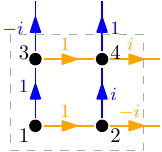}
\caption{graphical representation of the Hofstadter model with magnetic flux $\Phi= \frac{\pi}{2}$ in the square gauge. The dashed gray square identifies the unit cell; black dots represent the sites within the unit cell, indexed by $\r\in\{1,\dots,4\}$; the blue (resp. orange) arrows represent the vertical (resp. horizontal) hoppings with the related hopping strength.}
\label{fig_Hofstadter_square}
\end{figure}

The Hofstadter Hamiltonian on $\L_L$ reads:
\begin{equation}
\begin{split}
\mathcal{H}^{\pi/2}_{\mathrm{Hofs},L}:=& - \sum_{\v{x}\in\L_L} \Big( a^{\dg}_{\v{x},2} a_{\v{x},1}+ a^{\dg}_{\v{x},3} a_{\v{x},1} + i a^{\dg}_{\v{x},4} a_{\v{x},2}+ a^{\dg}_{\v{x},4} a_{\v{x},3}+\\
& -i a^{\dg}_{\v{x}+2\v{e}_1,1} a_{\v{x},2} + i a^{\dg}_{\v{x}+2\v{e}_1,3} a_{\v{x},4} -i a^{\dg}_{\v{x}+2\v{e}_2,1} a_{\v{x},3} + a^{\dg}_{\v{x}+ 2\v{e}_2,2} a_{\v{x},4} \Big) + \text{h.c.},
\end{split}
\end{equation}
where \virg{$\text{h.c.}$} stands for hermitian conjugate. In order to observe a TPT, it is convenient to add a \emph{staggered potential} of the form (see \cite{Nature15}):
\begin{equation}
\mc{V}_{1,L}= -\sum_{\v{x}\in\L_L} \Big( a^{\dg}_{\v{x},1}a_{\v{x},1} - a^{\dg}_{\v{x},4}a_{\v{x},4} \Big)
\end{equation}
to the Hamiltonian. 
We then let the finite-volume Hamiltonian be
$\mathcal{H}_L=\mathcal{H}^{\pi/2}_{\mathrm{Hofs},L}+d\mathcal{V}_{1,L} -\mu\mc{N}_L$, with $d\in\mbb{R}$ playing the role of the mass parameter and $\mu$ to be fixed in the following lines. From $\mc{H}_L$ we can readily obtain the Hamiltonian kernel $H^0$, and thus the Bloch Hamiltonian:
\begin{equation}
\label{eqn:Hofst_Hamiltonian}
\hat{H}^0(\v{k})= -\left(\begin{array}{cccc}
d  & 1-i e^{2ik_1} & 1-ie^{2ik_2} & 0 \\
1+i e^{-2ik_1}& 0& 0  & -i+ e^{2ik_2}\\
1+ i e^{-2ik_2} & 0 & 0 & 1+i e^{2ik_1}\\
0 & i+ e^{-2ik_2} & 1- i e^{-2i k_1} & -d
\end{array} \right).
\end{equation}
There are four energy bands in $k$-space whose behavior depends on the value of $d$. Specifically, as shown in the Figure \ref{fig:bands_1}, for $d=0$ only the central bands intersect at $\vec k=(\pm\frac{\pi}{4},\pm\frac{\pi}{4})$ and $\vec k=(\pm\frac{\pi}{4},\mp\frac{\pi}{4})$, while for $d=2$:
\begin{itemize}
\item the third energy band has a conical intersection with the fourth one at exactly one point $\v{k}_F:=(\frac{\pi}{4},\frac{\pi}{4})$;
\item the first energy band has a conical intersection with the second one at exactly $-\v{k}_F$;
\item the second and third bands have a conical intersection at $(\pm\frac{\pi}{4},\mp\frac{\pi}{4})$.
\end{itemize}

\begin{figure}[ht]
\centering
\includegraphics[scale=1]{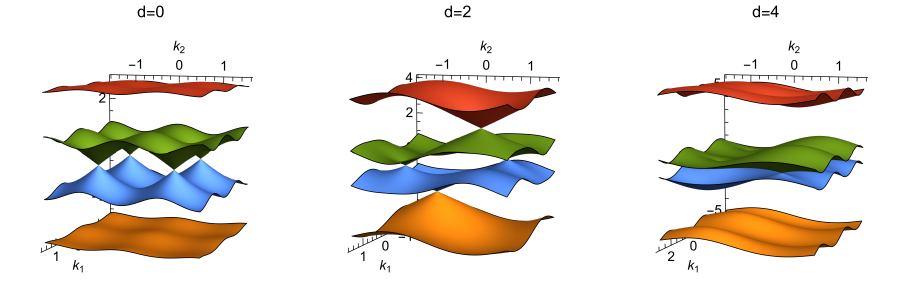}
\caption{plot of the energy bands of the Bloch Hamiltonian in Eq. \eqref{eqn:Hofst_Hamiltonian}, at the values $d=0$, $d=2$ and $d=4$ respectively.}
\label{fig:bands_1}
\end{figure}

Letting the mass parameter be $m:= d-2$ and 
$\hat H^0(\vec k;m):=\hat H^0(\vec k)|_{d=2+m}$, we then set the Fermi energy $\mu=\l_3(\v{k}_F;m)\big|_{m=0}= \l_4(\v{k}_F;m)\big|_{m=0}$.

In the \texttt{MathematicaNotebook1.nb} sheet we provide numerical evidence of the validity of Assumption \ref{assum:H}.\ref{it:1}, and specifically of the fact that for all
$m\in[-2,2]\setminus\{0\}$, $\min_{\v{k}\in\mbb{B}}(\l_4(\v{k};m)-\mu)>0$ and $\max_{\v{k}\in\mbb{B}}(\l_3(\v{k};m)-\mu)<0$. We expect that, with some additional technical effort, this can be fully proved analytically, but we prefer not to belabor the details of such a proof here. 
Since we deal with only one Fermi point, in the following we drop the label $\omega$ in all the involved quantities. The semi-metallic energy bands at dominant order can be computed in terms of the linearized 
Hamiltonian in the eigenspace associated with the second and third energy bands, given by (see Remark \ref{rmk:lin}): 
\begin{equation}
\label{eqn:53}
\hat{h}^{(1)}_{\r\r'}(\v{k}';m)= \sum_{j=1,2} k'_j \langle \phi_{\r}|\de_{k_j}\hat{H}^0(\v{k}_{F};0)| \phi_{\r'}\rangle + m \langle\phi_{\r}|\de_m\hat{H}^0(\v{k}_{F};0)| \phi_{\r'}\rangle,
\end{equation}
with $\r,\r'\in\{1,2\}$ and $\{\phi_1,\phi_2\}$ any system of orthonormal eigenvectors of $\hat{H}^0(\v{k}_F;0)$ with eigenvalue $\mu$. A straightforward computation shows that
\begin{align*}
\hat{h}^{(1)}(\v{k}';m)&=\tfrac{1}{3}m\mathds{1}_2 -\tfrac{2}{\sqrt{3}} k'_1 \s_2+ \tfrac{2}{\sqrt{3}}k'_2 \s_1 -\tfrac{2}{3}m\s_3,
\end{align*}
see \texttt{MathematicaNotebook1.nb}.
The eigenvalues of $\hat h^{(1)}(\vec k';m)$ coincide with $\lambda_\pm(\vec k';m)$ at dominant order in $\vec k'$ and $m$ (see Remark \ref{rmk:lin}), which leads to:
\[
\l_{\pm}(\v{k}_F+\v{k}';m)- \mu= \tfrac{1}{3}m \pm \sqrt{\tfrac{4}{3}|\v{k}'|^2 + \tfrac{4}{9}m^2} + \mc{O}(|\v{k}'|^2+ m^2),
\]
from which Assumption \ref{assum:H} is readily verified. Therefore, Theorem \ref{thm:main} yields:
\begin{equation}
\label{eqn:43}
\s(0)= \frac{1}{2} \Big(\s(0^+)+ \s(0^-) \Big)+ \frac{1}{16} \mathds{1}_2,
 \end{equation}
 where $\s(0^{\pm})\equiv \lim_{\e\to 0^+}\s(\pm \e)$. 

 \begin{remark}
Eq.\eqref{eqn:43} could also be derived as an application of Corollary \ref{cor:rotation}. Indeed, as one can easily check, the present model exhibits an exact rotation symmetry, namely it satisfies Assumption \ref{assum:ER} with $\a=\frac{\pi}{2}$, $\v{v}=\frac{\pi}{2}(1,0)$ and

\[
\mc{R}(\v{k};m)\equiv \mc{R}(\v{k}) = \left(\begin{array}{cccc}
1  & 0 & 0 & 0 \\
0& 0& 1  & 0\\
0 & ie^{-2ik_2} & 0 & 0\\
0 & 0& 0 & e^{-2ik_2}
\end{array} \right).
\]

It turns out that the transformation $\hat{H}^0(\v{k};m)\mapsto \mc{R}(\v{k}) \hat{H}^0(R\v{k}+ \v{v};m) \mc{R}^{\dg}(\v{k};m)$ can be rephrased, in position space, as the geometric rotation by $\frac{\pi}{2}$ in the counter-clockwise direction, accompanied by a gauge transformation which restores the square gauge.
 \end{remark}
 
Let us go back to Eq. \eqref{eqn:43}. Since, as already stated above, for $m\in [-2,0)\cup(0,2]$ the system is gapped, then trivially $\s_{11}(0^{\pm})=\s_{22}(0^{\pm})=0$. Moreover, $\s_{12}(0^{\pm})=-\frac{1}{2\pi}\mathrm{Ch}(\tilde P_{\l_4})(0^{\pm})$, where $\mathrm{Ch}(\tilde P_{\l})(m)$ is the Chern number associated with the projector $\tilde P_{\l}$ over the (isolated) energy band $\l$ at the mass value $m$. It is known that $\mathrm{Ch}(\tilde P_{\l_4})(0^-)=\mathrm{Ch}(\tilde P_{\l_4})(-2)=+1$ (see \cite{TKNN}). 
Hence we are left with computing $\mathrm{Ch}(\tilde P_{\l_4})(0^+)$, namely the value of the Chern number after the TPT. The answer is provided by \cite{B95} and \cite[Theorem 2]{D21} (see also \cite{MP14}), which establish:
\begin{equation}
\label{eqn:55}
\mathrm{Ch}(\tilde P_{\l_4})(0^+)-\mathrm{Ch}(\tilde P_{\l_4})(0^-) =\mathrm{sgn}\det J,
\end{equation}
where, letting $f_s(\v{k}';m):= \frac{1}{2}\mathrm{Tr}\big\{\s_s \hat{h}^{(1)}(\v{k}';m) \big\}$, with $s\in\{1,2,3\}$ and $\s_1,\s_2,\s_3$ the Pauli matrices, 
\begin{equation*}
J:= \left(\begin{array}{ccc}
\de_{k_1}f_1(\v{0};0) & \de_{k_2}f_1(\v{0};0) & \de_{m}f_1(\v{0};0)\\
\de_{k_1}f_2(\v{0};0) & \de_{k_2}f_2(\v{0};0) & \de_{m}f_2(\v{0};0)\\
\de_{k_1}f_3(\v{0};0) & \de_{k_2}f_3(\v{0};0) & \de_{m}f_3(\v{0};0)
\end{array} \right) \equiv \left(\begin{array}{cccc}
0 & \frac{2}{\sqrt{3}} & 0 \\
-\frac{2}{\sqrt{3}}& 0 & 0\\
0& 0 & -\frac{2}{3}
\end{array}\right). 
\end{equation*}
By inspection we have that $\det J<0$, so that $\mathrm{Ch}(\tilde P_{\l_4})(0^+)=0$.

Therefore, by Eq.\eqref{eqn:43}, we have that
 \begin{equation*}
\s(0)= \left(\begin{array}{cc}
         \frac{1}{16}& -\frac{1}{4\pi} \\
         \frac{1}{4\pi} & \frac{1}{16}
    \end{array} \right).
 \end{equation*}
Note incidentally that this conductivity matrix $\s(0)$ is the same obtained in the context of the (interacting) Haldane model, see \cite[Fig.1]{FGR25}. In order to get different values 
for the conductivity, in particular for the transverse one, it is enough to consider 
different values of the magnetic field, as discussed in the next subsection. 

\subsection{The Hofstadter model with flux $\Phi=\frac{\pi}{3}$ and on-site potential}
\label{ssect:Hofs_2}
A natural direction for exploring values of the Hall conductivity $\s_{12}(0)$ different from $\pm\frac{1}{2}\cdot\frac{1}{2\pi}$, is to investigate the Hofstadter model with a higher number of energy bands, where the insulating phases display higher values of $\s_{12}$.

We consider the Hofstadter model with magnetic flux $\Phi=\frac{\pi}{3}$; we choose the gauge such that the magnetic vector potential $\v{A}\colon\mbb{R}^2\to\mbb{R}^2$ is given by $\v{A}(x_1,x_2)=(-\frac{\pi}{3}x_2,0)$.

Denoting by $\v{e}_1$ and $\v{e}_2$ the canonical basis of $\mbb{R}^2$, 
the lattice is generated by $\v{a}_1= \v{e}_1$ and $ \v{a}_2=6\v{e}_2$. Within each unit cell we consider six internal degrees of freedom, labeled by $\r\in\{1,\dots,6\}$, associated with the six sites in Fig. \ref{fig_Hofstadter}. 

\begin{figure}[ht]
\centering
\includegraphics[scale=0.75]{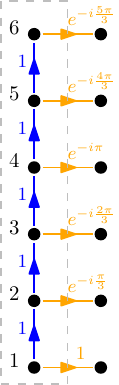}
    \caption{graphical representation of the Hofstadter model with magnetic flux $\Phi= \frac{\pi}{3}$ and one-site potential. The dashed gray square identifies the unit cell; black dots represent the sites within the unit cell, indexed by $\r\in\{1,\dots,6\}$; the blue (resp. orange) arrows represent the vertical (resp. horizontal) hoppings with the related hopping strengths.}
    \label{fig_Hofstadter}
\end{figure}

The Hamiltonian operator on $\L_L$ is defined as $\mathcal{H}_L=\mathcal{H}^{\pi/3}_{\mathrm{Hofs},L}+d\mathcal{V}_{2,L} -\mu\mc{N}_L$,
with:
\begin{align}
\mathcal{H}^{\pi/3}_{\mathrm{Hofs},L}&\label{eqn:Hofs_1a}:= -\sum_{\v{x}\in \L_L} \Big(\sum_{\r=1}^6 e^{-2\pi i (\frac{\r-1}{6})}a^{\dg}_{\v{x}+\v{e}_1,\r}  a_{\v{x},\r}+ \sum_{\r=1}^5 a^{\dg}_{\v{x},\r+1}  \,a_{\v{x},\r}  + a^{\dg}_{\v{x}+6\v{e}_2,1}  a_{\v{x},6}  \Big)+\text{h.c.}\\
\mathcal{V}_{2,L}&\label{eqn:Hofs_1b}:= - \sum_{\v{x}\in \Lambda_L} \Big( -\frac{1}{2}a^{\dg}_{\v{x},1}  a_{\v{x},1}+a^{\dg}_{\v{x},3}a_{\v{x},3}-\frac{1}{2}a^{\dg}_{\v{x},5}a_{\v{x},5} \Big),
\end{align}
where $d\in\mbb{R}$, $\mu$ to be fixed in the following lines, and \virg{h.c.} stands for \virg{Hermitian conjugate}. From $\mc{H}_L$ we can obtain the Hamiltonian kernel $H^0$, and thus the Bloch Hamiltonian $\hat{H}^0$, which is not reported here for the sake of brevity. There are six energy bands in $k$-space whose behavior depends on the value of $d$. Specifically, as shown in Figure \ref{fig:bands_2}, for $d=0$ only the two central bands have non-empty intersection, while\footnote{Here, $d_{cr}:=\frac{2}{3} \Big(4-\frac{4\cdot 2^{\frac{2}{3}}}{{(-13 + 3 \sqrt{33})}^{\frac{1}{3}}} + {(2 (-13 + 3 \sqrt{33}))}^{\frac{1}{3}}\Big)= 1.4088$ has been computed numerically in \texttt{MathematicaNotebook2.nb}}. for $d=d_{cr}$  the second energy band touches conically the third one at exactly one Fermi point $\v{k}_F=(-\frac{\pi}{3},0)$. 

\begin{figure}[ht]
\centering
\includegraphics[scale=1]{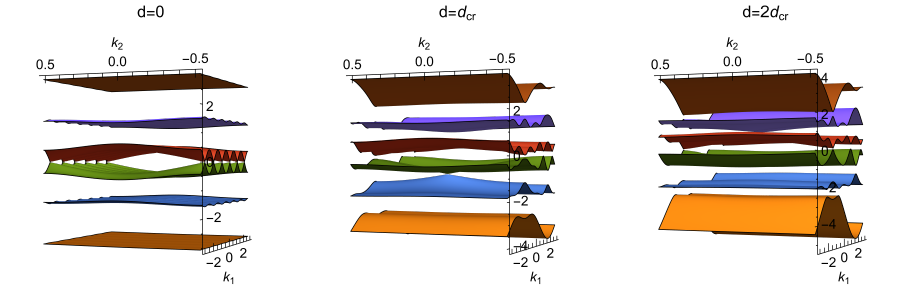}
\caption{plot of the energy bands of the model defined by Eqs. \eqref{eqn:Hofs_1a}-\eqref{eqn:Hofs_1b}, at the values $d=0$, $d=d_{cr}$ and $d=2d_{cr}$ respectively.}
\label{fig:bands_2}
\end{figure}

Letting $m:= d-d_{cr}$, we then set the Fermi energy $\mu=\l_2(\v{k}_F;m)\big|_{m=0}= \l_3(\v{k}_F;m)\big|_{m=0}$.
Again, since we deal with only one Fermi point, in the following we drop the $\omega$-dependence of the subsequent quantities. 
It can be easily checked numerically \cite{Wol} that the linearized Hamiltonian (compare with Eq. \eqref{eqn:53}) in the eigenspace associated with the second and third energy bands, is given by:
\begin{align*}
\hat{h}^{(1)}(\v{k}';m)&=-(\alpha_1\s_1+\alpha_3\s_3)k'_1+ \beta_2\s_2k'_2 -(-\gamma_0\id_2 - \gamma_1\s_1 +\gamma_3\s_3)m\\
&\equiv \gamma_0 m\id_2 -(\a_1k'_1- \gamma_1 m)\s_1 +\beta_2k'_2\s_2 -(\alpha_3k'_1+\gamma_3m)\s_3,
\end{align*}
where, the coefficients $\a_1,\a_3,\b_2,\gamma_0,\gamma_1,\gamma_3$ have values in the following intervals (see \\\texttt{MathematicaNotebook2.nb}):
\begin{equation}
\label{eqn:numerics}
\begin{split}
& 0.7776\le \a_1\le 0.7778, \quad 0.5546\le \a_3\le 0.5548, \quad 1.2770\le \b_2\le 1.2772,\\
&0.0140\le \gamma_0\le 0.0142, \quad 0.2157\le \gamma_1\le 0.2159, \quad 0.3025\le \gamma_3\le 0.3027.
\end{split}
\end{equation}

Thus, in analogy with the previous example (see also Remark \ref{rmk:lin}) we have that
\[
\l_{\pm}(\v{k}_F+\v{k}';m)-\mu=\gamma_0m\pm \sqrt{|A\v{k}'|^2 + (\gamma_1^2 + \gamma_3^2)m^2 +2(\a_1\gamma_1-\a_3\gamma_3) k'_1 m} + \mc{O}(|\v{k}'|^2+m^2),
\]

with $A= \left(\begin{array}{cc}
    \sqrt{\a_1^2+ \a_3^2} & 0 \\
    0 & \b_2
\end{array} \right)$. The validity of Assumption \ref{assum:H}.\ref{it:3} follows by direct inspection; the validity of Assumptions \ref{assum:H}.\ref{it:1} and \ref{assum:H}.\ref{it:2} has been checked numerically, see \\\texttt{MathematicaNotebook2.nb}; also in this case, as for the one discussed in the previous subsection, we expect that these assumptions can be fully proved analytically, with some additional technical effort, but we prefer not to belabor the details of such a proof here.
In view of these facts, Theorem \ref{thm:main} applies, yielding:
\begin{equation}
\label{eqn:42}
\s(0)= \frac{1}{2} \Big(\s(0^+)+ \s(0^-) \Big) + \frac{1}{16} \frac{(A^TA)}{|\det A|} \equiv \frac{1}{2} \Big(\s(0^+)+ \s(0^-) \Big)+ \frac{1}{16} \left(\begin{array}{cc}
    \frac{\sqrt{\a_1^2+\a_3^2}}{\b_2} & 0 \\
    0 & \frac{\b_2}{\sqrt{\a_1^2+\a_3^2}}
\end{array} \right),
 \end{equation}
 where $\s(0^{\pm})\equiv \lim_{\e\to 0^+}\s(\pm \e)$. Thus in order to compute $\s_{12}(0)$ we are left with computing $\s_{12}(0^{\pm})$.
 Notice that $\s_{12}(0^{\pm})=\frac{1}{2\pi}\mathrm{Ch}(\tilde P_{\l_1} + \tilde P_{\l_2})(0^{\pm})$. \\Moreover, it can be checked that the band $\l_1$ remains isolated from all the others for any $m\in[-d_{cr},\frac{1}{5}d_{cr}]$ (see \texttt{MathematicaNotebook2.nb}) and the band $\l_2$ remains isolated from all the others for any $m\in[-d_{cr},0)\cup(0,\frac{1}{5}d_{cr}]$. Hence we have that $\mathrm{Ch}(\tilde P_{\l_1} + \tilde P_{\l_2})(0^{\pm})= \mathrm{Ch}(\tilde P_{\l_1})(0^{\pm}) + \mathrm{Ch}(\tilde P_{\l_2})(0^{\pm})$, $\mathrm{Ch}(\tilde P_{\l_1})(0^+)= \mathrm{Ch}(\tilde P_{\l_1})(0^-)= \mathrm{Ch}(\tilde P_{\l_1})(-d_{cr})$ and $\mathrm{Ch}(\tilde P_{\l_2})(0^-) = \mathrm{Ch}(\tilde P_{\l_2})(-d_{cr})$. Moreover it is known \cite{TKNN} that $\mathrm{Ch}(\tilde P_{\l_1})(-d_{cr})= \mathrm{Ch}(\tilde P_{\l_2})(-d_{cr})= +1$. In order to compute $\mathrm{Ch}(\tilde P_{\l_2})(0^+)$, we proceed as in the previous example, namely we apply the analogue\footnote{Observe the negative sign in the r.h.s. of Eq. \eqref{eqn:55bis}, differently from the positive sign in the r.h.s. of Eq. \eqref{eqn:55}. This is due to the fact that Eq. \eqref{eqn:55} applies in general to the change of the Chern number associated to the \emph{upper band} (i.e. the band above $\mu$) \cite{B95}, while in the present case $\l_2$ plays the role of the \emph{lower band}.} of Eq. \eqref{eqn:55}:
 \begin{equation}
 \label{eqn:55bis}
 \mathrm{Ch}(\tilde P_{\l_2})(0^+)-\mathrm{Ch}(\tilde P_{\l_2})(0^-) = -\mathrm{sgn}\det J,
 \end{equation}
 where, letting $f_s(\v{k}';m):= \frac{1}{2}\mathrm{Tr}\big\{\s_s \hat{h}^{(1)}(\v{k}';m) \big\}$, with $s\in\{1,2,3\}$, 
\begin{equation*}
J:= \left(\begin{array}{ccc}
\de_{k_1}f_1(\v{0};0) & \de_{k_2}f_1(\v{0};0) & \de_{m}f_1(\v{0};0)\\
\de_{k_1}f_2(\v{0};0) & \de_{k_2}f_2(\v{0};0) & \de_{m}f_2(\v{0};0)\\
\de_{k_1}f_3(\v{0};0) & \de_{k_2}f_3(\v{0};0) & \de_{m}f_3(\v{0};0)
\end{array} \right) \equiv \left(\begin{array}{cccc}
    -\a_1 & 0 & \gamma_1 \\
     0& \b_2 & 0\\
     -\a_3& 0 & -\gamma_3
\end{array}\right). 
\end{equation*}

By inspection we have that $\det J>0$, so that $\mathrm{Ch}(\tilde P_{\l_2})(0^+)=0$. Consequently, we find that $\s_{12}(0^-)= 2\cdot\frac{1}{2\pi}$ and $\s_{12}(0^+)= \frac{1}{2\pi}$ and hence, by Eq. \eqref{eqn:42}, $\s_{12}(0)=\frac{3}{2}\cdot\frac{1}{2\pi}$. This is a novel value for the Hall conductance, compared to those observed in the Haldane model \cite{FGR25}. All in all for the conductivity matrix at the transition we get:

\begin{equation}
\s(0)= \left(\begin{array}{cc}
    \frac{\sqrt{\a_1^2+\a_3^2}}{16\b_2} & \frac{3}{4\pi} \\
    -\frac{3}{4\pi} & \frac{\b_2}{16\sqrt{\a_1^2+\a_3^2}}
\end{array} \right).    
\end{equation}

We finally observe that the anisotropy of the on-site potential causes $\s_{11}(0)\neq \s_{22}(0)$, which are both different from $\frac{1}{16}$, differently from the previous example. Indeed, by Eq. \eqref{eqn:numerics}, $0.747\cdot \frac{1}{16} \le \s_{11}(0)\le 0.749\cdot \frac{1}{16}$ and $1.336 \cdot \frac{1}{16} \le \s_{22}(0)\le 1.338 \cdot \frac{1}{16}$.

\subsection{The strained Haldane model}
\label{ssect:strain}

As a concluding example, we analyze a model which exhibits explicit non-universality for the Hall conductivity $\s_{12}$, namely lack of half-integer quantization. We consider a \emph{strained} version of the Haldane model, for which we refer to \cite{Mannai2020}. The lattice is generated by the vectors (we choose units for which $a\equiv \frac{1}{\sqrt{3}}|\v{a}_1|$ is equal to 1):

\[
\v{a}_1= \left(\begin{array}{cc}
\sqrt{3}  \\
 0  
\end{array}\right), \qquad \v{a}_2= \frac{1}{2}\left(\begin{array}{cc}
-\sqrt{3}  \\
3+2\eta  
\end{array}\right),
\]

with $\eta\in\mbb{R}$, to be fixed sufficiently small compared to 1.\footnote{Compared to \cite{Mannai2020}, in order to avoid confusion with the notation, we denote the strain parameter by $\eta$ instead of $\e$.} In units such that $\hslash=1$, the Bloch Hamiltonian reads (cf. \cite[Eq. (6)]{Mannai2020}\footnote{Note that, differently from our conventions, in \cite{Mannai2020} the Hamiltonian is supposed to be \emph{equivariant} rather than periodic: in the r.h.s. of Eq. \eqref{eqn:strain1} we multiplicate by $e^{\mp\frac{1}{2}i\v{k}\cdot\v{\t}_2(\mathds{1}_2-\s_2)}$ so that $\hat{H}^0$ is periodic.}):

\begin{equation}
\label{eqn:strain1}
\hat{H}^0(\v{k})= e^{-\frac{1}{2}i\v{k}\cdot\v{\t}_2(\mathds{1}_2-\s_3)}\left(\begin{array}{cc}
h_{AA}(\v{k})     & h^*_{AB}(\v{k}) \\
h_{AB}(\v{k})     & h_{BB}(\v{k})
\end{array} \right) e^{+\frac{1}{2}i\v{k}\cdot\v{\t}_2(\mathds{1}_2-\s_3)},
\end{equation}

where:

\begin{itemize}
\item $\v{\t}_1= \frac{1}{2}(\sqrt{3},1)$, $\v{\t}_2= \frac{1}{2}(-\sqrt{3},1)$, $\v{\t}_3= -(0,1+\eta)$ (see Fig. \ref{fig_Haldane_strain});
\item $h_{AA}(\v{k})= M+ 2t_2\Big(\cos\Phi \cos(\v{k}\cdot\v{a}_1)+ \tfrac{t'_2}{t_2}\cos\Phi'\big( \cos(\v{k}\cdot\v{a}_2)+ \cos(\v{k}\cdot\v{a}_3) \big) \Big) -2t_2 \Big(\sin\Phi \sin(\v{k}\cdot\v{a}_1)+ \tfrac{t'_2}{t_2}\sin\Phi'\big( \sin(\v{k}\cdot\v{a}_2) -\sin(\v{k}\cdot\v{a}_3) \Big)$, in which $\v{a}_3=\v{a}_1+\v{a}_2$ and $M,\Phi$ are real parameters, $\Phi'=\Phi(1+\eta)$; $t_1,t_2>0$ are hopping strengths and $t'_2= t_2\big(1-2\eta+ \frac{b}{2d}\eta \big)$, with $b,d>0$;
\item $h_{BB}(\v{k})= -M+ 2t_2\Big(\cos\Phi \cos(\v{k}\cdot\v{a}_1)+ \tfrac{t'_2}{t_2}\cos\Phi'\big( \cos(\v{k}\cdot\v{a}_2)+ \cos(\v{k}\cdot\v{a}_3) \big) \Big) +2t_2 \Big(\sin\Phi \sin(\v{k}\cdot\v{a}_1)+ \tfrac{t'_2}{t_2}\sin\Phi'\big( \sin(\v{k}\cdot\v{a}_2) -\sin(\v{k}\cdot\v{a}_3) \Big)$;
\item $h_{AB}(\v{k})= t'e^{i\v{k}\cdot\v{\t}_3}+ t\Big(e^{i\v{k}\cdot\v{\t}_1}+ e^{i\v{k}\cdot\v{\t}_2} \Big)$, with $t'=t(1-2\eta)$.
\end{itemize}

\begin{figure}[ht]
\centering
\includegraphics[scale=0.75]{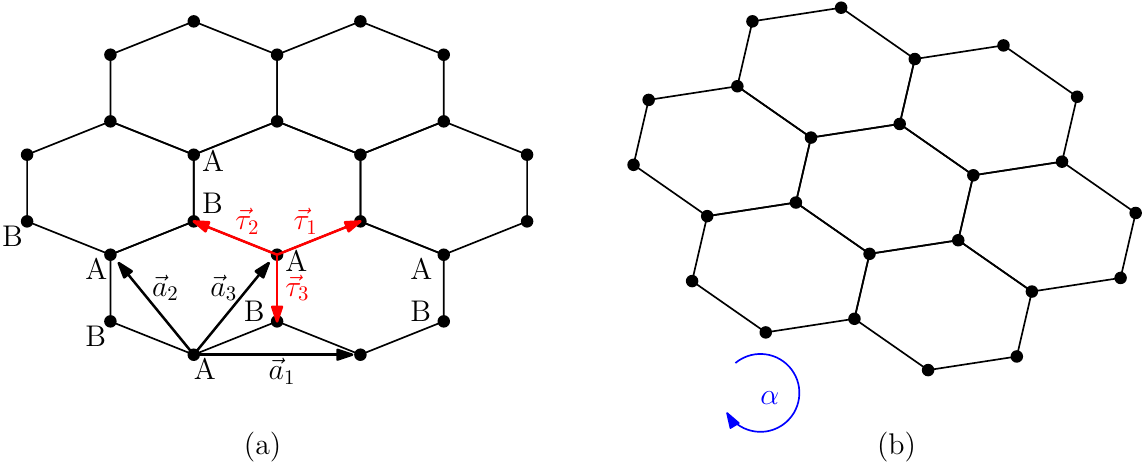}
    \caption{(a): lattice related to the strained Haldane model in Eq. \eqref{eqn:strain1}, as in \cite{Mannai2020}, where the labels A and B refer to the two sub-lattice; (b): lattice obtained from the original one by a clockwise rotation of angle $\a$.}
    \label{fig_Haldane_strain}
\end{figure}

The model has the same qualitative features as the Haldane model.
In particular there are two Fermi points \cite[Eq. (12)]{Mannai2020}: $ \v{k}_F^{\pm}= \pm \frac{2}{\sqrt{3}} \arccos \big(\frac{-t'}{2t}\big)(1,0)$, which are the only points at which the the two bands $\l_+$ and $\l_-$ can intersect; besides they intersect at $\v{k}_F^{\xi}$, for $\xi\in\{\pm\}$, only if 

\begin{equation*}
m_{\xi}:= M+ \xi2t_2 \left( \frac{2t'_2}{t_2}\sin \Phi' \sin \th - \sin\Phi\sin2\th\right)=0,
\end{equation*}

with $\th= \arccos\big(-\frac{t'}{2t} \big)$. For $\v{k}$ close to $\v{k}_F^{\xi}$ and $m_{\xi}$ close to $0$, the Bloch Hamiltonian takes the form \cite[Eq. (13)]{Mannai2020}:

\begin{equation}
\label{eqn:strain2}
\hat{H}^0(\v{k}_{F}^{\xi}+\v{k}')= \l_{+}(\v{k}_F^{\pm})+ \left(\begin{array}{cc}
    m_{\xi}+ \xi w_{0x} k'_1 & \xi w_xk'_1 - iw_yk'_2 \\
\xi w_{x} k'_1 +i w_yk'_2    & -m_{\xi}+ \xi w_{0x}k'_1
\end{array} \right)+ \mc{O}(|\v{k}'|^2 +m_{\xi}^2),
\end{equation}

where:

\begin{itemize}
\item $\l_+(\v{k}_F^{\xi})= \l_-(\v{k}_F^{\xi})= 2t_2 \Big(\cos\Phi\cos(\v{k}_F^+\cdot\v{a}_1)+ \frac{t'_2}{t_2}\cos\Phi' \big(\cos(\v{k}_F^+\cdot\v{a}_2)+ \cos(\v{k}_F^+\cdot\v{a}_3) \big) \Big)$ is independent of $M$;
\item $w_{0x}= -2\sqrt{3}\big(t_2\cos\Phi \sin(2\th) + t'_2\cos\Phi'\sin\th\big)$, $w_x= \frac{3}{2}t \big(1+ \frac{2}{3}\eta \big)$ and $w_y= \frac{3}{2}t\big(1-\frac{4}{3}\eta \big)$.
\end{itemize}

From Eq. \eqref{eqn:strain2} we see that the energy bands are given by \cite[Eq. (17)]{Mannai2020}:

\begin{equation}
\label{eqn:strain3}
\l_{\pm}(\v{k}_F^{\xi}+\v{k}') =\l_{+}(\v{k}_F^{\pm})+ \xi w_{0x} k'_1 \pm \sqrt{w_x^2 k_1'^2+ w_y^2 k_2'^2+ m_{\xi}^2}+ \mc{O}(|\v{k}'|^2+ m_{\xi}^2).
\end{equation}

\medskip

In order to investigate the lack of universality for $\s_{12}$, we study a simple generalization of the model just introduced, obtained from a 2D rotation of the lattice by an angle $\a$; equivalently we study the same model from a reference frame which is rotated by an angle $-\a$. More precisely, letting $R_{\a}= \left(\begin{array}{cc}
    \cos\a & -\sin\a \\
    \sin\a & \cos\a
\end{array}\right)$, we consider the lattice $\L:= \Span_{\mbb{Z}} \big\{R_{\a}^T\v{a}_1, R_{\a}^T\v{a}_2 \big\}$. It is straightforward to check that, according to Eq. \eqref{def:Hamiltonian}, the Bloch Hamiltonian of the model is $\hat{H}^0(R_{\a}\v{k})$, with $\hat{H}^0$ as in Eq. \eqref{eqn:strain1}. We study the Topological Phase Transition by fixing $\Phi>0$ such that $2t_2 \left( \frac{2t'_2}{t_2}\sin \Phi' \sin \th - \sin\Phi\sin2\th\right)>0$, so that, if $M\ge0$, the two energy bands can touch only at $\v{k}_{F}:= R_{\a}^T\v{k}_F^-$. We then set:

\begin{equation}
\hat H^0(\v{k};m):= \hat{H}^0(R_{\a}\v{k})\Big|_{M= m + 2t_2 \big( \frac{2t'_2}{t_2}\sin \Phi' \sin \th - \sin\Phi\sin2\th\big)}
\end{equation}

where the parameter $m$ implicitly corresponds to $m_-$. The energy bands $\l_{\pm}(\v{k}_{F}+\v{k}';m)$ are therefore given by

\begin{equation}
\begin{split}
&\l_{\pm}(\v{k}_{F}+\v{k}';m)- \l_{\pm}(\v{k}_{F};0)= -w_{0x}k'_1 \pm \sqrt{|A\v{k}'|^2 + m^2}+ \mc{O}(|\v{k}'|^2+m^2),
\end{split}
\end{equation}

with $A= \left(\begin{array}{cc}
w_x     & 0 \\
0     & w_y
\end{array}\right)R_{\a}$. We now apply Theorem \ref{thm:main}. Notice that the requirements of Assumption \ref{assum:H} are naturally fulfilled for $\eta$ small enough, since the energy bands are continuous w.r.t. $(\eta,\v{k},m)$ and for $\eta\to0$ the model reduces to the Haldane model \cite{H88}, which in turn verifies all the items of Assumption \ref{assum:H}. Therefore:

\begin{equation}
\s_{ij}(0)= \frac{1}{2}\Big(\s_{ij}(0^+)+ \s_{ij}(0^-)\Big)+ \frac{1}{16} \frac{(A^TA)_{ij}}{|\det A|}.
\end{equation}

Using that for $m>0$ small enough, $\s_{12}(m)=0,\s_{12}(-m)=+1$ \cite[Eq. (20)]{Mannai2020}, and using the explicit expression for the matrix $A$, we get that

\begin{align}
&\s_{12}(0)= \frac{1}{4\pi}+ \frac{1}{32}\sin 2\a \left(-\frac{3+2\eta}{3-4\eta}+ \frac{3-4\eta}{3+2\eta} \right)= \frac{1}{4\pi}- \frac{\eta}{8}\sin 2\a \left(1+ \mc{O}(\eta)\right),\\
&\s_{21}(0)= -\frac{1}{4\pi}+ \frac{1}{32}\sin 2\a \left(-\frac{3+2\eta}{3-4\eta}+ \frac{3-4\eta}{3+2\eta} \right)= -\frac{1}{4\pi}- \frac{\eta}{8}\sin 2\a \left(1+ \mc{O}(\eta)\right),
\end{align}

which are not quantized in units of $\frac{1}{4\pi}$, for $\eta\ne0$ and $\a$ not a multiple of $\frac{\pi}{2}$. Besides, one can also compute the value of the longitudinal conductivity:
\begin{align}
&\s_{11}(0)= \frac{1}{16} \left(\frac{3+2\eta}{3-4\eta} \cos^2\a+ \frac{3-4\eta}{3+2\eta}\sin^2\a \right)= \frac{1}{16} + \frac{\eta}{8}\cos2\a+ \mc{O}(\eta^2),\\
&\s_{22}(0)= \frac{1}{16} \left(\frac{3+2\eta}{3-4\eta} \sin^2\a+ \frac{3-4\eta}{3+2\eta}\cos^2\a \right)= \frac{1}{16} - \frac{\eta}{8}\cos 2\a+ \mc{O}(\eta^2),
\end{align}
which agree with \cite[Eq. (2.23)]{MPP26}\footnote{Note that in \cite{MPP26} the strain parameter was denoted by $\e$ instead of $\eta$.} for $\a=0$.

\section{Proof of the main theorem}
\label{sect:proof}

Throughout this section, we will always assume the validity of Assumptions \ref{assum:H0} and \ref{assum:H}, which are the hypotheses of Theorem \ref{thm:main}. The proof of Theorem \ref{thm:main} is structured in four subsequent parts:
\begin{enumerate}
    \item We reduce the $M$-band system to the semi-metallic two-band system, by using Kato-Nagy unitary $U_\o$, Eq. \eqref{eqn:KN}. 
    This strategy is similar to the one implemented in \cite[Section 8.3]{Cances}.
    
    \item The semi-metallic two-band system boils down to its linearization w.r.t.\ $\v{k}$ close to each Fermi point and $m$ close to zero. We follow  ideas analogous to those used in \cite[Section IV.C]{FGR25}. However, here we are able to reduce the computation to the relativistic model with $m=0$, showing that the two contributions from $m=\pm\e$ apriori cancel for the linearized system (cf. Remark \ref{rmk:1}.\ref{its:2}).
    \item We show that the linearized semi-metallic system is equivalent, via a unitary conjugation and a spatial rotation, to the Dirac Hamiltonian. We proceed similarly to \cite[Sections 8.3 - 8.4]{Cances}, even though we are able to include also the case of non-isotropic Dirac cones.   
    
    \item Finally, we perform the computation of the massless conductivity $\s(0)$ by relying on the explicit form of the Dirac system. For this step we combine ideas from \cite{GMP12}, where only the computation of the longitudinal conductivity of graphene was addressed, with  novel insights which are crucial for the adaptation to the present setting of generic semi-metallic systems. 
\end{enumerate}

As a starting point, we have that the current-current function can be computed via the fermionic Wick rule; denoting by $\bs{k}\equiv (k_0,\v{k})$ and $\bs{p}_0\equiv (p_0,\v{0})$, the following identity holds true:

\begin{equation}
\label{eq_3}
\hat{K}_{ij}(p_0;m)= - \int_{\mbb{R}\times\mbb{B}} \frac{d\bs{k}}{(2\pi)^3} \Tr_{\mbb{C}^M} \left\{ \Gamma_i(\v{k},\v{0};m) \hat{S}(\bs{k};m) \Gamma_j(\v{k},\v{0};m) \hat{S}(\bs{k}+\bs{p}_0;m) \right\}
\end{equation}

where $\hat{S}(\bs{k};m):= \big(-i k_0 + \hat{H}_0(\v{k};m) -\mu \big)^{-1}$ is the non-interacting two-point function.

\subsection{Reduction to two-band system}
\label{ssect:twobands}

Since the two semi-metallic bands $\l_\pm$ near the Fermi points are separated from the rest of the spectrum, the eigen-projector $\Pi$ associated to their eigenspace (see Eq. \eqref{def:Pi}) inherits the regularity of the Hamiltonian. This the content of the next result, whose proof is postponed in Appendix \ref{app:1_conical}.

\begin{lemma}
\label{lem:smoothP}
There exists $0<\kappa\leq m_0$, with $m_0$ as in Assumption \ref{assum:H}, such that for any $\o\in\{1,\dots,N\}$  the map $B_{\kappa}(\v{0})\times[-\kappa,\kappa]\ni (\v{k}';m)\mapsto\Pi(\v{k}_{F,\o}+\v{k}';m)$ is analytic.
\end{lemma}

Using the shorthand notation introduced after Eq. \eqref{eqn:R2}, namely:

\[
\hat{H}^0_{\o}(\v{k}';m)\equiv \hat{H}^0(\v{k}_{F,\o}+\v{k}';m),\qquad\Pi_{\o}(\v{k}';m)\equiv \Pi(\v{k}_{F,\o}+\v{k}';m),
\]

we fix a parameter $0<\d\le \kappa$ such that

\begin{enumerate}
\item Eq. \eqref{eq_6a} holds, namely $\frac{1}{2}C_{\star} \big(|\v{k}'|+|m|\big)\le \pm \big( \l_{\pm}(\v{k}'+\v{k}_{F,\o};m)- \mu \big) \le C'_{\star}\big(|\v{k}'|+|m|\big)$;
\item $B_{\d}(\v{0})\subset \mbb{B}$;\footnote{\label{footnote:B}Observe that in this way $\dist{\v{k}'}_{\mbb{B}}=|\v{k}'|$ for $\v{k}'\in B_{\d}(\v{0})$.}
\item $\d\le \frac{1}{2} \min_{1\le \o<\o'\le N} \dist{\v{k}_{F,\o}-\v{k}_{F,\o'}}_{\mbb{B}}$;
\item $\big\|\Pi_{\o}(\v{k}';m)- \Pi_{\o}(\v{0};0) \big\|\le \frac{1}{2}$ for $|\v{k}'|$, $|m|\le \d$ in view of Lemma \ref{lem:smoothP}.
\end{enumerate}

We can then construct the Kato-Nagy unitary:

\begin{equation}
\label{eqn:KN}
\begin{split}
U_{\o}(\v{k}';m):=& \Big(\Pi_{\o}(\v{k}';m)\Pi_{\o}(\v{0};0) + \big(\mathds{1}_M-\Pi_{\o}(\v{k}';m)\big)\big(\mathds{1}_M-\Pi_{\o}(\v{0};0)\big) \Big)\cdot\\
&\cdot\Big(\mathds{1}_M - \big(\Pi_{\o}(\v{k}';m)-\Pi_{\o}(\v{0};0)\big)^2\Big)^{-\frac{1}{2}}
\end{split}\end{equation}

which is analytic, since it has the same regularity as $\Pi_\o$ \cite[Chapter I, Section 4.6]{Kato} and maps any orthonormal basis of $\mathrm{Ran}( \Pi_{\o}(\v{0};0))$ into an orthonormal basis of $\mathrm{Ran} (\Pi_{\o}(\v{k}';m))$ (notice that $\hat{H}^0_\o(\v{0};0)\Pi_{\o}(\v{0};0)=\mu\Pi_{\o}(\v{0};0)$). Specifically, chosen any orthonormal basis $\big\{\phi_{\o;1}, \phi_{\o;2}\big\}$ of $\mathrm{Ran} (\Pi_{\o}(\v{0};0))$, we have that $\phi_{\o;1}(\v{k}';m):= U_{\o}(\v{k}';m) \phi_{\o;1}$ and $\phi_{\o;2}(\v{k}';m):= U_{\o}(\v{k}';m)\phi_{\o;2}$ form an orthonormal basis for $\mathrm{Ran}( \Pi_{\o}(\v{k}';m))$. Thus, we have the \emph{intertwining  property} of $U_\o$:
\begin{equation}
\label{eqn:2}
\Pi_\o(\v{k}';m)U_\o(\v{k}';m)=U_\o(\v{k}';m)\Pi_\o(\v{0};0)\quad \forall |\v{k}'|,|m|\le \d.
\end{equation}
Moreover, we choose an orthonormal basis $\{\phi_{\o;\rho}\}_{\rho=3}^{M}$ of $\mathrm{Ran}(\Pi_\o(\v{0};0)^\perp)$ consisting of eigenvectors for $\hat{H}^0_\o(\v{0};0)$. By the intertwining property of $U_\o(\v{k}';m)$, we have that $\{U_\o(\v{k}';m)\phi_{\o;\rho}\}_{\rho=3}^{M}$ is an orthonormal basis of $\mathrm{Ran}(\Pi_\o(\v{k}';m))^{\perp}$.

\begin{lemma}
\label{lemma_KN}
We have that
\begin{align*}
 &   U_\o^\dg(\v{k}';m)
\big(\hat{H}^0_\o(\v{k}';m)-\mu\big)U_\o(\v{k}';m)=
\\
&\sum_{1\leq \r,\r'\leq 2} \hat{h}_{\o;\r\r'}(\v{k}';m)\ket{\phi_{\o;\r}}\bra{\phi_{\o;\r'}}+\sum_{3\leq \r,\r'\leq M} \hat{r}_{\o;\r\r'}(\v{k}';m)\ket{\phi_{\o;\r}}\bra{\phi_{\o;\r'}},
\end{align*}

for every $|\v{k}'|,|m|\le \d$, where:

\begin{align}
\label{eqn:h1}
& \hat{h}_{\o;\r\r'}(\v{k}';m):= \bra{\phi_{\o;\r}} U_{\o}^{\dg}(\v{k}';m) \big( \hat{H}^0_{\o}(\v{k}';m)-\mu \big) U_{\o}(\v{k}';m) \ket{\phi_{\o;\r'}}, \; \r,\r'\in\{1,2\}\\
\label{eqn:h2}
&\hat{r}_{\o;\r\r'}(\v{k}';m):= \bra{\phi_{\o;\r}} U_{\o}^{\dg}(\v{k}';m) \big( \hat{H}^0_{\o}(\v{k}';m)-\mu \big) U_{\o}(\v{k}';m) \ket{\phi_{\o;\r'}}, \; \r,\r'\in\{3,...,M\}
\end{align}
are both analytic in $\v{k}',m$.
\end{lemma}

Its proof is deferred to Appendix \ref{app:1_conical}. Observe that $\hat{h}_{\o}$ captures the contributions from the \emph{quasi-massless} degrees of freedom (namely those related to the bands $\l_{\pm}(\v{k};m)$ close to the Fermi points), still preserving the regularity of $\hat{H}^0$. 

\medskip

Now we want to isolate the contributions to the current-current function, Eq. \eqref{eq_3}, from the quasi-massless degrees of freedom. We first isolate the modes close to the Fermi points, by considering, for every $\o\in\{1,\dots,M\}$ and $0< |\bs{k}'|\le \d$ (with $\bs{k}'\equiv (k_0,\v{k}')$),

\begin{align}
&\hat{s}_{\o;\r\r'}(\bs{k}';m):= \big(-ik_0 + \hat{h}_{\o}(\v{k}';m) \big)^{-1}_{\r\r'},\\
&\label{eqn:gamma}
\gamma_{\o;i;\r\r'}(\v{k}';m):=\bra{\phi_{\o;\r}(\v{k}';m)}\Gamma_i(\v{k}_{F,\o}+\v{k}',\v{0};m)\ket{ \phi_{\o;\r'}(\v{k}';m)},
\end{align}

with $\r,\r'\in\{1,2\}$. Consequently, we let:

\begin{equation}
\label{eqn:12}
\begin{split}
\hat{K}_{ij}^{(\mathrm{sing})}(p_0;m):=&- \sum_{\o=1}^N \int_{\mbb{R}\times\mbb{B}} \frac{d\bs{k}'}{(2\pi)^3} \chi(\d^{-1}|\bs{k}'|) \chi(\d^{-1}|\bs{k}'+\bs{p}_0|)\cdot \\
&\cdot \Tr_{\mbb{C}^2} \left\{ \gamma_{\o;i}(\v{k}';m) \hat{s}_{\o}(\bs{k}';m) \gamma_{\o;j}(\v{k}';m) \hat{s}_{\o}(\bs{k}'+\bs{p}_0;m)\right\},
\end{split}    
\end{equation}

with $\chi\in\mc{C}^{\infty}(\mbb{R}_+)$ a smooth cut-off function, such that $\chi\ge0$, $\chi(t)=1$ for $t\le\frac{1}{2}$ and $\chi(t)=0$ for $t\ge1$. It turns out that $\hat{K}_{ij}^{(\mathrm{sing})}$ captures all the non-vanishing contributions to the l.h.s.\ of Eq. \eqref{eqn:main}, as established by the following lemma.

\begin{lemma}
\label{lemma:K_reg}
Letting $\hat{K}_{ij}^{(\mathrm{reg})}(p_0;m):= \hat{K}_{ij}(p_0;m)- \hat{K}_{ij}^{(\mathrm{sing})}(p_0;m)$, it is true that $\hat{K}_{ij}^{(\mathrm{reg})}(\cdot,m)$ is $\mathscr{C}^1$ w.r.t.\ $p_0\in\mbb{R}$ and that $\de_{p_0}\hat{K}_{ij}^{(\mathrm{reg})}(p_0;\cdot)$ is $\mathscr{C}^1$ w.r.t.\ $m\in [-\d,\d]$.
\end{lemma}

The proof of Lemma \ref{lemma:K_reg} is given in Appendix \ref{app:1_rem}. In view of Lemma \ref{lemma:K_reg}, the well posedness of $\s_{ij}(m)$ is equivalent to showing the existence of the limit:

\begin{equation}
\label{eqn:13}
\s_{ij}^{(\mathrm{sing})}(m):= \lim_{p_0\rightarrow 0^-} \frac{1}{p_0} \big(\hat{K}_{ij}^{(\mathrm{sing})}(p_0;m)- \hat{K}_{ij}^{(\mathrm{sing})}(0;m) \big).\end{equation}

This is obvious in the case $m\ne0$, since one can easily check that $\hat{K}_{ij}^{(\mathrm{sing})}$ is smooth w.r.t.\ both $p_0$ and $m$. Instead, in the massless case $m=0$, the existence of the r.h.s.\ of Eq. \eqref{eqn:13} is more delicate and will be discussed in the next sections. 

Observe that, introducing the \emph{symmetric difference} of $\s_{ij}$:

\begin{equation}
\mf{D}\s_{ij}(\e):= \s_{ij}(0)- \frac{1}{2}\s_{ij}(\e)- \frac{1}{2}\s_{ij}(-\e),
\end{equation}

by using again Lemma \ref{lemma:K_reg} we have that

\begin{equation}
\label{eqn:21}
\mf{D}\s_{ij}(\e)= \mf{D}\s_{ij}^{(\mathrm{sing})}(\e) + \mc{O}(\e).
\end{equation}

\subsection{Reduction to linear system}
\label{ssect:linear}

We are now going to show that, for $\e\to 0^+$, the quantity $\mf{D}\s_{ij}^{(\mathrm{sing})}(\e)$ is dominated by contributions from momenta close to the Fermi points. Specifically, for each $\o\in\{1,\dots,N\}$ we are going to consider an ellipsoid of the form:

\begin{equation}
\label{eqn:D}
\mc{D}_{\e,\o}:=\big\{\bs{k}'\in\mbb{R}^3:\; k_0^2+|A_{\o}\v{k}'|^2\le \e\big\},
\end{equation}

(recall that $\v{k}'= \v{k}-\v{k}_{F,\o}$) for any $\e\le \e_0$ such that $\mc{D}_{\e_0,\o}\subseteq B_{\d/2}(\bs{0})$ for every $\o=1,\dots,N$. Besides, we are going to replace the $2\times2$ two-point function $\hat{s}_{\o}(\bs{k}';m)$ by its \emph{relativistic} counterpart, $\hat{s}^{(1)}_{\o}(\bs{k}';m)= \big(-i k_0+ \hat{h}^{(1)}_{\o}(\v{k}';m) \big)^{-1}$, with: 

\begin{equation}
\label{eqn:rel}
\hat{h}^{(1)}_{\o}(\v{k}';m):= \sum_{j=1}^2 k'_{j} \de_{k'_j}\hat{h}_{\o}(\v{0};0)+ m\de_m\hat{h}_{\o}(\v{0};0)
\end{equation}

being the linearized $2\times2$ Hamiltonian around the Fermi point $\v{k}_{F,\o}$ (recall that $\hat{h}_{\o}(\v{0};0)=0$ by construction).

\begin{remark}
\label{rmk:lin}
The linearized $2\times2$ Hamiltonian $\hat{h}^{(1)}_{\o}$ admits the following rewriting in terms of the $M\times M$ Bloch Hamiltonian $\hat{H}^0$:
\begin{equation*}
\hat{h}^{(1)}_{\o;\r\r'}(\v{k}';m)= \sum_{j=1,2} k'_j \langle \phi_{\o;\r}|\de_{k'_j}\hat{H}^0_{\o}(\v{0};0)| \phi_{\o;\r'}\rangle + m \langle\phi_{\o;\r}|\de_m\hat{H}^0_{\o}(\v{0};0)| \phi_{\o;\r'}\rangle,
\end{equation*}
for $\r\in\{1,2\}$, with $\phi_{\o;\r}$ introduced after Eq. \eqref{eqn:2}. The previous formula readily follows from the definitions of $\hat{h}^{(1)}_{\o}$ and $\hat{h}_{\o}$, Eqs. \eqref{eqn:rel} and \eqref{eqn:h1} respectively, together with the observation that $U_{\o}(\v{0};0)=\mathds{1}_M$ and $\big(\hat{H}^0_{\o}(\v{0};0)-\mu\big)|\phi_{\o;\r}\rangle=0$ for $\r\in\{1,2\}$. Moreover, since $\hat{h}_{\o}(\v{k}';m)$ has eigenvalues $\l_{\pm}(\v{k}_{F,\o}+\v{k}';m)-\mu$, it follows that $\hat{h}^{(1)}_{\o}(\v{k}';m)$ has eigenvalues:
\begin{equation*}
\v{c}_{\o}\cdot\v{k}'+ u_{\o}m \pm \sqrt{|A_{\o}\v{k}'|^2+ w_{\o}^2m^2+\v{k}'\!\cdot \v{b}_\o m},
\end{equation*}
with $\v{c}_{\o}, u_{\o}, A_{\o}, w_{\o}, b_{\o}$ as in Assumption \ref{assum:H}.\ref{it:3}.
\end{remark}

\begin{proposition}
\label{prop:linear}
Assuming that for every $m\in[-\d,\d]$ the function:

\begin{equation}
\label{eqn:14}
\begin{split}
\s^{(1)}_{ij;\o}(p_0;m):=& - \frac{1}{p_0} \int_{\mc{D}_{\e,\o}} \frac{d\bs{k}'}{(2\pi)^3} \Tr_{\mbb{C}^2} \Big\{\gamma_{\o;i}(\v{0};0) \hat{s}^{(1)}_{\o}(\bs{k}';m) \gamma_{\o;j}(\v{0};0) \cdot\\
&\cdot\Big(\hat{s}^{(1)}_{\o}(\bs{k}'+ \bs{p}_0;m)- \hat{s}^{(1)}_{\o}(\bs{k}';m) \Big) \Big\},
\end{split}
\end{equation}

admits a well defined limit as $p_0\to0^-$, then we have that $\s_{ij}^{(\mathrm{sing})}(m)$ is well defined for every $m\in[-\d,\d]$ and, for $\e\le \e_0$,

\begin{equation}
\label{eqn:linear}
\mf{D}\s_{ij}^{(\mathrm{sing})}(\e)= \sum_{\o=1}^N \lim_{p_0\to0^-}\s_{ij;\o}^{(1)}(p_0;0) + \mc{O}(\sqrt{\e}).\end{equation}

Moreover, it holds true that $\gamma_{\o;i}(\v{0};0)= -\de_{k'_i}\hat{h}_{\o}(\v{0};0)$.

\end{proposition}

\begin{remark}
\label{rmk:1}

Some comments follow.

\begin{enumerate}

\item\label{its:1} We stress that the assumption in the above proposition is non-trivial only in the massless case $m=0$. This will be proved by combing Proposition \ref{prop:sigma_Dirac}
 and Proposition \ref{prop:main}.
Instead, for $m\neq 0$ the existence of the limit as $p_0\to 0^-$ is a corollary of Dominated Convergence Theorem.  

\item\label{its:2} In view of equality \eqref{eqn:linear}, we notice that only the massless term $\s_{ij;\o}^{(1)}(p_0;0)$ contributes to  the symmetric difference of $\s_{ij}^{(\mathrm{sing})}$. As clear from the proof, this is ultimately a consequence of the oddness of the relativistic two-point function $\hat{s}_{\o}^{(1)}$.
\end{enumerate}

\end{remark}

\begin{proof}

It is possible to prove that (see Appendix \ref{app:1_rem}, Lemma \ref{lemma:linear}), letting: 

\begin{equation}
\label{def:sigma_sing}
\s^{(\mathrm{sing})}_{ij}(p_0;m):= \frac{1}{p_0}\Big(\hat{K}_{ij}^{(\mathrm{sing})}(p_0;m)- \hat{K}_{ij}^{(\mathrm{sing})}(0;m) \Big),
\end{equation}

one has that the difference:

\begin{equation}
\begin{split}
\mf{D}\s_{ij}^{(\mathrm{sing})}(p_0;\e)- \sum_{\o=1}^N \mf{D} \s^{(1)}_{ij;\o}(p_0;\e)
\end{split}
\end{equation}

admits a finite limit as $p_0\rightarrow0^-$, which is bounded in absolute value by const.$\sqrt{\e}$. Hence, using the assumption and the fact that $\s_{ij}^{(\mathrm{sing})}(m)$ is apriori well posed for $m\ne0$ (recall the discussion after Eq. \eqref{eqn:13}), we have that

\begin{equation*}
\begin{split}
\mf{D}\s_{ij}^{(\mathrm{sing})}(\e)
&= \lim_{p_0\to 0^-}\mf{D}\s_{ij}^{(\mathrm{sing})}(p_0;\e) = \sum_{\o=1}^N \lim_{p_0\to0^-}\mf{D} \s^{(1)}_{ij;\o}(p_0;\e) + \mc{O}(\sqrt{\e}),
\end{split}
\end{equation*}

implying that also $\s_{ij}^{(\mathrm{sing})}(0)$ is a well defined quantity. A key observation is that, for every $\o\in\{1,\dots,N\}$ and $i,j\in\{1,2\}$,

\begin{equation}
\label{eqn:15}
\lim_{p_0\to0^-} \Big(\s^{(1)}_{ij;\o}(p_0;\e)+ \s^{(1)}_{ij;\o}(p_0;-\e) \Big)=0,
\end{equation}

implying the validity of Eq. \eqref{eqn:linear}. 
Eq. \eqref{eqn:15} is a straightforward consequence of the oddness of the relativistic two-point function:

\begin{equation}
\label{eqn:odd}
\hat{s}^{(1)}_{\o}(-\bs{k}';-m)=- \hat{s}^{(1)}_{\o}(\bs{k}';m).
\end{equation}

Indeed, by the regularity in $\bs{k}'$ of $\hat{s}^{(1)}_{\o}(\bs{k}';\pm\e)$, we have that $\lim_{p_0\to0^-}\s^{(1)}_{ij;\o}(p_0;\pm\e)= \de_{p_0}\hat{K}^{(1)}_{ij;\o}(0;\pm\e)$,  where
\begin{equation*}
\hat{K}^{(1)}_{ij;\o}(p_0;m):= -\int_{\mc{D}_{\e,\o}} \frac{d\bs{k}'}{(2\pi)^3} \Tr_{\mbb{C}^2} \Big\{\gamma_{\o;i}(\v{0};0) \hat{s}^{(1)}_{\o}(\bs{k}';m) \gamma_{\o;j}(\v{0};0) \hat{s}^{(1)}_{\o}(\bs{k}'+ \bs{p}_0;m) \Big\}.
\end{equation*}
Using Eq. \eqref{eqn:odd} and the fact that the domain $\mc{D}_{\e,\o}$ is symmetric under $\bs{k}'\rightarrow-\bs{k}'$, we have that
\begin{equation*}
\begin{split}
\hat{K}^{(1)}_{ij;\o}(-p_0;\e)&= -\int_{\mc{D}_{\e,\o}} \frac{d\bs{k}'}{(2\pi)^3} \Tr_{\mbb{C}^2} \Big\{\gamma_{\o;i}(\v{0};0) \hat{s}^{(1)}_{\o}(\bs{k}';\e) \gamma_{\o;j}(\v{0};0) \hat{s}^{(1)}_{\o}(\bs{k}'- \bs{p}_0;\e) \Big\}\\
&= -\int_{\mc{D}_{\e,\o}} \frac{d\bs{k}'}{(2\pi)^3} \Tr_{\mbb{C}^2} \Big\{\gamma_{\o;i}(\v{0};0) \hat{s}^{(1)}_{\o}(-\bs{k}';\e) \gamma_{\o;j}(\v{0};0) \hat{s}^{(1)}_{\o}(-\bs{k}'- \bs{p}_0;\e) \Big\}\\
&= -\int_{\mc{D}_{\e,\o}} \frac{d\bs{k}'}{(2\pi)^3} \Tr_{\mbb{C}^2} \Big\{\gamma_{\o;i}(\v{0};0) \hat{s}^{(1)}_{\o}(\bs{k}';-\e) \gamma_{\o;j}(\v{0};0) \hat{s}^{(1)}_{\o}(\bs{k}'+ \bs{p}_0;-\e) \Big\}\\
&= \hat{K}^{(1)}_{ij;\o}(p_0;-\e).
\end{split}
\end{equation*}

Therefore:
\begin{equation*}
\begin{split}
\lim_{p_0\to0^-}\s^{(1)}_{ij;\o}(p_0;-\e)&= \de_{p_0}\hat{K}^{(1)}_{ij;\o}(p_0;-\e)\big|_{p_0=0}= \de_{p_0} \hat{K}^{(1)}_{ij;\o}(-p_0;\e)\big|_{p_0=0}\\
&= -\de_{p'_0}\hat{K}^{(1)}_{ij;\o}(p'_0;\e)\big|_{p'_0=0}= -\lim_{p_0\to 0^-}\s^{(1)}_{ij;\o}(p_0;\e),
\end{split}
\end{equation*}

proving Eq. \eqref{eqn:15}.\\

We are left with showing that $\gamma_{\o;i}(\v{0};0)= -\de_{k'_i}\hat{h}_{\o}(\v{0};0)$. By Eqs \eqref{eqn:gamma} and  \eqref{def_vertex},

\[\gamma_{\o;i;\r\r'}(\v{0};0)= - \bra{\phi_{\o;\r}} \Big( \de_{k_i} \hat{H}^0(\v{k}_{F,\o};0)+ i \big[\Delta_i, \hat{H}^0(\v{k}_{F,\o};0)\big] \Big) \ket{\phi_{\o;\r'}}, \]

where $(\Delta_i)_{\r_1\r_2}:= \d_{\r_1,\r_2} (\v{\d}_{\r_1})_i$. Notice that, since $\hat{H}^0(\v{k}_{F,\o};0) \ket{\phi_{\o;\r'}}= \l_{\pm}(\v{k}_{F,\o};0) \ket{\phi_{\o;\r'}}$, we have that $\bra{\phi_{\o;\r}} \big[\Delta_i, \hat{H}^0(\v{k}_{F,\o};0) \big] \ket{\phi_{\o;\r'}}=0$, hence:

\begin{equation}
\label{eqn:5}
\begin{split}
\gamma_{\o;i;\r\r'}(\v{0};0)= - \bra{\phi_{\o;\r}} \de_{k_i} \hat{H}^0(\v{k}_{F,\o};0) \ket{\phi_{\o;\r'}}.
\end{split}
\end{equation}

Furthermore, recalling the definition of $\hat{h}_{\o}$ in Lemma \ref{lemma_KN}, we have that

\[\begin{split}
\de_{k'_i}\hat{h}_{\o}(\v{0};0) =& \de_{k'_i} \bra{\phi_{\o;\r}} U_{\o}^{\dg}(\v{k}';0) \Big(\hat{H}^0_{\o}(\v{k}';0)- \mu \Big) U_{\o}(\v{k}';0) \ket{\phi_{\o;\r'}}\Big|_{\v{k}'=\v{0}}\\
=& \bra{\phi_{\o;\r}} \de_{k'_i}\hat{H}^0_{\o}(\v{0};0) \ket{\phi_{\o;\r'}}+ \bra{\phi_{\o;\r}} \de_{k'_i} U_{\o}^{\dg}(\v{0};0) \big(\hat{H}^0_{\o}(\v{0};0)-\mu\big) \ket{\phi_{\o;\r'}}\\
& +\bra{\phi_{\o;\r}} \big(\hat{H}^0_{\o}(\v{0};0)-\mu\big) \de_{k'_i}U_{\o}(\v{0};0) \ket{\phi_{\o;\r'}},
\end{split}\]

where we also used that $U_{\o}(\v{0};0)= \mathds{1}_{\mbb{C}^M}$. Recalling that $\big(\hat{H}^0_{\o}(\v{0};0)-\mu\big) \ket{\phi_{\o;\r}}= 0$ for every $\r\in\{1,2\}$, we see that $-\de_{k'_i}\hat{h}_{\o}(\v{0};0)$ equals the r.h.s. of Eq. \eqref{eqn:5}, thus proving the claim.

\end{proof}

\subsection{Reduction to the Dirac Hamiltonian}
\label{ssect:Dirac}
By Proposition \ref{prop:linear} we are left to analyze $\s_{ij}^{(1)}(p_0;m)$ with $m=0$, Eq. \eqref{eqn:14}.
We now show that, up to a spatial rotation and a unitary conjugation, we can reduce to the case where $\hat{h}^{(1)}_{\o}$ is the Dirac Hamiltonian. 

First of all, it is worth looking more closely at the conical structure of the energy bands, Eq. \eqref{eqn:bands1} for $m=0$. Observe that $|A_{\o}\v{k}'|^2= \v{k}'\cdot (A^T_{\o}A_{\o}\v{k}')$, with $A_{\o}^T$ the transpose of $A_{\o}$. By Eq. \eqref{eqn:bands2b} we must have that $A^T_{\o}A_{\o}\ge C_{\star}^2>0$, hence, being $A^T_{\o}A_{\o}$ symmetric, there exists $R_{\o}\in \mbb{SO}(2)$,\footnote{Hereafter $\mbb{SO}(2)$ is the group of special, orthogonal, real matrices, while $\mbb{SU}(2)$ denotes the group of special, unitary, complex matrices.}

\begin{equation}
\label{eqn:17}
R_{\o}= \left(\begin{array}{cc}
\cos\a_{\o}     & -\sin\a_{\o} \\
\sin\a_{\o}     & \cos\a_{\o}
\end{array}\right), \qquad \a_{\o}\in[0,2\pi),
\end{equation}

such that $R^T_{\o}A^T_{\o}A_{\o}R_{\o}= \mathrm{diag}\big(v_{\o;1}^2, v_{\o;2}^2 \big)$, for suitable $0<v_{\o;1}\leq v_{\o;2}$.

\begin{lemma}
\label{lemma:Dirac}
For every $\o\in\{1,...,N\}$ there exists $W_{\o}\in \mbb{SU}(2)$ such that the following is true. Letting $\s_1,\s_2,\s_3$  be the three Pauli matrices,

\begin{equation}
\label{eqn:22}
\hat{h}^{(D)}_{\o}(\v{k}'):= W_{\o}^{\dg}\, \hat{h}^{(1)}_{\o}(R_{\o}\v{k}';0)\, W _{\o}= \v{c'}_{\o}\cdot\v{k}' \mathds{1}_2+ v_{\o;2}k'_2\s_1+ v_{\o;1}k'_1\s_2, 
\end{equation}

where $\v{c'}_{\o}:= R^T_{\o}\v{c}_{\o}$, with the rotation matrix $R_{\o}$, the velocities $v_{\o;1}$ and $v_{\o;2}$ defined as in Eq. \eqref{eqn:17} and following lines.
\end{lemma}
Its proof is deferred to Appendix \ref{app:1_Dirac}.

\begin{proposition}
\label{prop:sigma_Dirac}
For every $\o\in\{1,...,N\}$, $i,j\in\{1,2\}$ and $p_0\ne0$, it holds true that
$$
\s^{(1)}_{\o;ij}(p_0;0)= \sum_{k,l=1,2} (R_{\o})_{ik} \,(R_{\o})_{jl} \,\s^{(D)}_{\o;kl}(p_0),
$$
where $R_{\o}$ is as in Eq. \eqref{eqn:17},

\begin{equation}
\label{eqn:Dirac}
\s^{(D)}_{\o;ij}(p_0):= -\frac{1}{p_0} \int_{\mc{A}_{\e,\o}} \frac{d\bs{k}'}{(2\pi)^3} \Tr_{\mbb{C}^2} \Big\{\gamma^{(D)}_{\o;i} \hat{s}^{(D)}_{\o}(\bs{k}') \gamma^{(D)}_{\o;j} \big( \hat{s}^{(D)}_{\o}(\bs{k}'+ \bs{p}_0)- \hat{s}^{(D)}_{\o}(\bs{k}') \big) \Big\},
\end{equation}

$\hat{s}^{(D)}_{\o}(\bs{k}'):= \big( -ik_0+ \hat{h}^{(D)}_{\o}(\v{k}')\big)^{-1}$, $\hat{h}^{(D)}_{\o}(\v{k}')$ as in Eq. \eqref{eqn:22}, $\gamma^{(D)}_{\o;i}:= -\big(\de_{k'_i} \hat{h}^{(D)}_{\o}\big)(\v{0})$ and $\mc{A}_{\e,\o}:= \big\{\bs{k}'\in\mbb{R}^3: \; k_0^2+ v_{\o;1}^2k_1'^2+ v_{\o;2}^2k_2'^2\le \e\big\}$.
\end{proposition}

Proposition \ref{prop:sigma_Dirac} is actually a corollary of a more general property, namely that the \emph{relativistic conductivity matrix} $\s^{(1)}_{\o}$ transforms as a tensor under spatial rotations, as established by the following lemma.

\begin{lemma}
\label{lem:sigmatensor}
For every $\o\in\{1,...,N\}$, orthogonal $2\times2$ matrix $O$, $m\in[-\delta,\delta]$ and $p_0\ne0$, we have that
\[
\s^{(1)}_{\o;ij}(p_0;m)= \sum_{k,l=1,2} O_{ki} \,O_{lj} \,\big[{\s}^{(1)}_{\o;kl}\big]_O(p_0;m)
\qquad \forall i,j\in \{1,2\},
\]
where: 
\begin{equation}
\label{eqn:16}
\begin{split}
\big[\s^{(1)}_{\o;ij}\big]_O(p_0;m)&:= -\frac{1}{p_0} \int_{{[\mc{D}_{\e,\o}]}_O}\frac{d\bs{k}'}{(2\pi)^3} \Tr_{\mbb{C}^2} \Big\{\big[\gamma_{\o;i}(\v{0};0)\big]_O \big[\hat{s}^{(1)}_{\o}(\bs{k}';m)\big]_O \big[\gamma_{\o;j}(\v{0};0)\big]_O\cdot\\
&\cdot \Big( \big[\hat{s}^{(1)}_{\o}(\bs{k}'+ \bs{p}_0;m)\big]_O- \big[\hat{s}^{(1)}_{\o}(\bs{k}';m)\big]_O \Big) \Big\},
\end{split}\end{equation}

$\big[\hat{s}^{(1)}_{\o}(\bs{k}';m)\big]_O:= \big( -ik_0+ [\hat{h}^{(1)}_{\o}(\v{k}';m)]_O\big)^{-1}$, $[\hat{h}^{(1)}_{\o}(\v{k}';m)]_O:= \hat{h}^{(1)}_{\o}(O^T\v{k}';m)$, $[\gamma_{\o;i}(\v{0};0)]_O := -\big(\de_{k'_i} [\hat{h}^{(1)}_{\o}]_O\big)(\v{0};0)$ and $[\mc{D}_{\e,\o}]_O:= \big\{(k_0,O\v{k}'):\, \bs{k}' \in \mc{D}_{\e,\o}\big\}$, with $\mc{D}_{\e,\o}$ as in Eq. \eqref{eqn:D}.
\end{lemma}

Observe that by applying Lemma \ref{lem:sigmatensor} with $O=R_{\o}^T$, we find that

\[\s^{(1)}_{\o;ij}(p_0;m)= \sum_{k,l=1,2}(R_{\o})_{ik} (R_{\o})_{jl} [\s^{(1)}_{\o;kl}]_{R^T_{\o}}(p_0;m),\]

which holds true in particular for $m=0$. The claim of Proposition \ref{prop:sigma_Dirac} follows by performing, within the trace in the r.h.s.\ of Eq. \eqref{eqn:16}, the conjugation by the unitary $W_{\o}$ as in Lemma \ref{lemma:Dirac}, and noting that $[\mc{D}_{\e,\o}]_{R^T_{\o}}= \mc{A}_{\e,\o}$, so that $[\s^{(1)}_{\o;kl}]_{R^T_{\o}}(p_0;0)= {\s}^{(D)}_{\o;kl}(p_0)$ by applying Eq. \eqref{eqn:22}.

\begin{proof}[Proof of Lemma \ref{lem:sigmatensor}.]
From the very definition of $\s^{(1)}_{\o;ij}$ in Eq. \eqref{eqn:14}, it suffices to analyze:
\begin{align*}
 &\int_{\mc{D}_{\e,\o}} \frac{d\bs{k}'}{(2\pi)^3} \Tr_{\mbb{C}^2} \Big\{\gamma_{\o;i}(\v{0};0) \hat{s}^{(1)}_{\o}(\bs{k}';m) \gamma_{\o;j}(\v{0};0)\Big(\hat{s}^{(1)}_{\o}(\bs{k}'+ \bs{p}_0;m)- \hat{s}^{(1)}_{\o}(\bs{k}';m) \Big) \Big\}=\\
 &\int_{[\mc{D}_{\e,\o}]_O} \frac{d\bs{q}}{(2\pi)^3} \Tr_{\mbb{C}^2} \Big\{\gamma_{\o;i}(\v{0};0) [\hat{s}^{(1)}_{\o}(\bs{q};m)]_O \gamma_{\o;j}(\v{0};0)\Big([\hat{s}^{(1)}_{\o}(\bs{q}+ \bs{p}_0;m)]_O- [\hat{s}^{(1)}_{\o}(\bs{q};m)]_O \Big) \Big\},
\end{align*}

where we have implemented the change of variable $\bs{q}:=(k_0,O\v{k}')$. The thesis then follows from the basic observation that

\[ [\gamma_{\o;i}(\v{0};0)]_O= -\de_{q_i} \hat{h}^{(1)}_{\o}(O^T\v{q};0)\big|_{\v{q}=\v{0}}= -\sum_{l=1,2} O_{il} \de_{k'_l} \hat{h}_{\o}(\v{0};0), \]

which, in view of the identity $\gamma_{\o;i}(\v{0};0)= -\de_{k'_i}\hat{h}_{\o}(\v{0};0)$ from Proposition \ref{prop:linear}, implies: $\v{\gamma}_{\o}(\v{0};0)= O^T[\v{\gamma}_{\o}(\v{0};0)]_O$. 
\end{proof}

\subsection{Computation of the conductivity matrix}
\label{ssect:comp}

In view of the previous subsections, the proof of Theorem \ref{thm:main} is now reduced to the explicit computation of $\s^{(D)}_{\o;ij}$.

\begin{proposition}
\label{prop:main}
For every $\o\in\{1,\dots,N\}$ and $i,j\in\{1,2\}$, it holds true that

\[ \lim_{p_0\to0^-}\s_{\o;ij}^{(D)}(p_0) = \frac{v_{\o;i}^2}{16 v_{\o;1} v_{\o;2}} \d_{i,j}. \]
\end{proposition}

Before proving the above proposition, we combine the previous results to complete the proof of the main theorem.

\paragraph{Proof of Theorem \ref{thm:main}.}  As a first step, observe that  Propositions \ref{prop:sigma_Dirac} and \ref{prop:main} imply that

\begin{equation}
\label{eqn:18}
\begin{split}
\lim_{p_0\to0^-}\s_{\o;ij}^{(1)}(p_0;0)&=  \sum_{k,l=1,2} (R_{\o})_{ik} (R_{\o})_{jl} \lim_{p_0\to0^-} \s^{(D)}_{\o;kl}(p_0)\\
&= \frac{1}{16 v_{1;\o} v_{2;\o}} \big(R_{\o} \,\mathrm{diag}(v_{\o;1}^2, v_{\o;2}^2) \,R_{\o}^T \big)_{ij}\\
&= \frac{1}{16|\det A_{\o}|} (A_{\o}^TA_{\o})_{ij},
\end{split}
\end{equation}

where we used that, by construction, $R^T_{\o}A^T_{\o}A_{\o}R_{\o}= \mathrm{diag}\big(v_{\o;1}^2, v_{\o;2}^2 \big)$, as discussed at the beginning of Subsection \ref{ssect:Dirac}. In force of Eq. \eqref{eqn:18}, Proposition \ref{prop:linear} implies\footnote{Eq. \eqref{eqn:18} in particular establishes that the limit $p_0\to0^-$ of $\s^{(1)}_{\o;ij}(p_0;m)$ exists finite for $m=0$; besides, as mentioned in Remark \ref{rmk:1}, this is trivially true in the case $m\ne0$.} that $\s_{ij}^{(\mathrm{sing})}(m)$ is well defined (and so is $\s_{ij}(m)$) and

\begin{equation}
\label{eqn:19}
\mf{D}\s_{ij}^{(\mathrm{sing})}(\e)= \sum_{\o=1}^N \frac{1}{16|\det A_{\o}|} (A_{\o}^TA_{\o})_{ij} + \mc{O}(\sqrt{\e}).
\end{equation}

Finally, in view of Eq. \eqref{eqn:21}, we have that

\begin{equation}
\begin{split}
\mf{D}\s_{ij;\o}(\e)= \mf{D}\s^{(\mathrm{sing})}_{ij;\o}(\e) + \mc{O}(\e)= \sum_{\o=1}^N \frac{1}{16|\det A_{\o}|} (A_{\o}^TA_{\o})_{ij} + \mc{O}(\sqrt{\e}).   
\end{split}
\end{equation}

Taking the limit $\e\to0^+$ at both sides proves the claim of Theorem \ref{thm:main}.

\paragraph{Proof of Proposition \ref{prop:main}.}

In order to simplify the notation, throughout this proof, we drop the dependence upon the valley index $\o$ in all the quantities involved.
Conversely, it is convenient to make the dependence of $\s^{(D)}$ upon $\v{v}$ and $\v{c'}$ explicit: $\s^{(D)}_{ij}(p_0)\equiv \s^{(D)}_{ij}(p_0;\v{v},\v{c'})$. 

As a first step, we notice that, by performing the change of variable $(k'_1,k'_2)\to (q_1,q_2):= (v_1k'_1, v_2k'_2)$ in the integral at the r.h.s.\ of Eq. \eqref{eqn:Dirac}, we get that

\begin{equation}
\label{eqn:23}
\begin{split}
\s^{(D)}_{ij}(p_0;\v{v},\v{c'})& = \frac{v_iv_j}{v_1v_2} \s^{(D)}_{ij}\big(p_0;\v{1},\v{c''} \big),
\end{split}
\end{equation}

where $\v{1}\equiv (1,1)$, $\v{c''}\equiv (\frac{c'_1}{v_1},\frac{c'_2}{v_2})$. As a next step, it is convenient to extend the integral in the definition of $\s^{(D)}_{ij}\big(p_0;\v{1},\v{c''} \big)$ (see Eq. \eqref{eqn:Dirac}) from the ball $B_{\sqrt{\e}}(\bs{0})$ to the infinite cylinder:

\begin{equation}
\mc{C_{\e}}:= \big\{\bs{k}'\in\mbb{R}^3: \; |\v{k}'|\le \sqrt{\e}\big\}.
\end{equation}

To this purpose, we introduce the function:

\[\widehat{\d K}^{(D)}_{ij}(p_0):=- \int_{\mc{C}_{\e}\setminus B_{\sqrt\e}(\bs{0})} \frac{d\bs{k}'}{(2\pi)^3} \mathrm{Tr}_{\mbb{C}^2}\Big\{\gamma^{(D)}_i \hat{s}^{(D)}(\bs{k}') \gamma_j^{(D)} \hat{s}^{(D)}\big( \bs{k}'+ \bs{p}_0\big) \Big\}\Big|_{(\v{1},\v{c''})},\]

which, as we shall discuss below, is both $\mathscr{C}^1$ and even w.r.t.\ $p_0$, for $|p_0|\le \frac{\sqrt\e}{4}$. In this way $\de_{p_0} \widehat{\d K}^{(D)}_{ij}(0)=0$ and so:

\begin{equation}
\lim_{p_0\to 0^-} \s_{ij}^{(D)}(p_0;\v{1},\v{c''})= \lim_{p_0\to 0^-} \tilde{\s}_{ij}^{(D)}(p_0;\v{1},\v{c''}),
\end{equation}

where: 

\begin{equation}
\label{def:sigma_tilde}
\begin{split}
\tilde{\s}_{ij}^{(D)}(p_0;\v{1},\v{c''})&:= \s_{ij}^{(D)}(p_0;\v{1},\v{c''})+ \frac{1}{p_0}\Big( \widehat{\d K}^{(D)}_{ij}(p_0)- \widehat{\d K}^{(D)}_{ij}(0) \Big)\\
&\equiv -\frac{1}{p_0} \int_{\mc{C}_{\e}} \frac{d\bs{k}'}{(2\pi)^3} \Tr_{\mbb{C}^2} \Big\{\gamma_i^{(D)} \hat{s}^{(D)}(\bs{k}') \gamma_j^{(D)} \big(\hat{s}^{(D)}(\bs{k}'+\bs{p}_0)- \hat{s}^{(D)}(\bs{k}') \big) \Big\}\Big|_{(\v{1},\v{c''})}.
\end{split}
\end{equation}

The evenness of $\widehat{\d K}^{(D)}_{ij}$ follows immediately from the oddness of $\hat{s}^{(D)}$ and the symmetry of the domains $\mc{C}_{\e}$ and $B_{\sqrt{\e}}(\bs{0})$. Its regularity instead follows from purely dimensional considerations. Indeed, observe that

\[\begin{split}
&\frac{1}{q_0} \Big( \widehat{\d K}^{(D)}_{ij}(p_0+ q_0)- \widehat{\d K}^{(D)}_{ij}(p_0) \Big)=\\
& -  \int_{\mc{C}_{\e}\setminus B_{\sqrt{\e}}(\bs{0})} \frac{d\bs{k}'}{(2\pi)^3} \int_0^1 dt \Tr_{\mbb{C}^2} \Big\{\gamma^{(D)}_i \hat{s}^{(D)}(\bs{k}') \gamma^{(D)}_j \de_{k_0}\hat{s}^{(D)}\big(\bs{k}'+ \bs{p}_0+ t\bs{q}_0\big) \Big\}\Big|_{(\v{1},\v{c''})}.
\end{split}\]

Using also that $|\hat{s}^{(D)}(\bs{k}')|\le \frac{C}{|\bs{k}'|}$, $|\de_{k_0}\hat{s}^{(D)}(\bs{k}')|\le \frac{C}{|\bs{k}'|^2}$, for a suitable constant $C>0$, and noting that for $|p_0|,|q_0|\le\frac{\sqrt\e}{4}$, $|\bs{k}'+ \bs{p}_0+ t\bs{q}_0|\ge \frac{1}{2}|\bs{k}'|$, we see that

\begin{equation*}
\Big|\Tr_{\mbb{C}^2} \Big\{\gamma^{(D)}_i \hat{s}^{(D)}(\bs{k}') \gamma^{(D)}_j \de_{k_0}\hat{s}^{(D)}\big(\bs{k}'+ \bs{p}_0+ t\bs{q}_0\big) \Big\}\Big|\le \frac{C'}{|\bs{k}'|^3},
\end{equation*}

for some $C'>0$, and the r.h.s. is in $L^1\big(\mc{C}_{\e}\setminus B_{\sqrt{\e}}(\bs{0})\big)$. Hence, by an application of the Dominated Convergence Theorem, we find that

\[\begin{split}
&\lim_{q_0\to0} \frac{1}{q_0} \Big( \widehat{\d K}^{(D)}_{ij}(p_0+ q_0)- \widehat{\d K}^{(D)}_{ij}(p_0) \Big)=\\
& -  \int_{\mc{C}_{\e}\setminus B_{\sqrt{\e}}(\bs{0})} \frac{d\bs{k}'}{(2\pi)^3} \Tr_{\mbb{C}^2} \Big\{\gamma^{(D)}_i \hat{s}^{(D)}(\bs{k}') \gamma^{(D)}_j \de_{k_0}\hat{s}^{(D)}\big(\bs{k}'+ \bs{p}_0\big) \Big\}\Big|_{(\v{1},\v{c''})},
\end{split}\]

which is also continuous w.r.t.\ $p_0$.  All in all, for the moment we have shown that

\begin{equation}
\label{eqn:44}
\lim_{p_0\to0^-}\s^{(D)}_{ij}(p_0;\v{v},\v{c'})= \frac{v_iv_j}{v_1v_2} \lim_{p_0\to 0^-} \tilde{\s}^{(D)}_{ij}(p_0;\v{1},\v{c''}),
\end{equation}

so we are left with analyzing $\tilde{\s}^{(D)}_{ij}(p_0;\v{1},\v{c''})$. We shall show that 

\begin{equation}
\label{eqn:45}
\tilde{\s}^{(D)}_{ij}(p_0;\v{1},\v{c''})= \d_{i,j} \tilde{\s}^{(D)}_{11}(p_0;\v{1},\v{0}).
\end{equation}

Notice that the computation of $\lim_{p_0\to0^-}\tilde{\s}^{(D)}_{11}(p_0;\v{1},\v{0})$ is performed in \cite[Appendix A]{GMP12}: it is half of the longitudinal conductivity of non-interacting graphene  and equals $\frac{1}{16}$. In this way, in view of Eq. \eqref{eqn:44}, the claim of Proposition \ref{prop:main} follows. In order to prove Eq. \eqref{eqn:45}, we first show that we can set $\v{c''}=\v{0}$ within the two-point functions $\hat{s}^{(D)}$ appearing in the r.h.s.\ of Eq. \eqref{def:sigma_tilde}. It suffices to show that

\begin{equation}
\label{eqn:47}
\begin{split}
&\int_{\mc{C}_{\e}} \frac{d\bs{k}'}{(2\pi)^3} \Tr_{\mbb{C}^2} \Big\{\gamma_i^{(D)} \hat{s}^{(D)}(\bs{k}') \gamma_j^{(D)} \hat{s}^{(D)}(\bs{k}'+\bs{p}_0) \Big\}=\\
&\int_{\mc{C}_{\e}} \frac{d\bs{k}'}{(2\pi)^3} \Tr_{\mbb{C}^2} \Big\{\gamma_i^{(D)} \hat{s}^{(D)}(\bs{k}')\big|_{\v{c''}=\v{0}} \,\gamma_j^{(D)} \hat{s}^{(D)}(\bs{k}'+\bs{p}_0)\big|_{\v{c''}=\v{0}} \Big\},
\end{split}
\end{equation}

for every $|p_0|\le \frac{\sqrt\e}{4}$, with the understanding that $\v{v}=\v{1}$. In fact, denoting:

\begin{equation*}
f_{\v{k}',p_0}(k_0):= \Tr_{\mbb{C}^2} \Big\{\gamma_i^{(D)} \hat{s}^{(D)}(\bs{k}') \gamma_j^{(D)} \hat{s}^{(D)}(\bs{k}'+\bs{p}_0) \Big\},
\end{equation*}

for every $0<|\v{k}'|<\sqrt\e$ and $|p_0|\le\frac{\sqrt\e}{4}$ fixed, by Fubini's Theorem we can write:

\begin{equation}
\int_{\mc{C}_{\e}} \frac{d\bs{k}'}{(2\pi)^3} \Tr_{\mbb{C}^2} \Big\{\gamma_i^{(D)} \hat{s}^{(D)}(\bs{k}') \gamma_j^{(D)} \hat{s}^{(D)}(\bs{k}'+\bs{p}_0) \Big\}= \int_{0<|\v{k}'|\le\sqrt{\e}} \frac{d\v{k}'}{(2\pi)^2} \int_{\mbb{R}} \frac{dk_0}{2\pi} f_{\v{k}',p_0}(k_0).
\end{equation}

Observe that $f_{\v{k}',p_0}$ is analytic away from the poles

\begin{equation*}
\left\{\begin{array}{cc}
k_0=& -i \big( \v{c''}\cdot\v{k}' \pm |\v{k}'| \big),  \\
k_0=& -i \big( \v{c''}\cdot\v{k}' \pm |\v{k}'| \big) -p_0,
\end{array}\right.
\end{equation*}

whose imaginary parts are such that $  |\v{k}'| \pm\v{c''}\cdot\v{k}' > 0$ in view of Eq. \eqref{eqn:bands2b}. Since $f_{\v{k}',p_0}(k_0)$ decays as $\frac{1}{k_0^2}$ for $k_0$ large and since $f_{\v{k}',p_0}$ is analytic on the strip between the real axis and $\big\{\mathrm{Im}(k_0)=-\v{c''}\cdot\v{k}' \big\}$, by the Cauchy Theorem we get that

\[
\int_{\mbb{R}} \frac{d k_0}{2\pi} f_{\v{k}',p_0}(k_0)= \int_{\mbb{R}-i\v{c''}\cdot\v{k}'} \frac{d k_0}{2\pi} f_{\v{k}',p_0}(k_0)= \int_{\mbb{R}} \frac{d k_0}{2\pi} f_{\v{k}',p_0}(k_0- i\v{c''}\cdot\v{k}').
\]

Eq. \eqref{eqn:47} follows by observing that 

\begin{equation*}
f_{\v{k}',p_0}(k_0-i\v{c''}\cdot\v{k}')= \Tr_{\mbb{C}^2} \Big\{\gamma_i^{(D)} \hat{s}^{(D)}(\bs{k}')\big|_{\v{c''}=\v{0}}\, \gamma_j^{(D)} \hat{s}^{(D)}(\bs{k}'+\bs{p}_0)\big|_{\v{c''}=\v{0}} \Big\}.
\end{equation*}

Now, recalling that for $\v{v}=\v{1}$, $\gamma_1^{(D)}= -\s_2- c''_1\mathds{1}_2$ and $\gamma_2^{(D)}=-\s_1 - c''_2\mathds{1}_2$ and by applying Lemma \ref{lemma:kcal} below, we find that

\begin{equation*}
\begin{split}
&\int_{\mc{C}_{\e}} \frac{d\bs{k}'}{(2\pi)^3} \Tr_{\mbb{C}^2} \Big\{\gamma_i^{(D)} \hat{s}^{(D)}(\bs{k}')\big|_{\v{c''}=\v{0}} \,\gamma_j^{(D)} \hat{s}^{(D)}(\bs{k}'+\bs{p}_0)\big|_{\v{c''}=\v{0}} \Big\}=\\
& \hat{\mc{K}}_{ij}(p_0)- c''_i \hat{\mc{K}}_{0j}(p_0) - c''_j \hat{\mc{K}}_{i0}(p_0) + c''_ic''_j\hat{\mc{K}}_{00}(p_0)= \d_{i,j}\hat{\mc{K}}_{11}(p_0),
\end{split}
\end{equation*}

(with $\hat{\mc{K}}_{\tau\nu}$ as in Eq. \eqref{eqn:kcal} below) which immediately implies Eq. \eqref{eqn:45}.

\begin{lemma}
\label{lemma:kcal}
Defining, for every $|p_0|\le \frac{\sqrt\e}{4}$ and $\nu,\tau\in\{0,1,2\}$,

\begin{equation}
\label{eqn:kcal}
\hat{\mc{K}}_{\t\nu}(p_0):= -\int_{\mc{C}_{\e}} \frac{d\bs{k}'}{(2\pi)^3} \Tr_{\mbb{C}^2} \Big\{\gamma_{\tau}^{(D)} \hat{s}^{(D)}(\bs{k}') \gamma_{\nu}^{(D)} \hat{s}^{(D)}(\bs{k}'+\bs{p}_0) \Big\}\Big|_{(\v{v},\v{c})=(\v{1},\v{0})},
\end{equation}

with $\gamma_0^{(D)}:=\mathds{1}_2$, it holds true that
\[
\hat{\mc{K}}_{\tau\nu}=\delta_{\tau,\nu}(1-\delta_{\tau,0})\hat{\mc{K}}_{11}.
\]
In particular, $\hat{\mc{K}}$ vanishes if either $\tau$ or $\nu$ is zero, or $\tau\neq \nu$.
\end{lemma}

The proof of this lemma is deferred to Appendix \ref{app:1_Dirac}.

\paragraph{Acknowledgments.} 
We gratefully acknowledge M. Porta for enlightening advice and valuable discussions.  We are thankful to A. Marrazzo for sharing with us a different perspective about topological phases and suggesting physically relevant models. We sincerely appreciate discussions with D. Monaco and G. Panati, willing to share with us their geometrical insights.

This work was supported by the National Group of
Mathematical Physics (GNFM–INdAM) within the project \virg{Progetto Giovani GNFM 2025} -- CUP E5324001950001. G.~M.\ gratefully acknowledges financial support from the European Research Council through the ERC CoG UniCoSM, grant agreement n.724939. S.~F.\ acknowledges financial support from the ERC-StG MaMBoQ, grant agreement n.802901, and the  MUR, PRIN 2022 project MaIQuFi cod. 20223J85K3.

\vspace{0.5cm}

\noindent\textbf{Statements and Declarations.}\\ 
\textbf{Competing Interests.} All authors declare that they have no conflicts of interest to disclose.

\vspace{0.5cm}

\noindent\textbf{Data Availability.} Data sharing is not applicable to this article as no datasets were generated or
analyzed during the current study.

\appendix

\section{Systems with a rotation symmetry}
\label{app:rot}

In this appendix we will first discuss the proof of Corollary \ref{cor:rotation}, which is a direct application of Theorem \ref{thm:main} to systems with an emergent rotation symmetry. Secondly, in Lemma \ref{lemma:R_exact} we will show that a rotation symmetry which is exact is also emergent (see Assumptions \ref{assum:R} and \ref{assum:ER}).

\begin{proof}[Proof of Corollary \ref{cor:rotation}]
In force of Theorem \ref{thm:main}, the proof of Corollary \ref{cor:rotation} reduces to showing that the $2\times2$ matrix $|A_{\o}|^2\equiv A_{\o}^TA_{\o}$ is proportional to the identity, for every $\o\in\{1,\dots,N\}$. As a convenient starting point, we have the rephrasing of Eq. \eqref{eqn:35} in terms of the $2\times2$ Hamiltonian $\hat{h}_{\o}$, defined in Eq. \eqref{eqn:h1}. According to Lemma \ref{lemma_2by2} below, it turns out that, for every $\o\in\{1,\dots,N\}$,

\begin{equation}
\label{eqn:38}
\mc{S}^{\dg}_{\o} \,\hat{h}_{\o}(\v{k}';m) \,\mc{S}_{\o}= \hat{h}_{\o}(R_{\o}\v{k}';m) + \mc{O}(|\v{k}'|^2+m^2),
\end{equation}

for a suitable $2\times2$ unitary matrix $\mc{S}_{\o}$, where $R_\o$ is the rotation matrix introduced in Assumption \ref{assum:R}. Recalling that $\hat{h}_{\o}(\v{k}';m)$ has eigenvalues $\l_{\pm}(\v{k}_{F,\o}+ \v{k}';m)-\mu$, from Eq. \eqref{eqn:38} it follows that

\begin{equation}
\label{eqn:39}
\l_{\pm}(\v{k}_{F,\o}+\v{k}';m)= \l_{\pm}(\v{k}_{F,\o}+R_{\o}\v{k}';m) + \mc{O}\big(|\v{k}'|^2+m^2 \big), \qquad \forall\o\in\{1,\dots,M\},
\end{equation}

which is readily checked in view of the Pauli-matrices decomposition of $\hat{h}_{\o}$. Combining Eq. \eqref{eqn:39} with Eq. \eqref{eqn:bands1}, we find that 

\begin{equation}
\begin{split}
&\v{c}_{\o}\cdot\v{k}'+ u_{\o}m \pm \sqrt{|A_{\o}\v{k}'|^2+ w_{\o}^2m^2+\v{k}'\!\cdot \v{b}_\o m} =\\
&\v{c}_{\o}\cdot (R_{\o}\v{k}')+ u_{\o}m \pm \sqrt{|A_{\o}R_{\o}\v{k}'|^2+ w_{\o}^2m^2+(R_{\o}\v{k}')\!\cdot \v{b}_\o m} + \mc{O}\big(|\v{k}'|^2+m^2 \big),
\end{split}
\end{equation}

which implies:

\begin{align}
&(\v{c}_{\o}- R_{\o}^T\v{c}_{\o})\cdot\v{k}'=0,\\
& \v{k}'\cdot \big(A_{\o}^TA_{\o}- R_{\o}^TA_{\o}^TA_{\o}R_{\o} \big)\v{k}' + m(\v{b}_{\o}- R_{\o}^T\v{b}_{\o})\cdot\v{k}'=0.
\end{align}

From the arbitrariness of $(\v{k}',m)$ small enough, we must have that $R_{\o}^T\v{c}_{\o}=\v{c}_{\o}$ and $R_{\o}^T\v{b}_{\o}=\v{b}_{\o}$, but since $R_{\o}^T$ clearly cannot have 1 as eigenvalue, we must have $\v{c}_{\o}=\v{b}_{\o}=\v{0}$. Besides, since both $A_{\o}^TA_{\o}$ and $R_{\o}^TA_{\o}^TA_{\o}R_{\o}$ are symmetric matrices, we must also have:

\begin{equation}
\label{eqn:40}
R_{\o}^TA_{\o}^TA_{\o}R_{\o} = A_{\o}^TA_{\o}.
\end{equation}

As we are going to show, Eq. \eqref{eqn:40} is sufficient to conclude that $A_{\o}^TA_{\o}$ is proportional to $\mathds{1}_2$. Indeed, assuming that 

\begin{equation*}
A^T_{\o}A_{\o}= \left( \begin{array}{cc}
    a & b \\
    b & d
\end{array}\right),
\end{equation*}

with $a,b,d\in\mbb{R}$ (the dependence upon $\o$ is suppressed to simplify the notation), Eq. \eqref{eqn:40} implies that

\begin{equation}
\label{eqn:41}
\left\{\begin{array}{cc}
a= a \cos\a_{\o}^2 + 2b \sin\a_\o \cos\a_\o + d \sin^2\a_\o,  \\
b= (d-a)\sin\a_\o \cos\a_\o + b(\cos^2\a_\o -\sin^2\a_\o),\\
d= a\sin^2\a_\o -2b \sin\a_\o \cos\a_\o +d \cos^2\a_\o.
\end{array}\right.
\end{equation}

Taking the difference of the first and the third of Eq. \eqref{eqn:41}, we find that

\begin{equation*}
(d-a)\big(1-\cos^2\a_\o+ \sin^2\a_\o\big)=-4b \sin\a_\o \cos\a_\o \Rightarrow d-a = -\tfrac{2\cos\a_\o}{\sin\a_\o} b,
\end{equation*}

where recall that by Assumption \ref{assum:R}, $\sin\a_{\o}\ne0$. Plugging this equality in the second of Eq. \eqref{eqn:41}, we further get:

\begin{equation*}
b=-2b\cos^2\a_\o + b(\cos^2\a_\o - \sin^2\a_\o) \equiv -b,
\end{equation*}

so that: $b=0$, $d=a$ and hence $A_{\o}^TA_{\o}= a \mathds{1}_2$, as desired.
\end{proof}

\begin{lemma}
\label{lemma_2by2}
Under Assumptions \ref{assum:H0}, \ref{assum:H} and \ref{assum:R}, we have that for every $\o\in\{1,...,N\}$ there exists a $2\times2$ unitary matrix $\mc{S}_{\o}$ such that

\begin{equation}
\label{eqn:38bis}
\mc{S}^{\dg}_{\o} \,\hat{h}_{\o}(\v{k}';m) \,\mc{S}_{\o}= \hat{h}_{\o}(R_{\o}\v{k}';m) + \mc{O}(|\v{k}'|^2+m^2),
\end{equation}

where $R_\o$ is the rotation matrix introduced in Assumption \ref{assum:R}.
\end{lemma}

\begin{proof}
Let us rewrite Eq. \eqref{eqn:35} in the basis $\{\phi_{\o;\r}\}_{\r=1,\dots,M}$, using Lemma \ref{lemma_KN}:

\begin{equation*}
\begin{split}
& \sum_{\r_1,\r_2=1}^2 \hat{h}_{\o;\r_1\r_2}(\v{k}';m) \bra{\phi_{\o;\r}} \mc{R}_{\o}^{\dg}U_{\o}(\v{k}';m)\ket{\phi_{\o;\r_1}}\bra{\phi_{\o;\r_2}}  U_{\o}^{\dg}(\v{k}';m) \mc{R}_{\o} \ket{\phi_{\o;\r'}}=\\
&\sum_{\r_1,\r_2=1}^2 \hat{h}_{\o;\r_1\r_2}(R_{\o}\v{k}';m) \bra{\phi_{\o;\r}} U_{\o}(R_{\o}\v{k}';m) \ket{\phi_{\o;\r_1}} \bra{\phi_{\o;\r_2}} U_{\o}^{\dg} (R_{\o}\v{k}';m) \ket{\phi_{\o;\r'}} +\mc{O}(|\v{k}'|^2+m^2),
\end{split}   
\end{equation*}

for any $\r,\r'=1,\dots,M$. Using that $U_{\o}(\v{k}';m)= \mathds{1}_M+ \mc{O}(|\v{k}'|+|m|)$, we further get:

\begin{equation}
\label{eqn:36}
\left(\begin{array}{cc}
\mc{S}^{\dg}_{\o} & \mc{R}_{4;\o}^{\dg} \\
\mc{R}_{3;\o}^{\dg} & \mc{R}_{2;\o}^{\dg}
\end{array} \right) \left( \begin{array}{cc}
\hat{h}_{\o}(\v{k}';m) &  0\\
0 & 0
\end{array}\right) \left(\begin{array}{cc}
\mc{S}_{\o} & \mc{R}_{3;\o} \\
\mc{R}_{4;\o} & \mc{R}_{2;\o} \end{array} \right)= \left( \begin{array}{cc}
\hat{h}_{\o}(R_{\o}\v{k}';m) &  0\\
0 & 0
\end{array}\right)+ \mc{O}(|\v{k}'|^2+m^2),
\end{equation}

where $(\mc{S}_{\o})_{\r\r'}:= \bra{\phi_{\o;\r}} \mc{R}_{\o} \ket{\phi_{\o;\r'}}$ with $\r,\r'=1,2$, $(\mc{R}_{2;\o})_{\r\r'}:= \bra{\phi_{\o;\r}} \mc{R}_{\o} \ket{\phi_{\o;\r'}}$ with $\r,\r'=3,\dots, M$, $(\mc{R}_{3;\o})_{\r\r'}:= \bra{\phi_{\o;\r}} \mc{R}_{\o} \ket{\phi_{\o;\r'}}$ with $\r=1,2$ and $\r'=3,\dots,M$, $(\mc{R}_{4;\o})_{\r\r'}:= \bra{\phi_{\o;\r}} \mc{R}_{\o} \ket{\phi_{\o;\r'}}$ with $\r=3,\dots, M$ and $\r'=1,2$. Eq. \eqref{eqn:36} readily implies Eq. \eqref{eqn:38}. We are left with showing that $\mc{S}_{\o}$ is unitary. By rewriting:

\begin{equation}
\label{eqn:37}
\left( \begin{array}{cc}
\hat{h}_{\o}(\v{k}';m) &  0\\
0 & 0
\end{array}\right) \left(\begin{array}{cc}
\mc{S}_{\o} & \mc{R}_{3;\o} \\
\mc{R}_{4;\o} & \mc{R}_{2;\o} \end{array} \right)= \left(\begin{array}{cc}
\mc{S}_{\o} & \mc{R}_{3;\o} \\
\mc{R}_{4;\o} & \mc{R}_{2;\o} \end{array} \right) \left( \begin{array}{cc}
\hat{h}_{\o}(R_{\o}\v{k}';m) &  0\\
0 & 0
\end{array}\right)+ \mc{O}(|\v{k}'|^2+m^2),
\end{equation}

we see that

\begin{itemize}
\item $\hat{h}_{\o}(\v{k}';m) \mc{R}_{3;\o} =\mc{O}(|\v{k}'|^2+m^2)$, from which: $\mc{R}_{3;\o}= \hat{h}_{\o}(\v{k}';m)^{-1} \mc{O}(|\v{k}'|^2+m^2)= \mc{O}(|\v{k}'|+|m|)$, hence $\mc{R}_{3;\o}=0$;
\item $\mc{R}_{4;\o}\hat{h}_{\o}(\v{k}';m) =\mc{O}(|\v{k}'|^2+m^2)$, from which: $\mc{R}_{4;\o}= \mc{O}(|\v{k}'|^2+m^2) \hat{h}_{\o}(\v{k}';m)^{-1}= \mc{O}(|\v{k}'|+|m|)$, hence $\mc{R}_{4;\o}=0$.
\end{itemize}

Therefore, the condition $\mc{R}_{\o}^{\dg} \mc{R}_{\o}= \mathds{1}_M$ implies $\mc{S}_{\o}^{\dg}\mc{S}_{\o}= \mathds{1}_2$.

\end{proof}

\begin{lemma}
\label{lemma:R_exact}
Under Assumptions \ref{assum:H0}, \ref{assum:H} and \ref{assum:ER}, then Assumption \ref{assum:R} holds true with $R_{\o}$ and $\mc{R}_{\o}$ equal to $R$ and $\mc{R}(\v{k}_{F,\o};0)$ respectively.
\end{lemma}

\begin{proof} 
Letting $\mc{R}_{\o}(\v{k}';m):= \mc{R}(\v{k}_{F,\o}+ \v{k}';m)$ and using the fact that for every $\o\in\{1,...,N\}$, $R\v{k}_{F,\o}+ \v{v}= \v{k}_{F,\o}$ mod. $\Span_{\mbb{Z}}\{\v{g}_1,\v{g}_2\}$, we find that
\begin{equation}
\label{eqn:46}
\mc{R}^{\dg}_{\o}(\v{k}';m) \big( \hat{H}^0_{\o}(\v{k}';m)  -\mu\mathds{1}_M \big) \mc{R}_{\o}(\v{k}';m)= \hat{H}^0_{\o}(R\v{k}';m) -\mu\mathds{1}_M, \;\; \forall\o\in\{1,...,N\}.
\end{equation}
We want to show that Eq. \eqref{eqn:46} keeps holding with $\hat{H}^0_{\o}$ replaced by $\hat{H}^0_{\Pi_{\o}}$, the latter defined in Eq. \eqref{eqn:R2}, namely:
\begin{equation}
\label{eqn:51}
\mc{R}^{\dg}_{\o}(\v{k}';m) \hat{H}^0_{\Pi_\o}(\v{k}';m) \mc{R}_{\o}(\v{k}';m)= \hat{H}^0_{\Pi_\o}(R\v{k}';m), \;\; \forall\o\in\{1,...,N\},
\end{equation}
for $\v{k}'$ and $m$ small enough (say $\max\{|\v{k}'|,|m|\} \le \d'$ for some $\d'>0$). Taking for granted Eq. \eqref{eqn:47}, and using the fact that $\hat{H}^0_{\Pi_{\o}}(\v{k}';m)= \mc{O}(|\v{k}'|+|m|)$ and $\mc{R}_{\o}(\v{k}';m)=\mc{R}_{\o}(\v{0};0)+ \mc{O}(|\v{k}'|+|m|)$, we see that Assumption \ref{assum:R} follows immediately:
\[\mc{R}^{\dg}_{\o}(\v{0};0) \hat{H}^0_{\Pi_\o}(\v{k}';m) \mc{R}_{\o}(\v{0};0)= \hat{H}^0_{\Pi_\o}(R\v{k}';m) + \mc{O}(|\v{k}'|^2+m^2), \;\; \forall\o\in\{1,...,N\}.\]
We shall therefore prove Eq. \eqref{eqn:51}. Recalling that due to Assumption \ref{assum:H}, $\hat{H}^0_{\o}(\v{k}';m)- \mu\mathds{1}_M$ is invertible for any $(\v{k}',m)\ne(\v{0},0)$, we can rewrite Eq. \eqref{eqn:46} as:
\begin{equation}
\label{eqn:48}
\mc{R}_{\o}(\v{k}';m)= \big(\hat{H}^0_{\o}(\v{k}';m)- \mu\mathds{1}_M\big)^{-1}\mc{R}_{\o}(\v{k}';m) \big(\hat{H}^0_{\o}(R\v{k}';m)- \mu\mathds{1}_M\big).
\end{equation}
If we introduce $\Sigma_{\o}(\v{k}';m):= \mathds{1}_M - \Pi_{\o}(\v{k}';m)$ and we let, for any $\sharp,\sharp'\in\{\Pi,\Sigma\}$,
\[\begin{split} &\mc{R}_{\sharp_{\o},\sharp'_{\o}}(\v{k}';m):= \sharp_{\o}(\v{k}';m) \mc{R}_{\o}(\v{k}';m) \sharp'_{\o}(R\v{k}';m), \\
&(\hat{H}^0)^p_{\sharp_{\o}}(\v{k}';m):= \sharp_{\o}(\v{k}';m) \big(\hat{H}^0_{\o}(\v{k}';m) -\mu\mathds{1}_M\big)^p \sharp_{\o}(\v{k}';m), \; p\in\mbb{Z},
\end{split}\]
Eq. \eqref{eqn:48} gives:
\begin{equation}
\label{eqn:49}
\mc{R}_{\Sigma_{\o},\Pi_{\o}}(\v{k}';m)= \big(\hat{H}^0\big)^{-1}_{\Sigma_{\o}}(\v{k}';m)\mc{R}_{\Sigma_{\o},\Pi_{\o}}(\v{k}';m) \hat{H}^0_{\Pi_{\o}}(R\v{k}';m).\end{equation}
Note that after Eq. \eqref{eqn:1bis} we have that $\big\|\big(\hat{H}^0_{\o}\big)^{-1}_{\Sigma_{\o}}(\v{k}';m)\big\|\le \Delta^{-1}$ and in view of Eq. \eqref{eq_6a}, $\big\|\hat{H}^0_{\Pi_{\o}}(\v{k}';m)\big\|\le C'_{\star} (|\v{k}'|+|m|)$ for $\max\{|\v{k}'|,|m|\}\le a$. Therefore, Eq. \eqref{eqn:49}, regarded as a linear equation for the matrix elements of $\mc{R}_{\Sigma_{\o},\Pi_{\o}}(\v{k}';m)$, implies, for $|\v{k}'|$ and $|m|$ small enough, that $\mc{R}_{\Sigma_{\o},\Pi_{\o}}(\v{k}';m)=0$. Besides, reverting Eq. \eqref{eqn:48}, we find the analogue of Eq. \eqref{eqn:49}:
\[\mc{R}_{\Pi_{\o},\Sigma_{\o}}(\v{k}';m)= \hat{H}^0_{\Pi_{\o}}(\v{k}';m)\mc{R}_{\Pi_{\o},\Sigma_{\o}}(\v{k}';m) \big(\hat{H}^0\big)^{-1}_{\Sigma_{\o}}(R\v{k}';m),\]
which analogously implies $\mc{R}_{\Pi_{\o},\Sigma_{\o}}(\v{k}';m)=0$ for $\v{k}'$ and $m$ small enough. Therefore, since $\Pi_{\o}+\Sigma_{\o}=\mathds{1}_M$, we find that 
\begin{equation}
\label{eqn:50}
\mc{R}_{\o}(\v{k}';m)= \mc{R}_{\Pi_{\o},\Pi_{\o}}(\v{k}';m)+ \mc{R}_{\Sigma_{\o},\Sigma_{\o}}(\v{k}';m).
\end{equation}
At this point we go back to Eq. \eqref{eqn:46}, and we find:
\begin{equation}
\begin{split}
&\hat{H}^0_{\Pi_{\o}}(R\v{k}';m)=\\
&\Pi_{\o}(R\v{k}';m)\mc{R}^{\dg}_{\o}(\v{k}';m) \big( \hat{H}^0_{\o}(\v{k}';m) -\mu\mathds{1}_M\big) \mc{R}_{\o}(\v{k}';m) \Pi_{\o}(R\v{k}';m)= \\
&\big(\mc{R}_{\Pi_{\o},\Pi_{\o}}(\v{k}';m)\big)^{\dg} \big( \hat{H}^0_{\o}(\v{k}';m) -\mu \mathds{1}_M\big) \mc{R}_{\Pi_{\o},\Pi_{\o}}(\v{k}';m)= \\
&\big(\mc{R}_{\Pi_{\o},\Pi_{\o}}(\v{k}';m)\big)^{\dg} \hat{H}^0_{\Pi_{\o}}(\v{k}';m) \mc{R}_{\Pi_{\o},\Pi_{\o}}(\v{k}';m)= \mc{R}_{\o}^{\dg}(\v{k}';m) \hat{H}^0_{\Pi_{\o}}(\v{k}';m) \mc{R}_{\o}(\v{k}';m),
\end{split}
\end{equation}
where for the last equality we used Eq. \eqref{eqn:50} and the fact that $\hat{H}^0_{\Pi_{\o}}(\v{k}';m)\mc{R}_{\Sigma_{\o},\Sigma_{\o}}(\v{k}';m)=0$, thus proving Eq. \eqref{eqn:51}.
\end{proof}

\section{Technical results}
\label{app:tech}

This appendix is devoted to collecting the proof of several claims stated in Section \ref{sect:proof} which serve as building blocks for the proof of the main theorem.

\subsection{The conical splitting}
\label{app:1_conical}

Here we focus on the proof of two lemmas related to the \emph{conical splitting}, namely the block-diagonalization induced by the Kato-Nagy unitary.

\begin{proof}[Proof of Lemma \ref{lem:smoothP}.]
In view of Remark \ref{rmk:bands}.\ref{itt:1}-\ref{itt:4}, for every $\o\in\{1,\dots,N\}$, $j\ne n,n+1$ and $\max\{|\v{k}'|,|m|\}\le a$, we have that

\begin{align*}
&\big|\l_{\pm}(\v{k}_{F,\o}+\v{k}';m)- \mu\big|\le C'_{\star}(|\v{k}'|+|m|),\\
&\big|\lambda_j(\v{k}_{F,\o}+\v{k}';m)-\mu\big|\ge \Delta.
\end{align*}

Moreover, setting $\kappa:= \min\big\{\frac{\Delta}{4C'_{\star}}, a \big\}$, we see that by restricting $\max\{|\v{k}'|,|m|\}\le \kappa$, then $\big|\l_{\pm}(\v{k}_{F,\o}+\v{k}';m)- \mu\big|\le \frac{\Delta}{2}$. Therefore, considering the complex contour $\gamma:=\{z\in\mbb{C}:\; |z-\mu|=\frac{3}{4}\Delta\}$, enclosing only  $\l_\pm(\v{k}_{F,\o}+\v{k}';m)$ and having a positive distance from all the eigenvalues $\lambda_j(\v{k}_{F,\o}+\v{k}';m)$ with $j=1,\dots,M$, we can employ the Riesz formula as:  
\[
\Pi(\v{k}_{F,\o}+\v{k}';m)=\frac{i}{2\pi}\oint_{\gamma}dz\,\big(\hat{H}^0(\v{k}_{F,\o}+\v{k}';m)-z\big)^{-1}.
\]
Thus, from the analyticity of the map $B_{\kappa}(\v{0})\times[-\kappa,\kappa]\ni (\v{k}';m)\mapsto\hat{H}^0(\v{k}_F^{\o}+\v{k}';m)$, we deduce the analyticity of $\Pi$.
\end{proof}

\begin{proof}[Proof of Lemma \ref{lemma_KN}.]
Obviously, we have that 
\begin{equation}
\label{eqn:sumH}
\hat{H}^0_\o(\v{k}';m)-\mu=\big(\hat{H}^0_\o(\v{k}';m)-\mu\big)\Pi_\o(\v{k}';m)+\big(\hat{H}^0_\o(\v{k}';m)-\mu\big)\big(\mathds{1}_M- \Pi_\o(\v{k}';m)\big).
\end{equation}
Now, we unitarily conjugate the first summand in the r.h.s by using the Kato-Nagy unitary:
\begin{align*}
&U_{\o}^{\dg}(\v{k}';m)(\hat{H}^0_\o(\v{k}';m)-\mu)\Pi_\o(\v{k}';m)U_{\o}(\v{k}';m)\\
&=U_{\o}^{\dg}(\v{k}';m)\Pi_\o(\v{k}';m)U_{\o}(\v{k}';m)U_{\o}^{\dg}(\v{k}';m)(\hat{H}^0_\o(\v{k}';m)-\mu)U_{\o}(\v{k}';m)\cdot\\
&\quad\cdot U_{\o}^{\dg}(\v{k}';m)\Pi_\o(\v{k}';m)U_{\o}(\v{k}';m)\\
&=\Pi_\o(\v{0};0)U_{\o}^{\dg}(\v{k}';m)(\hat{H}^0_\o(\v{k}';m)-\mu)U_{\o}(\v{k}';m)\Pi_\o(\v{0};0)\\
&=\sum_{1\leq \r,\r'\leq 2} \ket{\phi_{\o,\r}} \bra{\phi_{\o,\r}}U_{\o}^{\dg}(\v{k}';m)(\hat{H}^0_\o(\v{k}';m)-\mu)U_{\o}(\v{k}';m)\ket{\phi_{\o,\r'}}\bra{\phi_{\o,\r'}}\\
&=\sum_{1\leq \r,\r'\leq 2} \hat{h}_{\o;\r\r'}(\v{k}';m)\ket{\phi_{\o,\r}}\bra{\phi_{\o,\r'}},
\end{align*}
where we have used that $U_{\o}(\v{k}';m)U_{\o}^{\dg}(\v{k}';m)=\id_M$ and the intertwining property of $U_{\o}(\v{k}';m)$, Eq. \eqref{eqn:2}. From the analyticity of the Kato-Nagy operator follows the analyticity of the $2\times 2$ matrix $\hat{h}_{\o}$.
Analogously for the second summand in the r.h.s.\ of Eq. \eqref{eqn:sumH}.
\end{proof}

\subsection{The remainders to the relativistic conductivity}
\label{app:1_rem}

This part of the appendix focuses on the steps allowing to express $\mf{D}\s_{ij}(\e)$ in terms of its \emph{relativistic} counterpart, namely $\mf{D}\s_{ij}^{(1)}(\e)$.

\begin{proof}[Proof of Lemma \ref{lemma:K_reg}.]
By introducing $\dist{\bs{k}}_{\mbb{R}\times\mbb{B}}:=\sqrt{k_0^2 + \dist{\v{k}}_{\mbb{B}}^2}$, with $\dist{\v{k}}_{\mbb{B}}$ defined after Eq. \eqref{def:Brillouin}, we can write the two-point function as:

\begin{equation}
\label{eqn:A1}
\begin{split}
\hat{S}(\bs{k};m)= \sum_{\o=1}^N \frac{\chi(\d^{-1}\dist{\bs{k}-\bs{k}_{F,\o}}_{\mbb{R}\times\mbb{B}})}{-ik_0+ \hat{H}^0(\v{k};m)-\mu}+ \underbrace{\frac{1- \sum_{\o=1}^N \chi(\d^{-1}\dist{\bs{k}-\bs{k}_{F,\o}}_{\mbb{R}\times\mbb{B}})}{-ik_0+ \hat{H}^0(\v{k};m)-\mu}}_{=:\hat{S}_C(\bs{k};m)},
\end{split}
\end{equation}

where $\bs{k}_{F,\o}\equiv (0,\v{k}_{F,\o})$. Furthermore, in view of Lemma \ref{lemma_KN}, we can write:

\begin{equation}
\frac{\chi(\d^{-1}\dist{\bs{k}-\bs{k}_{F,\o}}_{\mbb{R}\times\mbb{B}})}{-ik_0+ \hat{H}^0(\v{k};m)-\mu} =   \hat{S}_{A;\o}(\bs{k}- \bs{k}_{F,\o};m) + \hat{S}_{B;\o}(\bs{k}-\bs{k}_{F,\o};m),
\end{equation}

where:

\begin{equation}
\label{eqn:A5a}
\begin{split}
\hat{S}_{A;\o}(\bs{k}';m)&= \chi(\d^{-1}\dist{\bs{k}'}_{\mbb{R}\times\mbb{B}}) U_{\o}(\v{k}';m) \sum_{1\le \r,\r'\le 2} \big( -ik_0 + \hat{h}_{\o}(\v{k}';m) \big)^{-1}_ {\r\r'}\ket{\phi_{\o;\r}} \bra{\phi_{\o;\r'}}   U^{\dg}_{\o}(\v{k}';m),  \\ &\equiv\chi(\d^{-1}\dist{\bs{k}'}_{\mbb{R}\times\mbb{B}})  \sum_{1\le \r,\r'\le 2}  \hat{s}_{\o;\r\r'}(\bs{k}';m)\ket{\phi_{\o;\r}(\v{k}';m)} \bra{\phi_{\o;\r'}(\v{k}';m)}   
\end{split}
\end{equation}

\begin{equation}
\label{eqn:A5b}
\begin{split}
\hat{S}_{B;\o}(\bs{k}';m)&= \chi(\d^{-1}\dist{\bs{k}'}_{\mbb{R}\times\mbb{B}}) U_{\o}(\v{k}';m) \sum_{3\le \r,\r'\le M} \!\!\big( -ik_0 + \hat{r}_{\o}(\v{k}';m) \big)^{-1}_{\r\r'} \ket{\phi_{\o;\r}} \bra{\phi_{\o;\r'}}   U^{\dg}_{\o}(\v{k}';m).
\end{split}
\end{equation}

Therefore, by writing $\hat{S}(\bs{k};m)= \sum_{\o=1}^N\big( \hat{S}_{A;\o}(\bs{k}-\bs{k}_{F,\o};m)+ \hat{S}_{B;\o}(\bs{k}-\bs{k}_{F,\o};m) \big)+ \hat{S}_{C}(\bs{k};m)$, Eq. \eqref{eq_3} takes the form:

\begin{equation}
\label{eqn:A2}
\begin{split}
&\hat{K}_{ij}(p_0;m)=\\
& -\sum_{\o=1}^N \int_{\mbb{R}\times\mbb{B}} \frac{d\bs{k}}{(2\pi)^3} \Tr_{\mbb{C}^M} \left\{ \Gamma_i(\v{k},\v{0};m) \hat{S}_{A;\o}(\bs{k}-\bs{k}_{F,\o};m) \Gamma_j(\v{k},\v{0};m) \hat{S}_{A;\o}(\bs{k}- \bs{k}_{F,\o}+\bs{p}_0;m) \right\}\\
&+ \hat{K}_{ij}^{(\mathrm{reg})}(p_0;m),
\end{split}
\end{equation}

where, letting $\hat{S}_{BC}(\bs{k};m):= \sum_{\o=1}^N \hat{S}_{B;\o}(\bs{k}-\bs{k}_{F,\o};m) + \hat{S}_{C}(\bs{k};m)$, 

\begin{equation}
\label{eqn:A3}
\begin{split}
& \hat{K}_{ij}^{(\mathrm{reg})}(p_0;m)=\\
&- \int_{\mbb{R}\times\mbb{B}} \frac{d\bs{k}}{(2\pi)^3} \Tr_{\mbb{C}^M} \Big\{ \Gamma_i(\v{k};\v{0};m) \hat{S}_{BC}(\bs{k};m) \Gamma_j(\v{k};\v{0};m) \hat{S}_{BC}(\bs{k}+\bs{p}_0;m) \Big\}+\\
&- \sum_{\o=1}^N \int_{\mbb{R}\times\mbb{B}} \frac{d\bs{k}}{(2\pi)^3} \Tr_{\mbb{C}^M} \Big\{ \Gamma_i(\v{k};\v{0};m) \hat{S}_{A;\o}(\bs{k}- \bs{k}_{F,\o};m) \Gamma_j(\v{k};\v{0};m) \hat{S}_{BC}(\bs{k}+\bs{p}_0;m) \Big\}+\\
&- \sum_{\o=1}^N \int_{\mbb{R}\times\mbb{B}} \frac{d\bs{k}}{(2\pi)^3} \Tr_{\mbb{C}^M} \Big\{ \Gamma_i(\v{k};\v{0};m)  \hat{S}_{BC}(\bs{k};m) \Gamma_j(\v{k};\v{0};m) \hat{S}_{A;\o}(\bs{k}+ \bs{p}_0- \bs{k}_{F,\o};m) \Big\}.
\end{split}
\end{equation}

In deriving Eq. \eqref{eqn:A2} we used the property that, by the definition of $\d$ (see Subsection \ref{ssect:twobands}), the product 
$$\chi(\d^{-1}\dist{\bs{k}-\bs{k}_{F,\o}}_{\mbb{R}\times\mbb{B}}) \chi(\d^{-1}\dist{\bs{k}-\bs{k}_{F,\o'}}_{\mbb{R}\times\mbb{B}})$$ 
is identically zero unless $\o=\o'$. Notice that the second line of Eq. \eqref{eqn:A2} coincides with $\hat{K}_{ij}^{(\mathrm{sing})}$ defined in Eq. \eqref{eqn:12}. Indeed, using Eq. \eqref{eqn:A5a} and recalling the definition of the vertex functions $\gamma_{\o;i}$, Eq. \eqref{eqn:gamma}, we have that

\begin{equation*}
\begin{split}
&\int_{\mbb{R}\times\mbb{B}} \frac{d\bs{k}}{(2\pi)^3} \Tr_{\mbb{C}^M} \left\{ \Gamma_i(\v{k},\v{0};m) \hat{S}_{A;\o}(\bs{k}-\bs{k}_{F,\o};m) \Gamma_j(\v{k},\v{0};m) \hat{S}_{A;\o}(\bs{k}- \bs{k}_{F,\o}+\bs{p}_0;m) \right\}=\\
&\int_{\mbb{R}\times\mbb{B}} \frac{d\bs{k}'}{(2\pi)^3} \Tr_{\mbb{C}^M} \left\{ \Gamma_i(\v{k}'+\v{k}_{F,\o},\v{0};m) \hat{S}_{A;\o}(\bs{k}';m) \Gamma_j(\v{k}'+\v{k}_{F,\o},\v{0};m) \hat{S}_{A;\o}(\bs{k}'+\bs{p}_0;m) \right\}=\\
&\int_{\mbb{R}\times\mbb{B}} \frac{d\bs{k}}{(2\pi)^3} \Tr_{\mbb{C}^2}\Big\{\gamma_{\o;i}(\v{k}';m) \hat{s}_{\o}(\bs{k}';m) \gamma_{\o;j}(\v{k}';m) \hat{s}_{\o}(\bs{k}'+ \bs{p}_0;m) \Big\}\cdot\\
& \qquad\qquad \cdot\chi(\d^{-1}\dist{\bs{k}'}_{\mbb{R}\times\mbb{B}}) \chi(\d^{-1}\dist{\bs{k}' + \bs{p}_0}_{\mbb{R}\times\mbb{B}})=\\
& \int_{\mbb{R}\times\mbb{B}} \frac{d\bs{k}'}{(2\pi)^3} \Tr_{\mbb{C}^2}\Big\{\gamma_{\o;i}(\v{k}';m) \hat{s}_{\o}(\bs{k}';m) \gamma_{\o;j}(\v{k}';m) \hat{s}_{\o}(\bs{k}'+ \bs{p}_0;m) \Big\} \chi(\d^{-1}|\bs{k}'|) \chi(\d^{-1}|\bs{k}' +\bs{p}_0|),
\end{split}
\end{equation*}

where we performed the change of variable of integration $\bs{k}\mapsto \bs{k}'= \bs{k}-\bs{k}_{F,\o}$ and used the fact that, due to the definition of $\d$, on $\mbb{R}\times\mbb{B}$ it holds true that $\chi(\d^{-1}\dist{\,\cdot\,}_{\mbb{R}\times\mbb{B}})=\chi(\d^{-1}|\,\cdot\,|)$ (recall Footnote \ref{footnote:B}).

We shall now discuss the regularity properties of $\hat{K}_{ij}^{(\mathrm{reg})}$. It is first convenient to perform the change of variables $\bs{k}\mapsto \bs{k}-\bs{p}_0$ within the integral in the last line of Eq. \eqref{eqn:A3}, so that

\[
\begin{split}
& \hat{K}_{ij}^{(\mathrm{reg})}(p_0;m)=\\
&- \int_{\mbb{R}\times\mbb{B}} \frac{d\bs{k}}{(2\pi)^3} \Tr_{\mbb{C}^M} \Big\{ \Gamma_i(\v{k};\v{0};m) \hat{S}_{BC}(\bs{k};m) \Gamma_j(\v{k};\v{0};m) \hat{S}_{BC}(\bs{k}+\bs{p}_0;m) \Big\}+\\
&- \sum_{\o=1}^N \int_{\mbb{R}\times\mbb{B}} \frac{d\bs{k}}{(2\pi)^3} \Tr_{\mbb{C}^M} \Big\{ \Gamma_i(\v{k};\v{0};m) \hat{S}_{A;\o}(\bs{k}- \bs{k}_{F,\o};m) \Gamma_j(\v{k};\v{0};m) \hat{S}_{BC}(\bs{k}+\bs{p}_0;m) \Big\}+\\
&- \sum_{\o=1}^N \int_{\mbb{R}\times\mbb{B}} \frac{d\bs{k}}{(2\pi)^3} \Tr_{\mbb{C}^M} \Big\{ \Gamma_j(\v{k};\v{0};m) \hat{S}_{A;\o}(\bs{k}- \bs{k}_{F,\o};m) \Gamma_i(\v{k};\v{0};m) \hat{S}_{BC}(\bs{k}-\bs{p}_0;m) \Big\}.
\end{split}
\]

Using the fact that, in view of Remark \ref{rmk:bands}.\ref{itt:1}-\ref{itt:3}, $\|\hat{S}_{BC}(\bs{k};m)\|\le C(1+|k_0|)^{-1}$, $\\\|\de_{k_0}\hat{S}_{BC}(\bs{k};m)\|\le C(1+|k_0|^2)^{-1}$ and $\|\hat{S}_{A;\o}(\bs{k}';m)\|\le C\dist{\bs{k}'}_{\mbb{R}\times\mbb{B}}^{-1}$, for a suitable constant $C>0$, we see that $\hat{K}_{ij}^{(\mathrm{reg})}$ is continuously differentiable w.r.t. $p_0$, with:

\begin{equation}
\begin{split}
& \de_{p_0}\hat{K}_{ij}^{(\mathrm{reg})}(p_0;m)=\\
&- \int_{\mbb{R}\times\mbb{B}} \frac{d\bs{k}}{(2\pi)^3} \Tr_{\mbb{C}^M} \Big\{ \Gamma_i(\v{k};\v{0};m) \hat{S}_{BC}(\bs{k};m) \Gamma_j(\v{k};\v{0};m) \de_{k_0}\hat{S}_{BC}(\bs{k}+\bs{p}_0;m) \Big\}+\\
&- \sum_{\o=1}^N \int_{\mbb{R}\times\mbb{B}} \frac{d\bs{k}}{(2\pi)^3} \Tr_{\mbb{C}^M} \Big\{ \Gamma_i(\v{k};\v{0};m) \hat{S}_{A;\o}(\bs{k}- \bs{k}_{F,\o};m) \Gamma_j(\v{k};\v{0};m) \de_{k_0}\hat{S}_{BC}(\bs{k}+\bs{p}_0;m) \Big\}+\\
&+ \sum_{\o=1}^N \int_{\mbb{R}\times\mbb{B}} \frac{d\bs{k}}{(2\pi)^3} \Tr_{\mbb{C}^M} \Big\{ \Gamma_j(\v{k};\v{0};m) \hat{S}_{A;\o}(\bs{k}- \bs{k}_{F,\o};m) \Gamma_i(\v{k};\v{0};m) \de_{k_0}\hat{S}_{BC}(\bs{k}-\bs{p}_0;m) \Big\}.
\end{split}
\end{equation}

Now, using the fact that the vertex functions $\Gamma_j(\v{k};\v{0};m)$ are analytic, and that $\\\|\de_m\de_{k_0} \hat{S}_{BC}(\bs{k};m)\|\le C'(1+|k_0|^3)^{-1}$, $\|\de_m\hat{S}_{A;\o}(\bs{k}';m)\|\le C'\dist{\bs{k}'}_{\mbb{R}\times\mbb{B}}^{-2}$, for a suitable $C'>0$, we find that $\de_{p_0}\hat{K}_{ij}^{(\mathrm{reg})}$ is $\mathscr{C}^1$ w.r.t. $m$. 
\end{proof}

\begin{lemma}
\label{lemma:linear}
Letting $\s^{(1)}_{ij;\o}$ and $\s_{ij}^{(\mathrm{sing})}$ as in Eq. \eqref{eqn:14} and Eq. \eqref{def:sigma_sing} respectively, one has that for every $0<\e<\e_0$ (with $\e_0$ defined at the beginning of Subsection \ref{ssect:linear}), the difference:

\begin{equation}
\label{eqn:A6}
\mf{D}\s^{(\mathrm{sing})}_{ij}(p_0;\e) - \sum_{\o=1}^N \mf{D}\s^{(1)}_{ij;\o}(p_0;\e)
\end{equation}

admits a finite limit as $p_0\to 0^-$, which is bounded in absolute value by $C_0\sqrt{\e}$, for a suitable constant $C_0>0$.
\end{lemma}

\begin{proof}
For $m\in\{0,\pm\e\}$ we write:

\begin{align}
& \s^{(\mathrm{sing})}_{ij}(p_0;m)= \sum_{\o=1}^N \int_{\mbb{R}^3} \frac{d\bs{k}'}{(2\pi)^3}\, \mf{d}_0\mc{F}^{(\mathrm{sing})}_{ij;\o}(\bs{k}';p_0;m),\\
& \s_{ij;\o}^{(1)}(p_0;m)=  \sum_{\o=1}^N \int_{\mc{D}_{\e,\o}} \frac{d\bs{k}'}{(2\pi)^3}\, \mf{d}_0\mc{F}^{(1)}_{ij;\o}(\bs{k}';p_0;m),
\end{align}

with (hereafter in this proof we denote $\Tr_{\mbb{C}^2}$ simply by $\Tr$):

\begin{align}
& \mc{F}^{(\mathrm{sing})}_{ij;\o}(\bs{k}';p_0;m):= -\chi(\tfrac{|\bs{k}'|}{\d}) \chi(\tfrac{|\bs{k}'+\bs{p}_0|}{\d}) \Tr\Big\{ \gamma_{\o;i}(\v{k}';m) \hat{s}_{\o}(\bs{k}';m) \gamma_{\o;j}(\v{k}';m) \hat{s}_{\o}(\bs{k}'+\bs{p}_0;m) \Big\},\\
& \mc{F}^{(1)}_{ij;\o}(\bs{k}';p_0;m):= -\Tr\Big\{ \gamma_{\o;i}(\v{0};0) \hat{s}^{(1)}_{\o}(\bs{k}';m) \gamma_{\o;j}(\v{0};0) \hat{s}^{(1)}_{\o}(\bs{k}'+\bs{p}_0;m) \Big\},
\end{align}

$\bs{p}_0\equiv (p_0,\v{0})$ and $\mf{d}_0$ the \emph{discrete-derivative operator} w.r.t. $p_0$, namely $\mf{d}_0\mc{F}(p_0)\equiv \frac{1}{p_0}\big(\mc{F}(p_0)- \mc{F}(0) \big)$. Given these definitions, we have that

\begin{equation}
\label{eqn:A7}
\begin{split}
& \mf{D}\s^{(\mathrm{sing})}_{ij}(p_0;\e) - \sum_{\o=1}^N \mf{D}\s^{(1)}_{ij;\o}(p_0;\e)=\\
&+\frac{1}{2}\sum_{\o=1}^N \sum_{\a=\pm} \int_{(\mc{D}_{\e,\o})^c} \frac{d\bs{k}'}{(2\pi)^3} \mf{d}_0 \Big( \mc{F}^{(\mathrm{sing})}_{ij;\o}(\bs{k}';p_0;0)- \mc{F}^{(\mathrm{sing})}_{ij;\o}(\bs{k}';p_0;\a \e) \Big) \quad \Bigg\}:= I_1(p_0;\e)\\
&-\frac{1}{2}\sum_{\o=1}^N \sum_{\a=\pm} \int_{\mc{D}_{\e,\o}} \frac{d\bs{k}'}{(2\pi)^3} \mf{d}_0 \Big( \mc{F}^{(\mathrm{sing})}_{ij;\o}(\bs{k}';p_0;\a \e)- \mc{F}^{(1)}_{ij;\o}(\bs{k}';p_0;\a \e) \Big) \quad \Bigg\}:= I_2(p_0;\e)\\
& +\sum_{\o=1}^N \int_{\mc{D}_{\e,\o}} \frac{d\bs{k}'}{(2\pi)^3} \mf{d}_0 \Big( \mc{F}^{(\mathrm{sing})}_{ij;\o}(\bs{k}';p_0;0)- \mc{F}^{(1)}_{ij;\o}(\bs{k}';p_0;0) \Big) \quad \Bigg\}:= I_3(p_0;\e),
\end{split}
\end{equation}

where $(\mc{D}_{\e,\o})^c\equiv (\mbb{R}\times\mbb{B})\setminus\mc{D}_{\e,\o}$. We are going to analyze the three terms in the r.h.s. of Eq. \eqref{eqn:A7} separately.

\medskip
\textbf{Analysis of $I_1$.} Note that since the integral is restricted to $(\mc{D}_{\e,\o})^c$, the integrand is smooth in $p_0$. Moreover, due to the presence of the cut-off functions $\chi$, the integrand is supported on a bounded subset of $\mbb{R}\times\mbb{B}$. Therefore:

\begin{equation}
\lim_{p_0\rightarrow 0^-} I_1(p_0;\e)= \frac{1}{2}\sum_{\o=1}^N \sum_{\a=\pm} \int_{(\mc{D}_{\e,\o})^c} \frac{d\bs{k}'}{(2\pi)^3} \Big( \de_{p_0}\mc{F}^{(\mathrm{sing})}_{ij;\o}(\bs{k}';0;0)- \de_{p_0}\mc{F}^{(\mathrm{sing})}_{ij;\o}(\bs{k}';0;\a \e) \Big).
\end{equation}

Explicitly, for $m\in\{0,\pm\e\}$, we have that

\[\begin{split}
\de_{p_0}\mc{F}^{(\mathrm{sing})}_{ij;\o}(\bs{k}';0;m)&=  -\chi(\tfrac{|\bs{k}'|}{\d}) \chi'(\tfrac{|\bs{k}'|}{\d}) \tfrac{\bs{k}'}{\d|\bs{k}'|} \Tr\Big\{ \gamma_{\o;i}(\v{k}';m) \hat{s}_{\o}(\bs{k}';m) \gamma_{\o;j}(\v{k}';m) \hat{s}_{\o}(\bs{k}';m) \Big\}+\\
&- \chi(\tfrac{|\bs{k}'|}{\d})^2  \Tr\Big\{ \gamma_{\o;i}(\v{k}';m) \hat{s}_{\o}(\bs{k}';m) \gamma_{\o;j}(\v{k}';m) \de_{k_0}\hat{s}_{\o}(\bs{k}';m) \Big\}.
\end{split}\]

Therefore:

\begin{equation}
\label{eqn:A8}
\begin{split}
\lim_{p_0\rightarrow 0^-} I_1(p_0;\e)&=- \frac{1}{2}\sum_{\o=1}^N \sum_{\a=\pm} \int_0^{\e} dm \int_{(\mc{D}_{\e,\o})^c} \frac{d\bs{k}'}{(2\pi)^3} \cdot\\
&\cdot\Bigg\{ \chi(\tfrac{|\bs{k}'|}{\d}) \chi'(\tfrac{|\bs{k}'|}{\d}) \tfrac{\bs{k}'}{\d|\bs{k}'|} \de_m \Tr\Big\{ \gamma_{\o;i}(\v{k}';\a m) \hat{s}_{\o}(\bs{k}';\a m) \gamma_{\o;j}(\v{k}';\a m) \hat{s}_{\o}(\bs{k}';\a m) \Big\}\\
&\qquad +\chi(\tfrac{|\bs{k}'|}{\d})^2  \de_m \Tr\Big\{ \gamma_{\o;i}(\v{k}';\a m) \hat{s}_{\o}(\bs{k}';\a m) \gamma_{\o;j}(\v{k}';\a m) \de_{k_0}\hat{s}_{\o}(\bs{k}';\a m) \Big\}\Bigg\}.
\end{split}
\end{equation}

Using the bounds (which follow from Assumption \ref{assum:H}): $\|\hat{s}_{\o}(\bs{k}';\a m)\|\le C_1|\bs{k}'|^{-1}$, $\\\|\de_{k_0}\hat{s}_{\o}(\bs{k}';\a m)\|\le C_1|\bs{k}'|^{-2}$, $\|\de_m\de_{k_0}\hat{s}_{\o}(\bs{k}';\a m)\|\le C_1|\bs{k}'|^{-3}$ and $|\gamma_{\o;i}(\v{k}';m)|,|\de_m\gamma_{\o;i}(\v{k}';m)|\le C_1$, for a suitable constant $C_1>0$, from Eq. \eqref{eqn:A8} we find that

\[\begin{split}
\Big|\lim_{p_0\rightarrow 0^-} I_1(p_0;\e)\Big|&\le \sum_{\o=1}^N \sum_{\a=\pm} \int_0^{\e} dm \int_{(\mc{D}_{\e,\o})^c} \frac{d\bs{k}'}{(2\pi)^3} \frac{C_1'}{|\bs{k}'|^4}\le \frac{NC_1'}{2\pi}\e \int_{\sqrt{\e}/\max\{v_{\o;2},1\}}^{\d} \frac{dk}{k^2}\\
&\le \frac{NC_1'(\max_{\o=1}^Nv_{\o;2}+1)}{2\pi} \sqrt{\e},
\end{split}\]

for a suitable constants $C_1'>0$ and with $v_{\o;2}>0$ the largest eigenvalue of $|A_{\o}|$, introduced at the beginning of Subsection \ref{ssect:Dirac}.

\medskip
\textbf{Analysis of $I_2$.}
We observe that since we have a non-zero mass $\pm\e$, the integrand is smooth w.r.t. $p_0$ and hence:

\begin{equation}
\label{eqn:A9}
\lim_{p_0\rightarrow 0^-} I_2(p_0;\e)= -\frac{1}{2}\sum_{\o=1}^N \sum_{\a=\pm} \int_{\mc{D}_{\e,\o}} \frac{d\bs{k}'}{(2\pi)^3}  \Big( \de_{p_0}\mc{F}^{(\mathrm{sing})}_{ij;\o}(\bs{k}';0;\a \e)- \de_{p_0}\mc{F}^{(1)}_{ij;\o}(\bs{k}';0;\a \e) \Big).
\end{equation}

Similarly to $I_1$, for $m\in\{\pm\e\}$ we have:

\[\begin{split}
\de_{p_0}\mc{F}^{(\mathrm{sing})}_{ij;\o}(\bs{k}';0;m)- \de_{p_0}\mc{F}^{(1)}_{ij;\o}(\bs{k}';0;m)=&-  \Tr\Big\{ \gamma_{\o;i}(\v{k}';m) \hat{s}_{\o}(\bs{k}';m) \gamma_{\o;j}(\v{k}';m) \de_{k_0}\hat{s}_{\o}(\bs{k}';m) \Big\}\\
& +\Tr\Big\{ \gamma_{\o;i}(\v{0};0) \hat{s}^{(1)}_{\o}(\bs{k}';m) \gamma_{\o;j}(\v{0};0) \de_{k_0}\hat{s}^{(1)}_{\o}(\bs{k}';m) \Big\},
\end{split}\]

where we also used that, by the definition of $\e_0$, within the support of $\mc{D}_{\e,\o}$, the cutoff $\chi(\d^{-1}|\,\cdot\,|)$ is identically equal to 1. The r.h.s. can be further expanded, yielding:

\begin{equation}
\label{eqn:A10}
\begin{split}
&\de_{p_0}\mc{F}^{(\mathrm{sing})}_{ij;\o}(\bs{k}';0;m)- \de_{p_0}\mc{F}^{(1)}_{ij;\o}(\bs{k}';0;m)=\\
& -  \Tr\Big\{ \big(\gamma_{\o;i}(\v{k}';m)- \gamma_{\o;i}(\v{0};0)\big) \hat{s}_{\o}(\bs{k}';m) \gamma_{\o;j}(\v{k}';m) \de_{k_0}\hat{s}_{\o}(\bs{k}';m) \Big\}+\\
&-  \Tr\Big\{ \gamma_{\o;i}(\v{0};0) \big(\hat{s}_{\o}(\bs{k}';m) - \hat{s}^{(1)}_{\o}(\bs{k}';m)\big) \gamma_{\o;j}(\v{k}';m) \de_{k_0}\hat{s}_{\o}(\bs{k}';m) \Big\}+\\
& -  \Tr\Big\{ \gamma_{\o;i}(\v{0};0) \hat{s}^{(1)}_{\o}(\bs{k}';m) \big(\gamma_{\o;j}(\v{k}';m)- \gamma_{\o;j}(\v{0};0)\big) \de_{k_0}\hat{s}_{\o}(\bs{k}';m) \Big\}+\\
& -  \Tr\Big\{ \gamma_{\o;i}(\v{0};0) \hat{s}^{(1)}_{\o}(\bs{k}';m) \gamma_{\o;j}(\v{0};0) \big(\de_{k_0}\hat{s}_{\o}(\bs{k}';m) - \de_{k_0}\hat{s}^{(1)}_{\o}(\bs{k}';m)\big) \Big\}.
\end{split}
\end{equation}

Observe that the following rewriting is possible:

\[\hat{s}_{\o}(\bs{k}';m)- \hat{s}^{(1)}_{\o}(\bs{k}';m)= -\hat{s}_{\o}(\bs{k}';m) \Big(\hat{h}_{\o}(\v{k}';m)- \hat{h}^{(1)}_{\o}(\v{k}';m) \Big)\hat{s}^{(1)}_{\o}(\bs{k}';m)
\]

and 

\[\begin{split}
&\de_{k_0}\hat{s}_{\o}(\bs{k}';m)- \de_{k_0}\hat{s}^{(1)}_{\o}(\bs{k}';m)= i\hat{s}_{\o}(\bs{k}';m)^2 -i \hat{s}^{(1)}_{\o}(\bs{k}';m)^2=\\
&i \Big( \hat{s}_{\o}(\bs{k}';m)- \hat{s}^{(1)}_{\o}(\bs{k}';m) \Big) \hat{s}_{\o}(\bs{k}';m)+ i\hat{s}^{(1)}_{\o}(\bs{k}';m) \Big( \hat{s}_{\o}(\bs{k}';m)- \hat{s}^{(1)}_{\o}(\bs{k}';m) \Big).
\end{split}\]

Recalling that $\hat{h}_{\o}$ and $\hat{h}^{(1)}_{\o}$ differ by $\mc{O}(|\v{k}'|^2+m^2)$ terms, we see that there exists a constant $C_2>0$ such that, in matrix norm, $\|\hat{s}_{\o}(\bs{k}';m)- \hat{s}^{(1)}_{\o}(\bs{k}';m)\|\le C_2$ and $\|\de_{k_0}\hat{s}_{\o}(\bs{k}';m)- \de_{k_0}\hat{s}^{(1)}_{\o}(\bs{k}';m)\|\le C_2 (|\bs{k}'|+|m|)^{-1}$. Moreover, using the analyticity of $\gamma_{\o;i}$ (which implies $\gamma_{\o;i}(\v{k}';m)- \gamma_{\o;i}(\v{0};0)= \mc{O}(|\v{k}'|+|m|)$) and the bounds for $\hat{s}_{\o}, \hat{s}^{(1)}_{\o}$ and their derivatives, mentioned in the previous part of the proof, from Eq. \eqref{eqn:A10} we readily find:

\[ \big| \de_{p_0}\mc{F}^{(\mathrm{sing})}_{ij;\o}(\bs{k}';0;\pm\e)- \de_{p_0}\mc{F}^{(1)}_{ij;\o}(\bs{k}';0;\pm\e)\big|\le \frac{C'_2}{|\bs{k}'|^2+\e^2}\]

with a suitable constant $C'_2>0$. Plugging this bound in the r.h.s. of eq. \eqref{eqn:A9}, we get:

\[\begin{split} 
\big|\lim_{p_0\rightarrow 0^-} I_2(p_0;\e)\big| &\le \frac{1}{2}\sum_{\o=1}^N \sum_{\a=\pm} \int_{\mc{D}_{\e,\o}} \frac{d\bs{k}'}{(2\pi)^3}  \frac{C'_2}{|\bs{k}'|^2+\e^2} \le \frac{C'_2}{2\pi} \sum_{\o=1}^N \int_0^{\sqrt{\e}/\min\{v_{\o;1},1\}} dk\frac{k^2}{k^2+\e^2}\\
&\le \frac{C'_2\sqrt{\e}}{2\pi}\sum_{\o=1}^N\frac{1}{\min\{v_{\o;1},1\}}\le \frac{C''_2 N\sqrt{\e}}{2\pi \min_{\o=1}^N v_{\o;1}/(1+v_{\o;1})},  \end{split}\]

with $v_{\o;1}>0$ the smallest eigenvalue of $|A_{\o}|$, introduced at the beginning of Subsection \ref{ssect:Dirac}.

\medskip
\textbf{Analysis of $I_3$.}
This term is slightly more delicate than the previous two, since the functions $\mf{d}_0\mc{F}^{(\mathrm{sing})}_{ij;\o}(\,\cdot\,;p_0;0)$ and $\mf{d}_0\mc{F}^{(1)}_{ij;\o}(\,\cdot\,;p_0;0)$ are not integrable for $p_0\rightarrow0^-$. Using the resolvent identity: $\hat{s}_{\o}(\bs{k}'+\bs{p}_0;m)- \hat{s}_{\o}(\bs{k}';m)= ip_0 \hat{s}_{\o}(\bs{k}'+\bs{p}_0;m)\hat{s}_{\o}(\bs{k}';m)$ and the analogue for $\hat{s}^{(1)}_{\o}$, we can rewrite:

\begin{equation}
\label{eqn:A11}
\begin{split}
&\mf{d}_0 \mc{F}^{(\mathrm{sing})}_{ij;\o}(\bs{k}';p_0;0)- \mf{d}_0 \mc{F}^{(1)}_{ij;\o}(\bs{k}';p_0;0)=\\
&- i\Tr\Big\{ \gamma_{\o;i}(\v{k}';0) \hat{s}_{\o}(\bs{k}';0) \gamma_{\o;j}(\v{k}';0) \hat{s}_{\o}(\bs{k}'+\bs{p}_0;0) \hat{s}_{\o}(\bs{k}';0) \Big\}+\\
&+i \Tr\Big\{ \gamma_{\o;i}(\v{0};0) \hat{s}^{(1)}_{\o}(\bs{k}';0) \gamma_{\o;j}(\v{0};0) \hat{s}^{(1)}_{\o}(\bs{k}'+\bs{p}_0;0) \hat{s}^{(1)}_{\o}(\bs{k}';0) \Big\}.
\end{split}
\end{equation}

Now we proceed similarly to the r.h.s. of Eq. \eqref{eqn:A10}, namely we expand the r.h.s. of Eq. \eqref{eqn:A11} in terms of the difference of propagators, $\hat{s}_{\o}(\bs{k}';0)- \hat{s}^{(1)}_{\o}(\bs{k}';0)$ and $\hat{s}_{\o}(\bs{k}'+\bs{p}_0;0)- \hat{s}^{(1)}_{\o}(\bs{k}'+\bs{p}_0;0)$, and of vertex functions, $\gamma_{\o;i}(\v{k}';0)- \gamma_{\o;i}(\v{0};0)$. Using the aforementioned bounds for $\hat{s}_{\o},\hat{s}^{(1)}_{\o}$ and $\gamma_{\o;i}(\v{k}';0)$, and the fact that 

\[\big\|\hat{s}_{\o}(\bs{k}';0)- \hat{s}^{(1)}_{\o}(\bs{k}';0)\big\|\le \|\hat{s}_{\o}(\bs{k}';0)\|\big\|\hat{h}_{\o}(\v{k}';0)- \hat{h}^{(1)}_{\o}(\v{k}';0) \big\|\hat{s}^{(1)}_{\o}(\bs{k}';0)\| \le \frac{C_3|\v{k}'|^2}{|\bs{k}'|^2},\] 

for a suitable constant $C_3>0$, we find:

\begin{equation}
\label{eqn:A12}
\begin{split}
&\big|\mf{d}_0 \mc{F}^{(\mathrm{sing})}_{ij;\o}(\bs{k}';p_0;0)- \mf{d}_0 \mc{F}^{(1)}_{ij;\o}(\bs{k}';p_0;0)\big|\le \\
& C_3 \left(\frac{|\v{k}'|}{|\bs{k}'|^2|\bs{k}'+\bs{p}_0|}+ \frac{|\v{k}'|^2}{|\bs{k}'|^3|\bs{k}'+\bs{p}_0|}+   \frac{|\v{k}'|^2}{|\bs{k}'|^2|\bs{k'}+\bs{p}_0|^2} \right)\le \frac{3C_3}{|\bs{k}'|^2},
\end{split}
\end{equation}

where we used the fact that $|\v{k}'|$ is smaller than $|\bs{k}'|$ and $|\bs{k}'+\bs{p}_0|$ at the same time. The bound in Eq. \eqref{eqn:A12} yields a bound on $I_3$:

\begin{equation}
\begin{split}
|I_3(p_0;\e)|\le \sum_{\o=1}^N \int_{\mc{D}_{\e,\o}} \frac{d\bs{k}'}{(2\pi)^3} \frac{3C_3}{|\bs{k}'|^2}\le \frac{3C_3}{2\pi} \sum_{\o=1}^N \int_0^{\sqrt{\e}/\min\{v_{\o;1},1\}} dk \le \frac{3C_3 N \sqrt{\e}}{2\pi \min_{\o=1}^N v_{\o;1}/(1+v_{\o;1})}.
\end{split}
\end{equation}

Moreover, the existence of the limit $\lim_{p_0\rightarrow 0^-} I_3(p_0;\e)$ follows by the Dominated Convergence Theorem, by combining the bound in Eq. \eqref{eqn:A12} with the fact that the limit:  

\[ \lim_{p_0\to 0^-} \mf{d}_0 \Big(\mc{F}^{(\mathrm{sing})}_{ij;\o}(\bs{k}';p_0;0)- \mc{F}^{(1)}_{ij;\o}(\bs{k}';p_0;0)\Big)= \de_{p_0}\mc{F}^{(\mathrm{sing})}_{ij;\o}(\bs{k}';0;0)- \de_{p_0}\mc{F}^{(1)}_{ij;\o}(\bs{k}';0;0) \]

exists finite for any $\bs{k}'\in\mc{D}_{\e,\o}\setminus\{\bs{0}\}$.

\medskip

All in all we have shown that, for every $0<\e\le \e_0$, the limit $p_0\to0^-$ of Eq. \eqref{eqn:A6} exists finite and, collecting the bounds above for each of the three terms in the r.h.s. of Eq. \eqref{eqn:A7}, we find that

\begin{equation}
\begin{split}
&\Big|\lim_{p_0\to 0^-}\Big(\mf{D}\s^{(\mathrm{sing})}_{ij}(p_0;\e) - \sum_{\o=1}^N \mf{D}\s^{(1)}_{ij;\o}(p_0;\e) \Big)\Big|\le \\
& \frac{N\sqrt{\e}}{2\pi} \left( C_1' \max_{\o=1}^N (v_{\o;2}+1) + \frac{C_2'+3C_3}{\min_{\o=1}^Nv_{\o;1}/(1+ v_{\o;1})} \right).
\end{split}
\end{equation}
\end{proof}

\subsection{The Dirac frame}
\label{app:1_Dirac}

Here we discuss two main aspects related to the Dirac model. First we discuss the proof that, via a spatial rotation and a unitary conjugation, it is always possible to map a massless relativistic $2\times2$ Hamiltonian into the Dirac Hamiltonian, namely $\hat{h}^{(1)}(\v{k}';0)$ into $\hat{h}^{(D)}(\v{k}')$. Next we prove a number of remarkable symmetries of the current-current correlator $\hat{\mc{K}}$ related to the Dirac Hamiltonian. Since throughout this Subsection we will always deal with $2\times2$ matrices, we will denote $\Tr_{\mbb{C}^2}$ simply by $\Tr$.

\begin{proof}[Proof of Lemma \ref{lemma:Dirac}.]

We closely follow \cite[Sect. 8.3 - 8.4]{Cances}. By Eq. \eqref{eqn:h1} and Eq. \eqref{eqn:bands1}, we have that the energy bands of $\hat{h}_{\o}(R_{\o}\v{k}';0)$ are given by:

\begin{equation}
\label{eqn:A13} 
\begin{split}
\l_{\pm}(\v{k}_{F,\o}+ R_{\o}\v{k}';0)- \mu&= \v{c}_{\o}\cdot(R_{\o}\v{k}') \pm |A_{\o}R_{\o}\v{k}'\big| + \mc{O}(|\v{k}'|^2)\\
&  =\v{c'}_{\o}\cdot\v{k}' \pm \sqrt{v_{\o;1}^2k_1'^2+ v_{\o;2}^2k_2'^2} + \mc{O}(|\v{k}'|^2).
\end{split}\end{equation}

Moreover, since $\hat{h}_{\o}(R_{\o}\v{k}';0)$ is a $2\times2$ Hermitian matrix, it admits the decomposition:

\begin{equation}
\label{eqn:A14}
\hat{h}_{\o}(R_{\o}\v{k}';0)= \hat{d}_{\o;0}(\v{k}')\mathds{1}_2+ \sum_{s=1}^3 \hat{d}_{\o;s}(\v{k}')\s_s,
\end{equation}

where $\hat{d}_{\o;0}(\v{k}')= \frac{1}{2}\Tr\big(\hat{h}_{\o}(R_{\o}\v{k}';0) \big)$ and $\hat{d}_{\o;s}(\v{k}')= \frac{1}{2}\Tr\big(\s_s\hat{h}_{\o}(R_{\o}\v{k}';0) \big)$ are real-analytic by Lemma \ref{lemma_KN}. From Eq. \eqref{eqn:A14} it follows that

\begin{equation}
\label{eqn:A15}
\begin{split}
&\l_{\pm}(\v{k}_{F,\o}+ R_{\o}\v{k}';0)- \mu= \hat{d}_{\o;0}(\v{k}') \pm \sqrt{\sum_{s=1}^3 \hat{d}_{\o;s}(\v{k}')^2}.
\end{split}
\end{equation}

Writing, for $s=0,...,3$, $\hat{d}_{\o;s}(\v{k}')= \hat{d}^{(1)}_{\o;s}(\v{k}')+ \mc{O}(|\v{k}'|^2)$, with: 

\begin{equation}
\label{eqn:A17}
\hat{d}^{(1)}_{\o;s}(\v{k}')= \v{l}_{\o;s}\cdot\v{k}'
\end{equation}

(recall that $\hat{h}_{\o}(\v{0};0)=0$ and so $\hat{d}_{\o;s}(\v{0})=0$) for suitable $\v{l}_{\o;s}\in\mbb{R}^2$, and plugging Eq. \eqref{eqn:A13} into Eq. \eqref{eqn:A15}, we find that 

\begin{equation}
\label{eqn:A16}
\hat{d}^{(1)}_{\o;0}(\v{k}')= \v{c'}_{\o}\cdot\v{k}', \qquad \sum_{s=1}^3 \hat{d}^{(1)}_{\o;s}(\v{k}')^2= v_{\o;1}^2k_1'^2+ v_{\o;2}^2k_2'^2.
\end{equation}

Plugging Eq. \eqref{eqn:A17} in the second relation above, we find that

\begin{equation*}
\sum_{s=1}^3 (\v{l}_{\o;s}\cdot\v{k}')^2 = v_{\o;1}^2k_1'^2+ v_{\o;2}^2k_2'^2,
\end{equation*}

from which:

\begin{equation}
\label{eqn:A18}
\sum_{s=1}^3 l_{\o;s;i}l_{\o;s;j}= v_{\o;i}^2 \d_{i,j}
\end{equation}

for every $i,j\in\{1,2\}$. From eq. \eqref{eqn:A18} we see that the three-component vectors:

\[ f_{\o;1}:= \frac{1}{v_{\o;2}}\left(\begin{array}{c}
    l_{\o;1;2}\\
    l_{\o;2;2}\\
    l_{\o;3;2}
\end{array} \right), \qquad f_{\o;2}:=\frac{1}{v_{\o;1}}\left(\begin{array}{c}
    l_{\o;1;1}\\
    l_{\o;2;1}\\
    l_{\o;3;1}
\end{array} \right), \qquad f_{\o;3}:= f_{\o;1}\times f_{\o;2}\]
form an orthonormal basis of $\mbb{R}^3$, and the $3\times3$ matrix $(Q_{\o})_{s,s'}:= (f_{\o;s})_{s'}$ is in $\mbb{SO}(3)$. Now we would like to implement a transformation which maps the relativistic Hamiltonian:

\begin{equation}
\hat{h}^{(1)}_{\o}(R_{\o}\v{k}';0) \equiv \hat{d}^{(1)}_{\o;0}(\v{k}') \mathds{1}_2 + \sum_{s=1}^3 \hat{d}^{(1)}_{\o;s}(\v{k}') \s_s,
\end{equation}

into the r.h.s.\ of Eq. \eqref{eqn:22}, which (apart from the term proportional to $\mathds{1}_2$) is basically the Dirac Hamiltonian. Note that in view of Eq. \eqref{eqn:A18}, the following is true:

\begin{align}
&\hat{d}^{(D)}_{\o;1}(\v{k}'):= \sum_{s=1}^3 (Q_{\o})_{1s} \hat{d}^{(1)}_{\o;s}(\v{k}')= v_{\o;2} k'_2,\\
&\hat{d}^{(D)}_{\o;2}(\v{k}'):= \sum_{s=1}^3 (Q_{\o})_{2s} \hat{d}^{(1)}_{\o;s}(\v{k}')= v_{\o;1} k'_1,\\
&\hat{d}^{(D)}_{\o;3}(\v{k}'):= \sum_{s=1}^3 (Q_{\o})_{3s} \hat{d}^{(1)}_{\o;s}(\v{k}')=0.
\end{align}

As a final step we implement a unitary transformation such that the component of $\hat{h}^{(1)}_{\o}$ along $\s_s$ is mapped into $\hat{d}^{(D)}_{\o;s}(\v{k}';m)$. We use the remarkable property \cite{Cornwell} that for every $Q\in\mbb{SO}(3)$ there exists $W\in\mbb{SU}(2)$ such that 

\[\tfrac{1}{2}\Tr\big\{\s_s W^{\dg} \s_{s'} W\big\}= Q_{s,s'}, \qquad \forall s,s'\in\{1,2,3\}.\]

Letting $W_{\o}$ be the unitary which satisfies the above condition with $Q=Q_{\o}$, we have that

\begin{equation*}
\begin{split}
&W_{\o}^{\dg} \hat{h}^{(1)}_{\o}(R_{\o}\v{k}';0) W_{\o}=\\
& \frac{1}{2}\mathds{1}_2 \Tr\Big\{W_{\o}^{\dg} \hat{h}^{(1)}_{\o}(R_{\o}\v{k}';0) W_{\o} \Big\}+ \frac{1}{2}\sum_{s=1}^3 \s_s \Tr\Big\{\s_s W_{\o}^{\dg} \hat{h}^{(1)}_{\o}(R_{\o}\v{k}';0) W_{\o} \Big\}=\\
& \hat{d}^{(1)}_{\o;0}(\v{k}') \mathds{1}_2 + \sum_{s,s'=1}^3 \s_s \hat{d}^{(1)}_{\o;s'}(\v{k}') \frac{1}{2} \Tr\big\{ \s_sW_{\o}^{\dg} \s_{s'} W_{\o} \big\}=\\
& \hat{d}^{(1)}_{\o;0}(\v{k}') \mathds{1}_2 + \sum_{s,s'=1}^3 \s_s \hat{d}^{(1)}_{\o;s'}(\v{k}') (Q_{\o})_{s,s'} = \hat{d}^{(1)}_{\o;0}(\v{k}') \mathds{1}_2 + \sum_{s=1}^3 \hat{d}^{(D)}_{\o;s}(\v{k}')  \s_s =\\
&\v{c'}_{\o}\cdot\v{k}'\mathds{1}_2+ v_{\o;2}k'_2\s_1+ v_{\o;1}k'_1\s_2,
\end{split}
\end{equation*}

as desired.
\end{proof}

\begin{proof}[Proof of Lemma \ref{lemma:kcal}.]
We begin by noting that
\begin{equation}
\label{eqn:symkcal}
\hat{\mc{K}}_{\t\nu}(p_0)= \hat{\mc{K}}_{\nu\t}(p_0)= \hat{\mc{K}}_{\nu\t}(-p_0),
\end{equation}
which follows by using the cyclicity of the trace, by the definition of the domain $\mc{C}_{\e}$, and by the fact that the inverse relativistic propagator $\hat{s}^{(D)}(\bs{k}')^{-1}$ is linear. This allows us to restrict to the case $\t\ge \nu$ for proving the claim.

Let us begin by discussing the vanishing of $\hat{\mc{K}}_{10}$ and $\hat{\mc{K}}_{20}$. By the invariance of the trace under unitary conjugation, choosing the unitary $W= e^{-i\frac{\pi}{4}\s_3}$, such that $W^{\dg}\s_1 W= - \s_2$ and $W^{\dg}\s_2 W=\s_1$, we obtain that (hereafter, in order to simplify the notation, the two-point functions $\hat{s}^{(D)}$ are understood to be evaluated in $\v{v}=\v{1}$ and $\v{c}=\v{0}$)
\begin{align*}
\hat{\mc{K}}_{10}(p_0)=&-\int_{\mc{C}_{\e}} \frac{d\bs{k}'}{(2\pi)^3}\Tr \Big\{\s_2 \,\hat{s}^{(D)}(\bs{k}') \,\hat{s}^{(D)}\big(\bs{k}'+\bs{p}_0\big) \Big\} \\
&= - \int_{\mc{C}_{\e}} \frac{d\bs{k}'}{(2\pi)^3}\Tr \Big\{\s_1 \,\hat{s}^{(D)}(k_0, -k'_2,k'_1) \,\hat{s}^{(D)}\big(k_0+p_0,-k'_2,k'_1\big) \Big\}\\
&= - \int_{\mc{C}_{\e}} \frac{d\bs{k}'}{(2\pi)^3}\Tr \Big\{\s_1 \,\hat{s}^{(D)}(\bs{k}') \,\hat{s}^{(D)}\big(\bs{k}'+\bs{p}_0\big) \Big\}= \hat{\mc{K}}_{20}(p_0),
\end{align*}

where we also used the invariance of $\mc{B}_{\e}$ under the \emph{discrete rotation} $k'_1\to k'_2,k'_2\to -k'_1$. Similarly, we find that

\begin{align*}
\hat{\mc{K}}_{20}(p_0)=&-\int_{\mc{C}_{\e}} \frac{d\bs{k}'}{(2\pi)^3}\Tr \Big\{\s_1 \,\hat{s}^{(D)}(\bs{k}') \,\hat{s}^{(D)}\big(\bs{k}'+\bs{p}_0\big) \Big\} \\
&=  \int_{\mc{C}_{\e}} \frac{d\bs{k}'}{(2\pi)^3}\Tr \Big\{\s_2 \,\hat{s}^{(D)}(k_0, -k'_2,k'_1) \,\hat{s}^{(D)}\big(k_0+p_0,-k'_2,k'_1\big) \Big\}\\
&=  \int_{\mc{C}_{\e}} \frac{d\bs{k}'}{(2\pi)^3}\Tr \Big\{\s_2 \,\hat{s}^{(D)}(\bs{k}') \,\hat{s}^{(D)}\big(\bs{k}'+\bs{p}_0\big) \Big\}= -\hat{\mc{K}}_{10}(p_0).
\end{align*}

Hence $\hat{\mc{K}}_{10}= \hat{\mc{K}}_{20}=0$. The same argument shows that $\hat{\mc{K}}_{11}= \hat{\mc{K}}_{22}$ and that $\hat{\mc{K}}_{12}=-\hat{\mc{K}}_{21}$, so that in view of Eq. \eqref{eqn:symkcal} we must have $\hat{\mc{K}}_{12}=\hat{\mc{K}}_{21}=0$.

Let us now focus on $\hat{\mc{K}}_{00}$. We are going to show that, for every $p_0\in\mbb{R}$, $|\v{k}'|>0$, 

\begin{equation}
\label{eqn:A19}
\int_{\mbb{R}}dk_0\Tr \Big\{ \hat{s}^{(D)}(\bs{k}')\hat{s}^{(D)}\big(\bs{k}'+\bs{p}_0\big) \Big\}=0.  
\end{equation}

Notice that the l.h.s. is apriori continuous in $p_0$ by the Dominated Convergence theorem, thus it suffices to show that it equals zero for every $p_0\neq 0$ and $|\v{k}'|>0$. By inspection, we have that
\[
\int_{\mbb{R}}dk_0\Tr \Big\{ \hat{s}^{(D)}(\bs{k}')\hat{s}^{(D)}\big(\bs{k}'+\bs{p}_0\big) \Big\} = \int_{\mbb{R}}dk_0 \,\underbrace{\frac{2\big(|\v{k}'|^2-k_0(k_0+p_0)\big)}{(k_0^2+|\v{k}'|^2)((k_0+p_0)^2+|\v{k}'|^2)}}_{=:f(k_0,\v{k}',p_0)}.
\]
Observe that as a function of $k_0$ the map $f(k_0,\v{k}',p_0)\sim\frac{1}{k_0^2}$ for large $k_0$, thus applying Cauchy's Residue Theorem in the upper half complex plane, we obtain that
\begin{align*}
\frac{1}{2\pi i}\int_{\mbb{R}} dk_0\, f(k_0,\v{k}',p_0)&=\mathrm{Res}\,f(\,\cdot\,,\v{k}',p_0)\big|_{i|\v{k}'|} + \mathrm{Res}\,f(\,\cdot\,,\v{k}',k_0)\big|_{i|\v{k}'|-p_0}\\
&=-i\left( \frac{2|\v{k}'|-ip_0}{p_0(p_0+2i|\v{k}'|)}+\frac{2|\v{k}'|+ip_0}{p_0(p_0-2i|\v{k}'|)}    \right)=0.   
\end{align*}
This proves Eq. \eqref{eqn:A19} and therefore we get that $\hat{\mc{K}}_{00}(p_0)=0$ for every $p_0\in\mbb{R}$.
\end{proof}

\printbibliography
\end{document}